\documentclass{theoretics}

\title{Succinct Oblivious Tensor Evaluation and Applications: Adaptively-Secure Laconic Function Evaluation and Trapdoor Hashing for All Circuits}
\ThCSshorttitle{Succinct Oblivious Tensor Evaluation and Applications}

\ThCSauthor[1]{Damiano Abram}{abram.damiano@protonmail.com}[0009-0004-3916-7550]
\ThCSauthor[2]{Giulio Malavolta}{giulio.malavolta@hotmail.it}[0009-0009-9737-0094]
\ThCSauthor[3]{Lawrence Roy}{ldr709@gmail.com}[0009-0003-8436-0029]
\ThCSaffil[1]{University of Edinburgh, UK}
\ThCSaffil[2]{Bocconi University, Italy}
\ThCSaffil[3]{IBM Research Zürich, Switzerland}

 \ThCSshortnames{D.~Abram, G.~Malavolta, L.~Roy}
 \ThCSthanks{Invited from STOC 2025 \cite{AbramMR25}.\\
 Damiano Abram and Giulio Malavolta are supported by the European Research Council through an ERC Starting Grant (Grant agreement No.~101077455, ObfusQation). Giulio Malavolta is also funded by the Deutsche Forschungsgemeinschaft (DFG, German Research Foundation) under Germany's Excellence Strategy - EXC 2092 CASA – 390781972. This work was supported by Input Output (iog.io) through their funding of the Edinburgh Blockchain Technology Lab.
Lawrence Roy is supported by the European Research Council (ERC) under the European Union's Horizon 2020 research and innovation programme under grant agreement number 101124977 (DECRYPSIS) and the Danish Independent Research Council under Grant-ID DFF-0165-00107B (C3PO).%
}

\ThCSkeywords{Lattice-based cryptography, homomorphic computation, laconic cryptography}

\ThCSyear{2026}
\ThCSarticlenum{14}
\ThCSreceived{Aug 14, 2025}
\ThCSrevised{Apr 20, 2026} 
\ThCSaccepted{May 26, 2026}
\ThCSpublished{1 Oct, 2026}
\ThCSdoicreatedtrue

\usepackage{xspace}
\usepackage{xparse}
\usepackage{xargs}
\usepackage{aliascnt}
\usepackage{mathtools}
\usepackage{tikz}
\usetikzlibrary{shapes.geometric}
\usetikzlibrary{positioning,shapes.multipart,calc,arrows.meta}
\usepackage{tcolorbox}
\usepackage{bm}
\usepackage[capitalise,noabbrev,nameinlink]{cleveref}

\newcommand{\ek}{\mathsf{ek}}

\newcommand{\td}{\mathsf{td}}

\newcommand{\LFE}{\mathsf{LFE}}
\newcommand{\MOLE}{\mathsf{MOLE}}

\newcommand{\getsr}{\overset{\raisebox{-1pt}{$\scriptscriptstyle \$$}}{\gets}}
\newcommand{\sample}{\getsr}

\newcommand{\Pre}{\mathsf{Pre}}
\newcommand{\Post}{\mathsf{Post}}
\newcommand{\BOTE}{\mathsf{hOTE}}

\newcommand{\OTE}{\mathsf{OTE}}
\newcommand{\GSW}{\mathsf{GSW}}

\newcommand{\N}{\mathbb{N}}
\newcommand{\Z}{\mathbb{Z}}

\newcommand{\poly}{\mathsf{poly}}
\newcommand{\polylog}{\mathsf{polylog}}
\newcommand{\negl}{\mathsf{negl}}

\renewcommand{\sec}{\ensuremath{\lambda}}
\newcommand{\secpar}{\sec}

\newcommand{\ct}{\mathsf{ct}}

\newcommand{\aux}{\mathsf{aux}}

\newcommand{\Sim}{\mathsf{Sim}}
\newcommand{\Adv}{\mathcal{A}}

\newcommand{\RTDH}{\mathsf{RTDH}}
\newcommand{\BV}{\mathsf{GSW}}
\newcommand{\Compress}{\mathsf{Compress}}
\newcommand{\ck}{\mathsf{ck}}

\newcommand{\algofont}[1]{\ensuremath{\mathsf{#1}}}

\newcommand{\tpk}{{\algofont{tpk}}}
\newcommand{\Enc}{{\algofont{Enc}}}
\newcommand{\Dec}{{\algofont{Dec}}}

\newcommand{\Setup}{{\algofont{Setup}}}
\newcommand{\Gen}{{\algofont{Gen}}}

\newcommand{\Eval}{\algofont{Eval}\xspace}

\newcommand{\Hash}{{\algofont{Hash}}}

\newcommand{\Expand}{\algofont{Expand}}

\newcommand{\dist}{\ensuremath{\mathcal{D}}}

\newcommand{\id}{\ensuremath{\mathbf{I}}}
\newcommand{\mpk}{\mathsf{mpk}}

\newcommand{\pk}{{\mathsf{pk}}}
\newcommand{\sk}{{\mathsf{sk}}}

\newcommand{\hk}{\ensuremath{\mathsf{hk}}}

\newcommand{\seclab}[1]{\label{sec:#1}}
\newcommand{\secref}[1]{Section~\ref{sec:#1}}

\newcommand{\thmlab}[1]{\label{thm:#1}}
\newcommand{\thmref}[1]{Theorem~\ref{thm:#1}}

\newcommand{\defref}[1]{Definition~\ref{def:#1}}

\ifdefined\isllncs\else
\newcommand\bbbone{1}
\fi

\SetAlgorithmName{Construction}{Construction}{}
\makeatletter
\patchcmd{\@algocf@start}{\rule{\textwidth}{3pt}}
  {\rule{\textwidth}{3pt}\par}{}{}
\pretocmd{\@algocf@start}
  {\setlength{\algowidth}{\textwidth}}{}{}
\apptocmd{\@algocf@finish}{\par}{}{}
\makeatother

\newenvironment{protocol}[1][]
{\begin{algorithm}\caption{#1}}
  {\end{algorithm}}

\begin{document}

\maketitle
\begin{abstract}
    We propose the notion of succinct oblivious tensor evaluation (OTE), where two parties compute an additive secret sharing of a tensor product of two vectors $\mathbf{x} \otimes \mathbf{y}$, exchanging two simultaneous messages. Crucially, the size of both messages and of the CRS is independent of the dimension of $\mathbf{x}$.
    We present a construction of OTE with optimal complexity from the standard learning with errors (LWE) problem. Then we show how this new technical tool enables a host of cryptographic primitives, all with security reducible to LWE, such as:
    \begin{itemize}
        \item Adaptively secure laconic function evaluation for depth-$D$ functions $f:\{0, 1\}^m\rightarrow\{0, 1\}^\ell$ with communication $m+\ell+ D\cdot\poly(\sec)$.
        \item A trapdoor hash function for all functions.
        \item An (optimally) succinct homomorphic secret sharing for all functions.
        \item A rate-$1/2$ laconic oblivious transfer for batch messages, which is best possible.
    \end{itemize}
    In particular, we obtain the first laconic function evaluation scheme that is adaptively secure from the standard LWE assumption, improving upon Quach, Wee, and Wichs (FOCS 2018). 
    As a key technical ingredient, we introduce a new notion of \emph{adaptive lattice encodings}, which may be of independent interest.
\end{abstract}

\section{Introduction}

Consider the scenario where Alice holds a \emph{long} vector $\mathbf{x}$, Bob holds a \emph{smaller} secret vector $\mathbf{y}$ and, after a single round of simultaneous messages, they should be able to locally compute an additive secret share of the tensor product $\mathbf{x} \otimes \mathbf{y}$ while preserving the privacy of $\mathbf{y}$. That is, after one simultaneous round of messages, Alice computes $\alpha$ and Bob computes $\beta$ such that:
\[
\alpha + \beta = \mathbf{x}\otimes \mathbf{y}.
\]
We refer to this problem as \emph{non-interactive oblivious tensor evaluation} (NI-OTE). In this work, we are interested in the communication complexity of secure NI-OTE, i.e., the minimum size of the messages needed in order to compute a correct additive secret sharing, while preserving the privacy of $\mathbf{y}$. While one may intuitively expect that Alice and Bob's messages should be long enough to fully specify both the vectors, this is in fact not so. Counterintuitively, we show that it is possible to complete the above protocol with communication complexity \emph{logarithmic} in the dimensions of the input $\mathbf{x}$.

The objective of this work is to construct explicit protocols for NI-OTE with minimal communication, and to explore the cryptographic consequences of this primitive.

\subsection{Our Results}

Our main technical contribution is a protocol for NI-OTE with minimal communication complexity, where the security is proven against the standard learning with errors (LWE) assumption \cite{STOC:Regev05}.
We prove this result in two steps: First, we construct an elementary (half-succinct) NI-OTE protocol where only the message of one party is short, whereas the message of the other party can depend arbitrarily on $|\mathbf{x}|$. Then we show a generic \emph{bootstrapping} procedure that makes the scheme fully succinct, i.e., the messages of both parties are short.
Overall, our main result is captured by the following informal theorem statement (treating the security parameter as constant).
\begin{theorem}[Informal]
    If the LWE problem is hard, then there exists an NI-OTE protocol for $\mathbf{x}\in\Z_q^m$ and $\mathbf{y}\in\Z_q^\ell$ with communication complexity $\ell\cdot \poly(\sec) + \poly(\sec, \log m)$ and CRS of size $\poly(\sec, \log m)$.
\end{theorem}
{Notice that we consider the problem in its strongest form, where Alice and Bob send each other a single, simultaneous message. However, outside the realm of obfuscation, it was unknown how to construct OTE with the same parameters, even if interaction is allowed.}
Besides being a primitive of independent interest, we show that the existence of our succinct NI-OTE protocol has surprising applications in cryptography. 

\paragraph{Application I: Trapdoor Hash Functions.}
For starters, we show how succinct NI-OTE, combined with recent results in laconic function evaluation \cite{FOCS:QuaWeeWic18,FOCS:HsiLinLuo23,C:DHMWW24}, enables a construction of a trapdoor hash (TDH) function \cite{C:DGIMMO19} for all functions, or even RAM programs, from Ring-LWE. 

{In a TDH, Alice holds a function $f$ and Bob holds an input $x$. They simultaneously exchange a message depending only on their own input and, in a second stage, they can locally compute an additive secret sharing of $f(x)$. Importantly, one wants Alice's message to be short. In our TDH, the} size of the hash is constant, and the size of the encoding key depends only on the {Turing machine} description of the function $f$, which is optimal. {To the best of our knowledge, prior constructions of TDHs \cite{C:DGIMMO19,C:RoySin21} only supported linear functions.}
This result is summarized by the following (informal) theorem statement.
\begin{theorem}[Informal]
    If the LWE problem is hard, then there exists a TDH for functions with depth $D$, where the size of the encoding is bounded by {$|f|\cdot \poly(\sec, D)$}. Additionally, assuming the hardness of \emph{circular} LWE, we obtain a bound on the size of the encodings of $|f|\cdot \poly(\sec)$.
\end{theorem}

\paragraph{Application II: Succint Homomorphic Secret Sharing and More.}
As a direct consequence of the above result, we obtain a new protocol homomorphic secret sharing (HSS) scheme \cite{C:BoyGilIsh16}. {In this work we consider the (stronger) public-key variant of the primitive: Suppose that Alice and Bob respectively hold inputs $\mathbf{x}$ and $\mathbf{y}$ and they want to evaluate the function $f$. Alice and Bob send each other (simultaneously) a single message, dependent on their input. Based on some local computation, Alice and Bob obtain an additive secret sharing of $f(\mathbf{x}, \mathbf{y})$.}

{In these settings, we obtain a \emph{succinct} \cite{EC:AbrRoySch24} protocol for all functions, where the communication complexity is logarithmic in Alice's input} $\mathbf{x}$, which is optimal. The only prior work on succinct HSS \cite{EC:AbrRoySch24} only supported $\mathsf{NC}_1$ circuits and had communication complexity proportional to $\lvert \mathbf{x}\rvert^\varepsilon$, for some $\varepsilon\in O(1)$. In addition, we obtain a new \emph{batched} laconic oblivious transfer protocol \cite{C:CDGGMP17}, with constant-size receiver's message and with rate $1/2$, which is best possible. 

We refer the reader to Section \ref{sec:apps} for a more detailed and precise discussion on these primitives, along with additional applications of succinct NI-OTE such as spooky encryption, and pseudorandom correlation generators.

\paragraph{Application III: Laconic Function Evaluation.}
Finally, we show how to leverage succinct NI-OTE to construct a new laconic function evaluation (LFE) \cite{FOCS:QuaWeeWic18} protocol. {An LFE is a two-party protocol where Alice holds a function $f$ and Bob holds an input $x$. Alice sends a \emph{hash} of her function $f$ to Bob, who computes an encoding of his input, depending on the hash sent by Alice. From Bob's message, Alice can recover $f(x)$ but nothing more. We propose the first LFE protocol that is} simultaneously:
\begin{itemize}
    \item \emph{Adaptively secure}: The attacker can choose the input adaptively, possibly depending on the public parameters.
    \item \emph{Rate-1}: The size of the encoding equals the size of the input, plus the size of the output, plus an additive factor.
\end{itemize}
Prior to our work, even constructing LFE with either of the two properties was considered an open problem. 
The question of adaptively-secure LFE from LWE was raised in \cite{FOCS:QuaWeeWic18}, where they proposed a construction provable against the \emph{adaptive} LWE assumption, whereas our construction is adaptively secure against the \emph{standard} LWE assumption. {This is not a cosmetic improvement: We also show a counterexample against the adaptive LWE assumption. This translates into an adaptive attack against \cite{FOCS:QuaWeeWic18}, underscoring the need for constructions proven adaptively secure against standard assumptions.}

The question of rate-1 LFE was considered in \cite{C:Wee24}, where they proposed a construction from $\ell$-succinct LWE, a recently-introduced variant of the LWE assumption. We improve upon this work by relying only on the standard LWE problem. Overall, our results can be summarized as follows:
\begin{theorem}[Informal]
     If the LWE problem is hard, there exists an adaptively secure LFE for depth-$D$ functions $f:\{0, 1\}^m\rightarrow\{0, 1\}^\ell$ with communication $m+\ell+ D\cdot \poly(\sec)$.
\end{theorem}
{We shall mention here that, allowing additional interaction beyond two rounds, such a result can also be achieved using rate-1 fully-homomorphic encryption \cite{TCC:BDGM19}, so the difficulty of the problem is to construct such primitive in two rounds, which matches the standard LFE syntax. On the other hand, for all other applications, the results are new even if we allow arbitrary interaction.}

As a key technical ingredient, we present a new variant of \emph{homomorphic lattice encodings} \cite{EC:BGGHNS14} that naturally supports adaptive security. This is the first construction of homomorphic lattice encodings that departs from the framework of \cite{EC:BGGHNS14}, and we expect it to find other applications in the future.

A diagram summarizing our results is given in Figure~\ref{fig:diagram}.

\begin{figure*}
    \centering
\begin{tikzpicture}[
  node distance=2cm and 3cm,
  primitive/.style={draw, rounded corners, align=center, font=\small},
  ref/.style={font=\scriptsize, midway, sloped, above},
  ref-nonsloped/.style={font=\scriptsize},
  scale = 1
]

\node[primitive] (half) at (-5,0) {Half-Succinct NI-OTE
};
\node[primitive] (niote) at (0,0) {Succinct NI-OTE
};
\node[primitive] (mole) at (0,2) {Succinct NI-MOLE
};
\node[primitive] (lenc) at (-5,-2) {Adaptive Lattice Encodings
};
\node[primitive] (comp) at (2,-2) {Compressed Encodings
};
\node[primitive] (LFE) at (2,4) {Laconic Function Evaluation\\ \cite{FOCS:QuaWeeWic18,FOCS:HsiLinLuo23,C:DHMWW24}};
\node[primitive] (RTDH) at (5,2) {Reverse TDH
};
\node[primitive] (TDH) at (10,3) {Regular TDH};
\node[primitive] (HSS) at (10,2) {Succinct HSS};
\node[primitive] (LOT) at (10,1) {Rate-$1/2$ Laconic OT};
\node[primitive] (wRTDH) at (5.5,0) {Input-Succinct Reverse TDH
};
\node[primitive] (mainLFE) at (9,-2) {Rate-1 Adaptively Secure LFE
};

\draw[->] (half) -- (niote);
\draw[->] (niote) -- (mole);
\draw[->] (mole) -- (RTDH);
\draw[->] (LFE) to [bend left=10] node{} (RTDH);
\draw[->] (RTDH) to [bend left=10] node{} (TDH);
\draw[->] (RTDH) -- (HSS);
\draw[->] (RTDH) to [bend right=10] node{} (LOT);
\draw[->] (niote) to [bend right=20] node{} (comp);
\draw[->] (lenc) -- (comp);
\draw[->] (comp) to [bend right=20] node{} (wRTDH);
\draw[->] (mole) to [bend left=20] node{} (wRTDH);
\draw[->] (wRTDH) to [bend left=20] node{} (mainLFE);

\end{tikzpicture}
\caption{A schematic representation of our results, and a summary of the implications to different cryptographic primitives.}\label{fig:diagram}
\end{figure*}
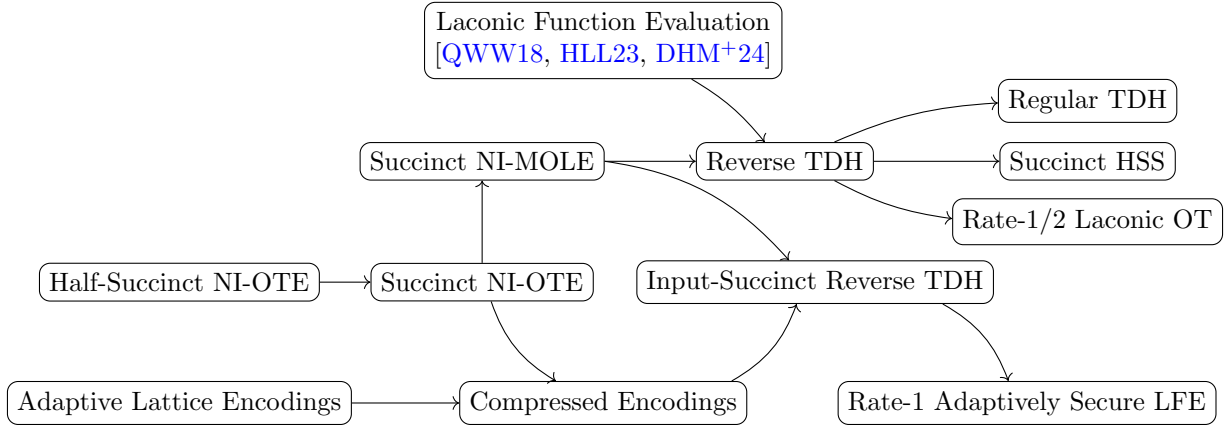

\subsection{Technical Outline}
From now on, we call Alice \emph{the hasher} and Bob \emph{the encoder}.
We start by presenting a half-succinct OTE, i.e., an OTE protocol where only the hasher's message is succinct in its input size. This is inspired by the work of \cite{EC:AbrRoySch24} 
and we extend their ideas in the context of tensor products. Suppose that we work over $\Z_q$, the hasher's input is a vector over $\Z_2^m$, while the encoder's input lies in $\Z_q^\ell$.

The construction relies on a setup that outputs a random matrix $\mathbf{A}\in\Z_q^{n\times m}$ where $m>n$. To hash $\mathbf{x}$, we simply compute an SIS-based hash $\mathbf{d}\gets \mathbf{A}\cdot \mathbf{x}$. To encode $\mathbf{y}$, on the other hand, we compute
\[\mathbf{C}\gets \mathbf{A}^\intercal \cdot \mathbf{S}+\mathbf{E}+\id_m\otimes \mathbf{y}^\intercal\]
where $\mathbf{S}\sample\Z_q^{n\times (m\cdot \ell)}$ and $\mathbf{E}\sample\chi(\bbbone^\sec)$. Above, we use $\id_m$ to denote the $m\times m$ identity matrix and $\chi$ to denote a low-norm distribution. Notice that, under LWE, $\mathbf{C}$ leaks no information about $\mathbf{y}$. Suppose that Alice and Bob exchanged $\mathbf{C}$ and $\mathbf{d}$. Alice can compute a ``noisy'' share of $\mathbf{x}\otimes \mathbf{y}$ by computing
\begin{align*}
\mathbf{v}:=\mathbf{C}^\intercal\cdot \mathbf{x}&=\mathbf{S}^\intercal\cdot \mathbf{A} \cdot \mathbf{x}+\mathbf{E}^\intercal\cdot\mathbf{x}+(\id_m\otimes \mathbf{y})\cdot\mathbf{x}\\
&=\mathbf{S}^\intercal\cdot\mathbf{A} \cdot \mathbf{x}+\mathbf{E}^\intercal\cdot\mathbf{x}+(\mathbf{x}\otimes \mathbf{y}).
\end{align*}
Bob can derive instead his own ``noisy'' share by computing
\[\mathbf{w}:=-\mathbf{S}^\intercal \cdot \mathbf{d}=-\mathbf{S}^\intercal\cdot\mathbf{A}\cdot \mathbf{x}.\]
It is easy to see that $\mathbf{v}+\mathbf{w}=\mathbf{x}\otimes \mathbf{y}+\mathbf{E}^\intercal\cdot\mathbf{x}$. Moreover, since $\mathbf{x}\in\Z_2^m$, the magnitude of $\mathbf{E}^\intercal\cdot\mathbf{x}$ is small. From this we can derive a fully correct, half succinct OTE over $\Z_p$ where $p$ is a sufficiently small divisor of $q$: Instead of hashing $\mathbf{x}$, hash its bit decomposition, instead of encoding $\mathbf{y}$, encode $q/p\cdot \mathbf{y}\otimes \mathbf{g}_q$ where $\mathbf{g}_q^\intercal$ is the gadget (row) vector $(1, 2, \dots, 2^{\log q})$. To reconstruct an exact secret-sharing of $\mathbf{x}\otimes \mathbf{y} \bmod p$, it is sufficient to apply a linear operation on the shares and round the result over $\Z_p$ following the ideas of \cite{C:DHRW16} (see \ref{def:alphacorrect}).

{\paragraph{An Attempt at Bootstrapping.}
We now describe a first attempt to obtain \emph{full} succinctness. In the half-succinct construction above, the encoder's message has size $m\ell$, whereas the hash size is independent of both $m$ and $\ell$. Thus the remaining bottleneck is Bob's message.
Following \cite{EC:AbrRoySch24}, we exploit two simple features of the protocol: It is non-interactive, and the functionality is linear. Suppose that Alice's input is much longer, say $\mathbf{x}\in\mathbb{Z}_q^M$ with $M \gg m,\ell$. Write
\[
\mathbf{x} = \begin{pmatrix}\mathbf{x}_0\\\vdots\\\mathbf{x}_{N-1}\end{pmatrix},
\qquad
\mathbf{x}_i \in \mathbb{Z}_q^m,
\]
where $N=M/m$. Alice hashes each block separately, obtaining digests
\[
\mathbf{d}_0,\ldots,\mathbf{d}_{N-1}.
\]
Bob, on the other hand, sends a \emph{single} encoding of $\mathbf{y}\in\mathbb{Z}_q^\ell$. Since the same encoding can be evaluated against every digest, the parties obtain an additive secret-sharing of $\mathbf{x}_i\otimes \mathbf{y}$ for every $i\in[N]$. Concatenating these shares yields an additive secret-sharing of $\mathbf{x}\otimes \mathbf{y}$.}

{This already makes Bob's communication sublinear in $M$: His message depends only on $m$ and $\ell$, not on the full length of $\mathbf{x}$. The drawback is that Alice must now send all $N$ digests, so her communication grows by a factor of $M/m$. In other words, we have traded a large encoding for many digests.
The next natural question is whether these $N$ digests can themselves be compressed. Towards this, let us take a closer look at Bob's shares in the noisy half-succinct scheme: For each block $i$, Bob's share can be obtained from
$\mathbf{d}_i \otimes \mathsf{vec}(\mathbf{S})$\footnote{$\mathsf{vec}(\mathbf{S})$ denotes the vectorisation of $\mathbf{S}$.}
by applying a public linear map.
Therefore, if Alice and Bob could obtain an additive secret-sharing of
\[
\mathbf{d}_i \otimes \mathsf{vec}(\mathbf{S})
\qquad\text{for every } i\in[N],
\]
then they could recover Bob's noisy shares for all blocks using only local linear operations. Combining these with Alice's local shares would give a noisy secret-sharing of $\mathbf{x}\otimes \mathbf{y}$, which could then be rounded exactly as above.}

{This suggests a recursive strategy. Let \[\mathbf{d}' := \begin{pmatrix}\mathbf{d}_0\\\vdots\\\mathbf{d}_{N-1})\end{pmatrix}\]
be the concatenation of all block digests, we apply our half-succinct OTE \emph{again}, this time to obtain a secret-sharing of
$\mathbf{d}' \otimes \mathsf{vec}(\mathbf{S})$.
If this worked, then Alice would no longer need to send the entire vector $\mathbf{d}'$, but rather its hash.
In fact, one could imagine iterating the same idea and, pushing this recursion to its limit, one might hope for a protocol in which Alice sends only a single short commitment to $\mathbf{x}$, while Bob sends one encoding per recursion level.}

{Unfortunately, this naive recursion fails because the second input grows too quickly. After the first step, Bob is no longer encoding $\mathbf{y}$, but rather $\mathsf{vec}(\mathbf{S})$. At the next level, he must encode the randomness used to encode $\mathsf{vec}(\mathbf{S})$, and so on. Since each recursion step increases the dimension of the randomness by a factor of about $n\cdot m$, Bob's messages grow by the same factor at every level. Thus the recursion indeed compresses Alice's side, but it causes Bob's side to blow up. To make the approach work, we need a way to keep the randomness small while preserving the linear structure that lets Bob's share be recovered from $\mathbf{d}_i \otimes \mathsf{vec}(\mathbf{S})$.}

\paragraph{Decreasing the size of $\mathbf{S}$.} Instead of sampling a random $\mathbf{S}$, we generate a pseudorandom one using LWE. Specifically, include $m\cdot \ell$ random $n\times n$ matrices $\mathbf{B}_0, \dots, \mathbf{B}_{\ell\cdot m-1}$ as part of the setup. Then, at encoding time, we sample a random vector $\mathbf{s}\sample\Z_q^n$ and we set the $i$-th row of $\mathbf{S}$ to be $\mathbf{B}_i\cdot \mathbf{s}+\mathbf{e}_i$ where $\mathbf{e}_i\sample\chi(\bbbone^\sec)$. In matrix notation, we obtain that $\mathbf{S}=\mathbf{B}\cdot (\id_{m\cdot \ell}\otimes \mathbf{s})+\hat{\mathbf{E}}$, i.e.,
\[\mathbf{S}=\underbrace{\begin{pmatrix}
\mathbf{B}_0 & \dots & \mathbf{B}_{\ell\cdot m-1}
    \end{pmatrix}}_{\mathbf{B}}\cdot \underbrace{\begin{pmatrix}
\mathbf{s} &  &\\
&\ddots &\\
&&\mathbf{s}
    \end{pmatrix}}_{m\cdot \ell \text{ times}}+\underbrace{\begin{pmatrix}
\mathbf{e}_0 & \dots & \mathbf{e}_{\ell\cdot m-1}
    \end{pmatrix}}_{\hat{\mathbf{E}}}.\]
The encoding of $\mathbf{y}$ becomes
\[\mathbf{C}=\mathbf{A}^\intercal\cdot \mathbf{B}\cdot (\id_{m\cdot \ell}\otimes \mathbf{s})+\mathbf{A}^\intercal \cdot \hat{\mathbf{E}}+\mathbf{E}+\id_m\otimes \mathbf{y}^\intercal.\]
We also introduce another modification: Instead of sampling the entries of $\mathbf{A}$ and $\mathbf{B}$ uniformly at random over $\Z_q$, we sample them uniformly over $\Z_2$. This trick ensures that the magnitude of $\mathbf{A}^\intercal \cdot \hat{\mathbf{E}}$ remains small, while, at the same time, it does not compromise security: LWE with respect to random \emph{binary} matrices is known to be as hard as standard LWE \cite{C:BLMR13}.
With these modifications to our half-succinct OTE, Alice's share becomes
\begin{align*}
\mathbf{v}:\!&=\mathbf{C}^\intercal\cdot \mathbf{x}\\
&=(\id_{m\cdot \ell}\otimes \mathbf{s}^\intercal)\cdot \mathbf{B}^\intercal\cdot \mathbf{A}\cdot \mathbf{x}+\hat{\mathbf{E}}^\intercal\cdot\mathbf{A}\cdot\mathbf{x}+\mathbf{E}^\intercal\cdot\mathbf{x}+(\id_m\otimes \mathbf{y})\cdot \mathbf{x}\\
&=(\id_{m\cdot \ell}\otimes \mathbf{s}^\intercal)\cdot \mathbf{B}^\intercal\cdot \mathbf{A}\cdot \mathbf{x}+\hat{\mathbf{E}}^\intercal\cdot\mathbf{A}\cdot\mathbf{x}+\mathbf{E}^\intercal\cdot\mathbf{x}+\mathbf{x}\otimes \mathbf{y}.
\end{align*}
Bob's share becomes instead
\[\mathbf{w}:=-(\id_{m\cdot \ell}\otimes \mathbf{s}^\intercal) \cdot \mathbf{B}^\intercal\cdot \mathbf{d}=-(\id_{m\cdot \ell}\otimes \mathbf{s}^\intercal) \cdot \mathbf{B}^\intercal\cdot\mathbf{A}\cdot \mathbf{x}.\]
This allows us to apply recursion without blowing up the size of $\mathbf{s}$: Since $\mathbf{w}$ is a bilinear function of $\mathbf{d}$ and $\mathbf{s}$, if the parties hold a ``noisy'' secret-sharing of $\mathbf{d}'\otimes \mathbf{s}$, they can easily convert it into a ``noisy'' secret-sharing of Bob's share by computing a local linear operation depending on $\mathbf{B}$. Moreover, since $\mathbf{B}$ is a binary matrix, this linear computation will not significantly increase the ``noisiness'' of the secret-sharing. At each recursion step, the size of $\mathbf{d}'$ decreases by a factor of $t(\sec):=m/n$; the size of $\mathbf{s}$, on the other hand, remains always the same. The final result is an LWE-based succinct OTE where the digest and the CRS dimensions are $\poly(\sec)$ and the encoding dimension is $O(\log M)\cdot \ell\cdot \poly(\sec)$. 

\paragraph{Succinct MOLEs and VOLEs.} No, this paragraph is not about tiny mammals: It's about two cryptographic primitives called \emph{matrix oblivious linear evaluation} and (non-interactive) \emph{vector oblivious linear evaluation}. In a MOLE, Alice holds a matrix $\mathbf{M}\in\Z_q^{m\times \ell}$ where $m\gg \ell$, whereas Bob holds a secret vector $\mathbf{x}\in\Z_q^\ell$. Their goal is to derive a secret-sharing of $\mathbf{M}\cdot \mathbf{x}$ in one round and without revealing any information about $\mathbf{x}$. A VOLE corresponds to a MOLE in the special case $\ell=1$. 

We would like to minimize the communication complexity of these protocols, especially in relation to $m$. It is easy to see that the succinct OTE protocol we just presented gives immediately a succinct MOLE where communication scales logarithmically in $m$: First, we use the succinct OTE protocol to compute a secret-sharing of $\mathsf{vec}(\mathbf{M})\otimes \mathbf{x}$, then, we apply local linear operations on the shares to obtain a secret-sharing of $\mathbf{M}\cdot \mathbf{x}$. Along the way, this solves a question left open in \cite{EC:AbrRoySch24}: We have just built the first non-interactive VOLE with logarithmic communication in $m$ from LWE.

\paragraph{Reverse Trapdoor Hashing.} We observe that many laconic function evaluation schemes have a particular structure \cite{FOCS:QuaWeeWic18,FOCS:HsiLinLuo23,C:DHMWW24,C:Wee24}. First of all, their digests consist of matrices $\mathbf{A}_f$ with a constant number of columns and a number of rows proportional to the output size of $f$. Moreover, the encoding can be split into two parts: A function-independent pre-encoding $E$ and an input-independent post-encoding $c$ consisting of an LWE-like sample where the matrix is (essentially) the digest $\mathbf{A}_f$ and the secret $\mathbf{s}$ is a random vector generated by the pre-encoding procedure. Finally, the output is computed by rounding the sum $c+\Eval(E, f)$. 

We observe that, if we ignore the noise, LWE samples are essentially matrix-vector multiplications. Therefore, by relying on the succinct MOLE we just built, the parties can derive a noisy secret-sharing of the post-encoding $c$ in a single round of simultaneous interaction and with logarithmic communication in the output size of $f$. Moreover, the encoder can send the pre-encoding $E$ along with its MOLE message. In this way, the parties can derive a secret-sharing of the output without the need for further interaction.

This yields the first rate-1 \emph{reverse} trapdoor hashing scheme. Usually, in rate-1 trapdoor hashing, we obtain a secret-sharing of $f(x)$ by sending a digest of $x$ and generating an encoding key for $f$. In reverse trapdoor hashing, we do the opposite: We send a digest for $f$ and an encoding key for $x$, {so rate in this context is defined as a function of $|f|$}. In our construction, the former corresponds to the MOLE hash of $\mathbf{A}_f$, whereas the latter corresponds to the pre-encoding $E$ and the MOLE encoding of $\mathbf{s}$.\footnote{We mention that we can recover the usual notion of trapdoor hashing via universal circuits. On the other hand, the reverse implication does not seem to trivially hold, since the universal circuit would introduce an efficiency penalty in the size of the encoding.}

\paragraph{Adaptive Lattice Encodings.} Much of the recent advancement in lattice-based homomorphic cryptography can be traced back to a single technique introduced in 2014 by Boneh et al. \cite{EC:BGGHNS14}: $\mathsf{BGG^+}$ encodings. Suppose that we are provided with a CRS consisting of a matrix $\mathbf{A}\in\Z_q^{k\times (\ell\cdot k\cdot \log q)}$. A $\mathsf{BGG^+}$ encoding of a bit string $\mathbf{x}\in\{0, 1\}^\ell$ consists of the vector
\[\mathbf{c}^\intercal=\mathbf{s}^\intercal \cdot(\mathbf{A}-\mathbf{x}^\intercal\otimes \mathbf{G})+\mathbf{e}^\intercal\]
where $\mathbf{G}:=\id_k\otimes \mathbf{g}_q^\intercal$, $\mathbf{s}\sample\Z_q^k$ and $\mathbf{e}\sample\chi(\bbbone^\sec)$. These encodings have amazing homomorphic properties: For any function $f:\{0, 1\}^\ell\rightarrow\{0, 1\}$, there exist efficiently computable, low-norm matrices $\mathbf{H}_f$ and $\mathbf{H}_{f, x}$ (independent of the random $\mathbf{s}$ and the noise $\mathbf{e}$ of the encoding) such that
\[\mathbf{c}^\intercal\cdot \mathbf{H}_{f, x}\approx \mathbf{s}^\intercal \cdot (\mathbf{A}_f-f(x)\cdot \mathbf{G})\]
where $\mathbf{A}_f:=\mathbf{A}\cdot \mathbf{H}_f$. Alas, $\mathsf{BGG^+}$ encodings are secure only in the selective setting: If $\mathbf{x}$ is independent of $\mathbf{A}$, the encoding $\mathbf{c}$ looks like a random vector, if, however, $\mathbf{x}$ is adaptively chosen after seeing $\mathbf{A}$, there exists an attack that allows us to recover $\mathbf{s}$. To see why, observe that there exists a binary matrix $\mathbf{H}'$ such that, for any matrix $\mathbf{M}\in\Z_q^{k\times k}$, if $\mathbf{x}=\mathsf{Bits}(\mathbf{M})$, we have that
\[\mathbf{c}^\intercal\cdot \mathbf{H}'\approx \mathbf{s}^\intercal \cdot (\mathbf{A}\cdot \mathbf{H}'-\mathbf{M}).\]
Suppose for convenience that $q$ is a power of 2. To recover the $i$-th most significant bit of $\mathbf{s}$, it is sufficient to set $\mathbf{M}:=\mathbf{A}\cdot \mathbf{H}'- 2^i\cdot \id_k$ and compute the most significant bit of $\mathbf{c}^\intercal\cdot \mathbf{H}'$. This proves that the Adaptive LWE assumption of \cite{FOCS:QuaWeeWic18} does not hold in the optimistic parameter setting in which $\ell$ can be arbitrarily bigger than $k$. Notice in the provable parameter setting (where the encodings are secure under the \emph{subexponential} hardness of LWE), our attack fails as the bound on $\ell$ is too small to encode $\mathbf{M}:=\mathbf{A}\cdot \mathbf{H}'- 2^i\cdot \id_k$.

To circumvent the attack without relying on complexity leveraging (and therefore obtain better asymptotic parameters), we introduce a new version of lattice encodings: The encoding of $\mathbf{x}$ is now
\[\mathbf{c}^\intercal=\mathbf{s}^\intercal\cdot \mathbf{A}+\mathbf{r}^\intercal\cdot (\mathbf{x}^\intercal\otimes \mathbf{G})+\mathbf{e}^\intercal\]
where $\mathbf{s}\sample\Z_q^k$, $\mathbf{r}\sample\Z_q^k$ and $\mathbf{e}\sample\chi(\bbbone^\sec)$. We call $\mathbf{s}$ \emph{the encryption key} of the encoding, whereas we refer to $\mathbf{r}$ as \emph{the authentication key}. We say that $\mathbf{c}$ is a $(\mathbf{s}, \mathbf{r})$-encoding. Observe that $\mathsf{BGG^+}$ encodings correspond to the special case in which $\mathbf{r}=-\mathbf{s}$. It is easy to see also that these encodings are adaptively secure under standard LWE, no matter the value of $\mathbf{x}$ and $\mathbf{r}$.

The drawback, however, becomes clear when we look at the homomorphic properties of the modified scheme. It is easy to see that the construction is linearly homomorphic, however, we no longer know how to perform multiplications. Or at least, we do not know \emph{when the factors are encoded using the same keys}. If instead the authentication key of the first encoding matches the encryption key of the second one, there is a solution: Suppose that we want to multiply the encodings $\mathbf{c}_x^\intercal=\mathbf{s}^\intercal\cdot \mathbf{A}+\mathbf{r}^\intercal\cdot (x\cdot \mathbf{G})+\mathbf{e}_x^\intercal$ and $\mathbf{c}_y^\intercal=\mathbf{r}^\intercal\cdot \mathbf{B}+\mathbf{t}^\intercal\cdot (y\cdot \mathbf{G})+\mathbf{e}_y^\intercal$. We observe that
\[\mathbf{c}_z^\intercal:=-\mathbf{c}_x^\intercal\cdot \mathbf{G}^{-1}(\mathbf{B})+x\cdot \mathbf{c}_y^\intercal\approx-\mathbf{s}^\intercal\cdot \mathbf{A} \mathbf{G}^{-1}(\mathbf{B})+\mathbf{t}^\intercal \cdot (x y\cdot \mathbf{G}).\]
In other words, we have obtained an encoding of $x\cdot y$ using $\mathbf{s}$ as encryption key and $\mathbf{t}$ as authentication key. Moreover, the matrix relative to the encoding is $-\mathbf{A}\cdot \mathbf{G}^{-1}(\mathbf{B})$. Notice that this is publicly computable, no need to know the keys or the plaintexts! To summarise, we are able to compute linear operations between $(\mathbf{s}, \mathbf{r})$-encodings obtaining other $(\mathbf{s}, \mathbf{r})$-encodings. Moreover, we are able to perform multiplications between $(\mathbf{s}, \mathbf{r})$-encodings and $(\mathbf{r}, \mathbf{t})$-encodings, obtaining $(\mathbf{s}, \mathbf{t})$-encodings as results.

\paragraph{Compressing Adaptive Lattice Encodings.} We present a procedure to compress our adaptive lattice encodings, where a compressed encoding of $\mathbf{x}\in\{0, 1\}^\ell$ will have size $\poly(\log \ell, \sec)$. By leveraging the knowledge of $\mathbf{x}$, it can then be re-expanded into a standard adaptive encoding. Once again, our techniques rely on our succinct OTE protocol and the (bi)linear structure of Bob's share derivation. Specifically, the compressed encoding is composed of two parts $\mathbf{h}$ and $E$. The former consists of an adaptive lattice encoding of $\mathbf{d}:=\OTE.\Hash(\mathbf{x})$. Let $\mathbf{r}$ be the authentication key of $\mathbf{h}$; then $E$ consists of an OTE encoding of a fresh authentication key $\mathbf{t}$ where $\mathbf{r}$ is used as randomness for $E$ (this is secure as $\mathbf{h}$ looks random even when $\mathbf{r}$ is leaked). We observe that
\begin{align*}
    \mathbf{h}^\intercal&\approx\mathbf{s}^\intercal \cdot \mathbf{A}+\mathbf{r}^\intercal \cdot (\mathbf{d}^\intercal\otimes \mathbf{G})\\
    &=\mathbf{s}^\intercal \cdot \mathbf{A}+\mathbf{r}^\intercal \cdot (\mathbf{d}^\intercal\otimes \id_k\otimes \mathbf{g}_q^\intercal)\\
    &=\mathbf{s}^\intercal \cdot \mathbf{A}+(\mathbf{d}^\intercal\otimes \mathbf{r}^\intercal \otimes \mathbf{g}_q^\intercal).
\end{align*}
In other words, $\mathbf{h}$ can be viewed as some sort of encoding of $\mathbf{d}\otimes \mathbf{r}$ where $\mathbf{r}$ is the randomness of $E$. Now, Bob's share of $\mathbf{x}\otimes \mathbf{t}$ would be $\mathbf{w}:=\mathbf{P}\cdot (\mathbf{d}\otimes \mathbf{r})$, where $\mathbf{P}$ is a public low-norm matrix derived from $\mathbf{B}$. Thus, by multiplying $\mathbf{h}$ on the right by $\mathbf{P}$, we derive
\[\hat{\mathbf{h}}^\intercal:=\mathbf{h}^\intercal\cdot \mathbf{P}\approx \mathbf{s}^\intercal \cdot \mathbf{A} \mathbf{P}+\mathbf{w}^\intercal\otimes \mathbf{g}_q^\intercal.\]
 Using $\mathbf{x}$ and $E$, we can also derive the other share $\mathbf{v}$. We conclude by observing that
\begin{align*}
    \hat{\mathbf{h}}^\intercal+\mathbf{v}^\intercal\otimes \mathbf{g}_q^\intercal&\approx \mathbf{s}^\intercal \cdot \mathbf{A} \mathbf{P}+ (\mathbf{x}^\intercal\otimes \mathbf{t}^\intercal \otimes \mathbf{g}_q^\intercal)\\
    &=\mathbf{s}^\intercal \cdot \mathbf{A} \mathbf{P}+ \mathbf{t}^\intercal\cdot (\mathbf{x}^\intercal\otimes \mathbf{G})
\end{align*}
We have just obtained a $(\mathbf{s}, \mathbf{t})$-encoding of $\mathbf{x}$.

\paragraph{Rate-1 Adaptive LFE.} We build our rate-1 adaptive LFE scheme using our compressed adaptive encodings in two steps: First, we construct a weak variant of (adaptive) reverse TDH for the functions that map pairs $(\mathbf{x}, \mathbf{a})$ to $f(\mathbf{x})\otimes \mathbf{a}$ where $f\in\mathsf{NC}_1$. The scheme satisfies all the security properties of standard reverse TDH, however, in order for correctness to hold, the hasher needs to know $\mathbf{x}$ (but not $\mathbf{a}$) for the derivation of her share. {However, importantly, the scheme additionally achieves} polylogarithmic communication in the size of $\mathbf{x}$ (but not $\mathbf{a}$). As a second step, we use our \emph{input-succinct} reverse TDH to build an adaptive LFE scheme for all depth-$D$ functions $f:\{0, 1\}^m\rightarrow \{0, 1\}^\ell$. The size of the hash will be $\poly(\sec)$ whereas the size of the encoding will be $\ell+m+D\cdot \poly(\sec)$.

\paragraph{Step I: Input-Succinct Reverse TDH.} We start by describing the input-succinct reverse TDH. Suppose that we want to evaluate functions of depth at most $\log d$: Each of them can be converted into an RMS program of depth $d$ \cite{C:BoyGilIsh16}. We recall that an RMS program consists of an arithmetic circuit over $\Z$ where multiplications are allowed only if one of the factors is an input. 

We start by describing the generation of encoding keys: We sample a chain of $d$ keys $\mathsf{s}_0, \dots, \mathbf{s}_{d-1}$, we set $\mathbf{s}_d:=\mathbf{a}$ and we generate compressed encodings $(\mathbf{h}_i, E_i)_{i\in[d]}$ so that, when we expand $(\mathbf{h}_i, E_i)$, we obtain a $(\mathbf{s}_i, \mathbf{s}_{i+1})$-encoding of $(\mathbf{x}, 1)$. Notice that we can homomorphically evaluate $f$ on these encodings: The operation proceeds by levels, starting from level $1$ (the inputs) to level $d$ (the output). A level-$i$ encoding consists of any $(\mathbf{s}_0, \mathbf{s}_i)$-encoding. We observe that $(\mathbf{h}_0, E_0)$ gives us a level-1 encoding of the inputs.

In the previous paragraphs, we showed that the levels are closed under linear operations. Moreover, we can also perform multiplications by inputs: If the first factor belongs to level $i$, we can multiply it by the $(\mathbf{s}_i, \mathbf{s}_{i+1})$-encoding of the other factor (the input), which can be derived from $(\mathbf{h}_i, E_i)$. Finally, we can perform hops across levels (always from level $i$ to level $i+1$) by performing multiplications by 1\footnote{Remember that the expansion of $(\mathbf{h}_i, E_i)$ provides also an $(\mathbf{s}_i, \mathbf{s}_{i+1})$-encodings of 1.}. To summarise, given $\mathbf{x}$ and $(\mathbf{h}_i, E_i)_{i\in[d]}$, Alice is able to derive an encoding
\begin{align*}
    \mathbf{c}_f^\intercal&\approx \mathbf{s}_0^\intercal \cdot \mathbf{A}_f+\mathbf{a}^\intercal \cdot (f(\mathbf{x})^\intercal \otimes \mathbf{G})\\
    &=\mathbf{s}_0^\intercal \cdot \mathbf{A}_f+ (f(\mathbf{x})^\intercal \otimes \mathbf{a}^\intercal\otimes \mathbf{g}_q^\intercal).
\end{align*}
To derive a secret-sharing of $f(\mathbf{x}) \otimes \mathbf{a}$, it is therefore sufficient that the parties run a succinct MOLE to compute a secret-sharing of $\mathbf{A}_f^\intercal\cdot \mathbf{s}_0$ similarly to what we did for the other reverse TDH we built. At that point, it is just a matter of applying local linear computations and rounding.

\paragraph{Step II: Building Rate-1 Adaptive LFE.} It is finally time to talk about the rate-1 adaptive LFE scheme. Our approach is the following: We pick a \emph{constant} $d$ and we decompose our function $f$ as $f_{L-1}\circ\dots\circ f_0$ where each $f_i$ is described by an RMS program of depth $d$. Our goal is to use our reverse TDH to evaluate all $f_i$ until we obtain the output. The issue is that Bob does not know what input to encode for $f_i$. 

Instead of evaluating $f_i$, we evaluate a function that allows us to retrieve the encoding key for $\mathbf{x}_i:=(f_i\circ \dots\circ f_0)(\mathbf{x})$. Specifically, for every $i\in[L]$, define\footnote{In our construction, $\OTE.\Hash$ has depth 0 because it is linear.}
\[\hat{f}_i:(\mathbf{x}, \mathbf{a})\longmapsto \OTE.\Hash(f_i(\mathbf{x}))\otimes \mathbf{a}.\] Suppose that Alice sent a hash for $\hat{f}_{i}$ for every $i\in[L]$. Suppose also that we somehow managed to find a way to provide Alice with an encoding key for $(\mathbf{x}_i, \mathbf{a}_{i+1})$ where $\mathbf{a}_{i+1}\sample\Z_q^k$. Under these premises, the parties can derive an additive secret-sharing $\mathbf{y}^{\mathsf{A}}_{i+1}+\mathbf{y}^{\mathsf{B}}_{i+1}=\mathbf{d}_{i+1}\otimes \mathbf{a}_{i+1}$ where $\mathbf{d}_{i+1}=\OTE.\Hash(\mathbf{x}_{i+1})$. We can convert this into a compressed adaptive encoding of $\mathbf{x}_{i+1}$ by making Bob send \[\mathbf{c}_{i+1}^\intercal=\mathbf{s}_{i+1}^\intercal\cdot \mathbf{A}+\mathbf{e}_{i+1}^\intercal+(\mathbf{y}^{\mathsf{B}}_{i+1})^\intercal\otimes \mathbf{g}_q^\intercal\] where $\mathbf{s}_{i+1}\sample\Z_q^k$ and $\mathbf{e}_{i+1}\sample\chi(\bbbone^\sec)$. By adding $(\mathbf{y}^{\mathsf{A}}_{i+1})^\intercal\otimes\mathbf{g}_q^\intercal$ to $\mathbf{c}_{i+1}$, Alice can derived a $(\mathbf{s}_{i+1}, \mathbf{a}_{i+1})$-encoding of $\mathbf{d}_{i+1}$. Indeed,
\begin{align*}
    \mathbf{c}_{i+1}^\intercal+(\mathbf{y}^{\mathsf{A}}_{i+1})^\intercal\otimes\mathbf{g}_q^\intercal&\approx \mathbf{s}_{i+1}^\intercal\cdot \mathbf{A}+(\mathbf{y}^{\mathsf{B}}_{i+1}+\mathbf{y}^{\mathsf{A}}_{i+1})^\intercal\otimes \mathbf{g}_q^\intercal\\
    &=\mathbf{s}_{i+1}^\intercal\cdot \mathbf{A}+(\mathbf{d}_{i+1}^\intercal\otimes \mathbf{a}_{i+1}^\intercal\otimes \mathbf{g}_q^\intercal)\\
    &=\mathbf{s}_{i+1}^\intercal\cdot \mathbf{A}+\mathbf{a}_{i+1}^\intercal\cdot (\mathbf{d}_{i+1}^\intercal\otimes \mathbf{G}).
\end{align*}Bob can of course send also the other information needed to complete the encoding key for $(\mathbf{x}_{i+1}, \mathbf{a}_{i+2})$ as this is independent of $f$ and $\mathbf{x}_{i+1}$. Notice that the size of all this material that Bob sends is independent of the input size and the output size of $f$. By continuing in this way, the parties end up with an additive secret-sharing of $f(\mathbf{x})\otimes \mathbf{a}_L$ where $\mathbf{a}_L\sample\Z_q^k$.

There is still one matter we need to take care of: In order to perform the operations we just described, Alice needs to know $\mathbf{x}$. So, how can we achieve privacy of the input? The trick is the same as in \cite{TCC:BTVW17}: We provide Alice with a $\mathsf{GSW}$ \cite{C:GenSahWat13} encryption of $\mathbf{x}$ using $\mathbf{a}_L$ as a secret key (we also send the corresponding encoding key, which is polylogarithmic in size). Then, instead of evaluating $f$, we evaluate $f':=\GSW.\Eval(f, \cdot)$. At the end, the parties obtain a secret-sharing of $\mathsf{Bits}(\ct)\otimes \mathbf{a}_L$ where $\ct$ is a $\GSW$ encryption of $f(\mathbf{x})$. Given that $\mathbf{a}_L$ is the secret-key and the $\GSW$ decryption consists of linear operations  between $\mathbf{a}_L$ and $\ct$ (followed by rounding), the parties can obtain a secret-sharing of $f(\mathbf{x})$ by applying only local operations. This immediately gives rate-1 communication in the output. What about rate-1 communication in the input? Well, instead of sending a $\GSW$ encryption of $\mathbf{x}$, send $\mathbf{z}:=\mathbf{x}\oplus \mathsf{PRG}(K)$ for $K\sample\{0, 1\}^\sec$ and a $\GSW$ encryption of $K$. Then, during the evaluation of $f'$, we remove the one-time-pad inside FHE.

\subsection{Other Applications}\label{sec:apps}

We outline a few additional additional applications of our results. In particular, we discuss some of the new implications from our construction of reverse TDH.

\paragraph{Regular Trapdoor Hash.}
As an immediate implication of our reverse TDH, we obtain a (regular) TDH \cite{C:DGIMMO19} with close to optimal parameters, that is: The size of the encoding key is $\lvert f\rvert \cdot \poly(\sec, \log m)$ where $\lvert f\rvert$ denotes the size of the description of the encoded function and $m$ denotes the size of its input. For instance, notice that if $f$ is a point function with domain of size $L$, $\lvert f\rvert=\log L$. 
We achieve this with a simple application of universal circuits: Instead of hashing a function, we hash the universal circuit $U_\mathbf{x}$ with the input $\mathbf{x}$ hardwired, that on input a function $f$, returns $f(\mathbf{x})$. Note that the size of the digest is anyway constant (ignoring factors in the security parameter) so the complexity of the TDH only grows with the bit description of $f$.

To our knowledge, this (along with a concurrent work \cite{EC:BJSS25}) is the first trapdoor hashing schemes that support all functions. This is also the first LWE-based TDH constructions achieving nearly optimal communication for a non-trivial class of functions. To our knowledge, the only other construction where the size of the encoding key is sublinear in the dimension of the input is a recently built pairing-based TDH for point functions \cite{Nico}. Such construction, however, pays the succinctness of the encoding key with a non-succinct CRS of size $O(m)$.

\paragraph{Homomorphic Secret Sharing, Spooky Encryption and Public-Key PCGs.}
In addition, note that the reconstruction of our reverse TDH is additive over $\Z_2$, thus, a reverse trapdoor hash also implies the existence of a (public-key) 2-party homomorphic secret sharing \cite{C:BoyGilIsh16} scheme as follows: Suppose that Alice and Bob respectively hold inputs $\mathbf{x}$ and $\mathbf{y}$ and they want to evaluate the function $f$. Suppose also that $\mathbf{x}$ is much longer than $\mathbf{y}$. Alice proceeds by hashing the function $f_\mathbf{x}$ that maps any $\mathbf{y}$ to $f(\mathbf{x}, \mathbf{y})$. She sends the digest to Bob and keeps the randomness as her part of the share of $\mathbf{x}$. Bob, on the other hand, sends an encoding key $\ek$ for $\mathbf{y}$ to Alice. He keeps the corresponding trapdoor $\td$ as his share of $\mathbf{y}$. By the correctness of reverse TDH, the parties can locally derive a secret sharing of the output. Notice also that the scheme is succinct in $\mathbf{x}$: The total communication of the protocol is $\lvert \mathbf{y}\rvert \cdot \poly(\sec, \log \lvert \mathbf{x}\rvert)$!

Previously, succinct homomorphic secret sharing had been built by Abram, Roy and Scholl \cite{EC:AbrRoySch24} under several assumptions, including LWE, DCR and DDH over class groups. Their constructions, however, supported only a limited class of functions: Alice and Bob could only compute secret-sharings of $\mathbf{x}^\intercal\cdot C(\mathbf{y})$ for any circuit $C\in\mathsf{NC}_1$. Their constructions have also a second drawback: The total complexity of the protocol is $\lvert \mathbf{y}\rvert \cdot \lvert \mathbf{x}\rvert^\varepsilon \cdot \poly(\sec)$ for a constant $\varepsilon\in (0, 1)$. Our solution instead scales polylogarithmically in $\lvert \mathbf{x}\rvert$. On the negative side, unlike \cite{EC:AbrRoySch24}, our solution does not allow the parties to evaluate a function that is adaptively chosen \emph{after} the secret-sharing phase. We can however plug our MOLE in the constructions of \cite{EC:AbrRoySch24} to obtain an HSS scheme that allows the evaluation of any \emph{adaptively} chosen function $\mathbf{x}^\intercal\cdot C(\mathbf{y})$ with improved communication $\lvert \mathbf{y}\rvert \cdot \poly(\sec, \log \lvert \mathbf{x}\rvert)$. 

Observe that our succinct HSS scheme can be also viewed as a form of 2-party spooky encryption for additively shared correlation \cite{C:DHRW16} where one of the parties can just send a small hash of its input instead of a full-size ciphertext. Once again, differently from \cite{C:DHRW16}, our construction does not allow us to choose the correlation after we committed to the inputs.

Finally, our HSS scheme can have interesting applications in the context of (public-key) pseudorandom correlation generators (PCGs) \cite{C:BCGIKS19,EC:OrlSchYak21,EC:AbrSchYak22}, especially when the tackled additively-shared correlation takes a long input from Alice: Let $\mathcal{C}(x)$ be the correlation function with long input. Alice can sample $s_0\sample\{0, 1\}^\sec$ and hash the function that maps $y$ to $(\mathcal{C}(x; r_i))_{i\in[n]}$ where $(r_0, \dots, r_{n-1})\gets \mathsf{PRG}(s_0\oplus y)$. Bob instead picks a random $y\sample\{0, 1\}^\sec$ and sends its encoding to Alice.

\paragraph{Rate-$\mathbf{1/2}$ Laconic Oblivious Transfer.} Since our construction of reverse TDH also extends to RAM programs using \cite{C:DHMWW24}, we obtain a new construction of laconic oblivious transfer \cite{C:CDGGMP17} with rate $1/2$ in the batch settings (which is best possible), where one transfers a set of messages with respect to different indices. The receiver hashes the function $f_D$ that has hardwired a database $D$ and, on input a $\mathsf{PRF}$ key $k$, does the following for all indices $i$ and all bits $b$:
\begin{itemize}
    \item If $b = D_i$: Return $0$.
    \item Else return $\mathsf{PRF}(k, i)$.
\end{itemize}
Then, on input a set of indices $I$ and pairs of bits $\{m_{i,0}, m_{i,1}\}_{i\in I}$, the sender
samples a key $k$ and sends $\ek \sample \Gen(\hk, k)$, where $\Gen$ is the encoding generation algorithm of the TDH, along with:
\[
c_{i,b} := \left\{\Dec(\hk, \td, d)_{i,b} \oplus m_{i,b}\right\}_{i,b}
\]
where we abuse the notation and assume that the TDH decoding algorithm $\Dec(\hk, \td, d)_{i,b}$ returns the $(i,b)$-bit of the share (which can be computed by a RAM program in time independent of the size of the database). Note that if $b = D_i$, then $\Dec(\hk, \td, d)_{i,b} = \Enc(\hk, \ek, f_D, \rho)_{i,b}$ and therefore the receiver can recover $m_{i,b}$ running the TDH encoding algorithm $\Enc$. Otherwise the pseudorandomness of $\mathsf{PRF}$ guarantees that the message is computationally hidden.

\paragraph{Attribute-Based Non-Interactive Key-Exchange.} Finally, reverse TDH implies the existence of a non-interactive key exchange (NIKE) with the following additional property: One of the two parties can include the hash of a function $f$ as part of their public key, whereas the other party holds an input $\mathbf{x}$. The NIKE succeeds if $f(\mathbf{x}) = 1$ and otherwise the key of either party is computationally indistinguishable from random. This can be constructed from reverse TDH in a natural manner: The former party hashes the function $F_{f, k_0}$ that takes as input some $\mathbf{x}$ and a key $k_1$ and returns $0$ if $f(\mathbf{x}) = 1$ and $\mathsf{PRF}(k_0 , \mathbf{x}) \oplus \mathsf{PRF}(k_1 , \mathbf{x})$ otherwise. If $f$ is satisfied, then both parties hold the same share, that can be used as a shared key, otherwise the pseudorandomness of $\mathsf{PRF}$ protects the share of either party.

\subsection{Concurrent Work}

A concurrent work by Boyle et al.~\cite{EC:BJSS25} also construct a family of TDHs for all functions $f:\{0,1\}^m \to \{0,1\}^\ell$, using a similar idea. An important difference is that they rely on the $|\mathbf{x}|^{2/3}$-succinct VOLE protocol from \cite{EC:AbrRoySch24}, whereas we (implicitly) construct a $\polylog(|\mathbf{x}|)$-succinct one, derived from our OTE. This translates in different parameters for the encoding key of the TDH: The encoding key of \cite{EC:BJSS25} has size $(|f|+ \ell^{2/3} + D)\cdot \poly(\sec)$, whereas in our TDH the encoding key has size $\lvert f\rvert \cdot \poly(\sec, \log m)$, which is close to optimal.

Besides TDHs, the results in \cite{EC:BJSS25} are largely orthogonal to ours. We also mention here that plugging in our OTE protocol in \cite{EC:BJSS25} leads to similar parameter improvements for their other applications. For instance, following the outline of \cite{EC:BJSS25}, OTE yields a rate-1 fully homomorphic encryption (FHE) with optimal parameters. We fully credit \cite{EC:BJSS25} for discovering the connection, and we sketch here the transformation only for the sake of completeness. Take any FHE scheme with almost-linear decryption \cite{TCC:BDGM19}, i.e., where decryption is a linear function is the secret key $\mathbf{s}$, followed by a rounding, such as \cite{C:GenSahWat13}. Then add to the public key an MOLE encoding $\mathsf{Enc}(\mathbf{s})$. We can compress $m$ ciphertexts $(\mathbf{c}_1, \dots, \mathbf{c}_m)$ as follows: Stack them into a matrix $\mathbf{C}$, then compute the MOLE-hash $\mathsf{Hash}(\mathbf{C})$, and run the Hash-Eval algorithm, to obtain a additive share of $\mathsf{Round}(\mathbf{C}\cdot \mathbf{s}) \in \{0,1\}^m$. Return the hash $\mathsf{Hash}(\mathbf{C})$, along with the bits of the share. Decryption works by simply running the Encoder-Eval algorithm of the MOLE, and reconstructing the output. Crucially, the compressed ciphertext consists of a hash (whose size is a fixed polynomial in the security parameter) plus $m$ bits, i.e., it is $\poly(\lambda) + m$, which is optimal. This improves upon \cite{TCC:BDGM19} since the size of the public key does not depend on $m$ (no amortization is needed).

\section{Preliminaries}
\paragraph{Notation.}
We denote the security parameter by $\sec$. We say that a function $\negl(\sec)$ is negligible if it vanishes faster than any polynomial, i.e., $0\le \negl(\sec) \leq \sec^{-\omega(1)}$. We say that an event happens with overwhelming probability, if it occurs with probability negligibly close to $1$. For any $n\in\N$, we use $[n]$ to denote the set $\{0, 1, \dots, n-1\}$. 

All vectors are denoted using lowercase bold font, whereas matrices are denoted using uppercase bold font; by default, all vectors are columns.
Given a matrix $\mathbf{M}$, we denote its transpose by $\mathbf{M}^\intercal$, whereas $\mathsf{vec}(\mathbf{M})$ denotes its vectorisation, i.e., stacking the columns of $\mathbf{M}$ one underneath the other. For any $n\in\N$, we denote the $n\times n$ identity matrix by $\id_n$ and the $n$-dimensional row vector where all entries are equal to 1 by $\mathbf{1}^n$. We define the Kronecker product between two $n\times m$ matrices $\mathbf{A} \otimes \mathbf{B}$ to be
\[
\mathbf{A} \otimes \mathbf{B} := \begin{pmatrix}
    a_{1,1}\mathbf{B} &\dots& a_{1,m}\mathbf{B}\\
    \vdots &\ddots& \vdots\\
    a_{n,1}\mathbf{B} &\dots& a_{n,m}\mathbf{B}
\end{pmatrix}.
\]
For a vector $\mathbf{x}\in \Z_q^n$, we denote by $\|\mathbf{x}\|_{\infty}$ its infinity norm, i.e., {the maximal distance from $0$ of its coordinate in the standard centered representative convention,} and we extend this notation to matrices by taking the maximum over their columns. 
For any integer $q>0$, let $\mathbf{g}_q^\intercal$ be the gadget row-vector $(1, 2, \dots, 2^{\log q})$ and we omit the subscript when clear from the context. Let $\mathbf{G}^{-1}$ be the algorithm that takes as input a matrix $\mathbf{M}\in\Z_q^{n\times m}$ and outputs a matrix $\mathbf{M}'\in\Z_2^{(n\cdot \log q)\times m}$, where each column is derived by stacking the bit-decompositions of the entries in the corresponding column of $\mathbf{M}$. Notice that $(\id_n\otimes \mathbf{g}_q^\intercal)\cdot \mathbf{G}^{-1}(\mathbf{M})=\mathbf{M}$. We use $\mathsf{Bits}(\mathbf{M})$ to denote $\mathsf{vec}(\mathbf{G}^{-1}(\mathbf{M}))$.
Given an integer $x\in\Z_q$ and an integer $p$ that divides $q$, we use $\lceil x \rfloor_p$ to denote rounding to $\Z_p$, in other words, if $x=y\cdot q/p+z$ where $ z\in[ -q/2p, q/2p)$, we have $\lceil x \rfloor_p=y$.

We state the following useful lemma about matrices.

\begin{lemma}[Matrix Linearisation]
\thmlab{linearisation}
    There exists a deterministic polynomial-time computable function:
    \[
    \mathsf{Lin} : \Z_q^{(t\cdot \ell)\times m} \to \Z_q^{t\times (m\cdot \ell)}
    \]
    such that, for any matrix 
    $\mathbf{B}\in \Z_q^{(t\cdot \ell)\times m}$ and any pair of vectors $\mathbf{x}\in\Z_q^{m}$ and $\mathbf{s}\in\Z_q^\ell$, it holds that:
    \[
    \mathsf{Lin}(\mathbf{B})(\mathbf{x}\otimes \mathbf{s})=(\id_t\otimes \mathbf{s}^\intercal) \mathbf{B}\mathbf{x} \quad\text{and}\quad \|\mathsf{Lin}(\mathbf{B})\|_\infty = \|\mathbf{B}\|_\infty.
    \]
    \end{lemma}
\begin{proof}
Note that the $h$-th entry of $(\id_t\otimes \mathbf{s}^\intercal) \mathbf{B} \mathbf{x}$ is: 
\[\sum_{i\in[\ell]}\sum_{j\in[m]} \mathbf{B}_{\ell\cdot h+i, j} \cdot (\mathbf{s}_{i}\cdot \mathbf{x}_{j}).\]
In other words, there exists a matrix $\mathsf{Lin}(\mathbf{B})$, obtained by rearranging the entries of $\mathbf{B}$ such that $\mathsf{Lin}(\mathbf{B}) (\mathbf{x}\otimes \mathbf{s})=(\id_t\otimes \mathbf{s}^\intercal) \mathbf{B} \mathbf{x}$. Since rearranging the entries does not change the infinity norm, the claim follows.
\end{proof}

We also recall the definition of RMS program. Essentially this consists of an algebraic circuit over $\Z$ where we can multiply two wires only if at least one of them is an input.

\begin{definition}[Restricted Multiplication Straightline]
A restricted multiplication straightline program (RMS) consists of a polynomial-sized family of algebraic circuits over $\Z$ where the only allowed gates are the following:
\begin{itemize}
    \item Additions: this gate takes as input two wires $x$ and $y$ and outputs their sum $x+y$.
    \item Scalar multiplication: each of these gates is parametrised by a constant $\alpha\in\Z$. It takes as input a single wire $x$ and outputs $\alpha\cdot x$.
    \item Multiplication: this gate takes as input two wires $x$ and $y$ \emph{where $y$ is an input to the RMS program}. The output is their product $x\cdot y$.
\end{itemize}
We say that the program has depth $d$ if the \emph{multiplicative depth} of the program is $d$.

Let $T$ be a positive integer. We say that an RMS program is $T$-bounded if, during any evaluation over \emph{binary} inputs, the absolute value of the wires never exceeds $T$.
\end{definition}

We recall the following result which states that any circuit in $\mathsf{NC}_1$ can be converted into a polynomial size RMS program (paying exponentially in the depth).

\begin{theorem}[\cite{barrington,C:BoyGilIsh16}]
Let $f:\{0, 1\}^n\rightarrow \{0, 1\}$ be described by a boolean circuit of size $s$ and depth $d$ made entirely of NAND gates. Then, $f$ can be computed using a $1$-bounded RMS program of depth at most $2^d$ and size $O(s\cdot 2^d)$.
\end{theorem}

\subsection{Lattices and Learning with Errors}
Throughout the paper, we often rely on a low-norm distribution $\chi(\sec)$ over $\Z$. We say that $\chi(\sec)$ is $B(\sec)$-bounded if: \[\Pr\left[\lvert e\rvert\le B(\sec)\middle| e\sample\chi(\sec)\right]=1.\] Sometimes, we abuse notation and we write $\mathbf{v}\sample\chi(\sec)$, even if $\mathbf{v}$ is a vector: With this, we mean that each entry of $\mathbf{v}$ is sampled from $\chi(\sec)$ independently of all the others. {Henceforth, we always assume that $\chi(\lambda)$ is a $B(\sec)$-bounded distribution.}

We recall the learning with error assumption, introduced for the first time by Regev \cite{STOC:Regev05}.
\begin{definition}[Learning with Errors]
    Let $k:=k(\sec)$, $m:=m(\sec)$ and $q:=q(\sec)$ be positive integers. Let $\chi(\sec)$ be a low-norm distribution over $\Z$.
    We say that the $(\chi, k, m, q)$-LWE problem is hard if, for every PPT adversary $\Adv$ there exists a negligible function $\negl(\sec)$ such that, for every $\sec\in\N$, we have
    \[\left\lvert \Pr\left[\Adv(\bbbone^\sec, \mathbf{M}, \mathbf{u}_b)=b\middle| \begin{aligned}[c]
    &b\sample\{0, 1\}, \mathbf{s}\sample\Z_q^k\\
    &\mathbf{M}\sample\Z_q^{m\times k}, \mathbf{e}\sample\chi(\sec)^m\\
    &\mathbf{u}_0\gets \mathbf{M}\cdot \mathbf{s}+\mathbf{e}, \mathbf{u}_1\sample\Z_q^m\\
    \end{aligned}\right]-\frac{1}{2}
    \right\rvert\le \negl(\sec).\]
Suppose that $\chi(\sec)$ is $B(\sec)$-bounded. We call the quantity $\alpha:=q/B$ the modulus-noise ratio.
\end{definition}
We also define the \emph{non-uniform LWE} assumption identically as above, except that the matrix $\mathbf{M}$ is no longer uniformly random over $\Z_q^{m\times k}$ but over $\Z_2^{m\times k}$. It is shown in \cite{C:BLMR13} that non-uniform LWE is at least as hard as LWE, with a slightly increased parameter size.

\begin{theorem}[\cite{C:BLMR13}]
    Assume the hardness of $(\chi, k, m, q)$-LWE. Then, $(\chi, k\cdot \log q, m, q)$-NLWE is hard.
\end{theorem}

\section{Non-Interactive Oblivious Tensor Evaluation}

In the following we define and construct a succinct Non-Interactive Oblivious Tensor Evaluation (NI-OTE) protocol.

\subsection{Definitions} \label{defOTE}

We begin by defining the syntax of NI-OTE.

\begin{definition}[Non-interactive oblivious tensor evaluation]
Let $m:=m(\sec)$, $\ell:=\ell(\sec)$, and  $q:=q(\sec)$ be a positive integer. A NI-OTE for $\Z_q^m\otimes\Z_q^\ell$ consists of a tuple of PPT algorithms $(\Setup, \Hash, \Enc, \Hash\Eval, \Enc\Eval)$ with the following syntax:
\begin{description}
    \item[$\Setup(\bbbone^\sec)$:] The setup algorithm is probabilistic, takes as input the security parameter $\bbbone^\sec$ and outputs a public key $\pk$.
    
    \item[$\Hash(\pk, \mathbf{x})$:] The hashing algorithm is probabilistic and takes as input a public key $\pk$ and the description of a vector $\mathbf{x}\in\Z_q^{m}$. The output is a digest $d$ and hasher's private information $\psi$.
    \item[$\Enc(\pk, \mathbf{y})$:] The encoding algorithm is probabilistic and takes as input a public key $\pk$, and a vector $\mathbf{y}\in\Z_q^{\ell}$. The output is an encoding $E$ and encoder's private information $\phi$. 
    \item[$\Hash\Eval(\pk, E, \psi)$:] The hasher's evaluation algorithm is deterministic and takes as input a public key $\pk$, an encoding $E$ and hasher's private information $\psi$. The output is a vector $\mathbf{v}\in\Z_q^{m\cdot \ell}$.
    \item[$\Enc\Eval(\pk, d, \phi)$:] The encoder's evaluation algorithm is deterministic and takes as input a public key $\pk$, a digest $d$ and encoder's private information $\phi$. The output is a vector $\mathbf{w}\in\Z_q^{m\cdot \ell}$.
\end{description}
\end{definition}
Sometimes it will be convenient for us to fix the private information $\phi$ and provide it as an input to the encoding algorithm. In a slight abuse of notation, we denote this by 
$\Enc(\pk, \mathbf{y}, {[\phi]})$, in which case, the algorithm just outputs $E$. If the NI-OTE scheme satisfies this syntactical requirement, which in particular means that $\phi$ does not depend on $\mathbf{y}$, we say that the scheme is \emph{programmable}.

Additionally, we say that an NI-OTE has a \emph{bilinear encoder evaluation} if:
\begin{itemize}
    \item The digest $d$ is a vector in $\Z_q^n$.
    \item The encoder secret information $\phi$ is a vector in $\Z_q^k$.
    \item The encoder evaluation algorithm consists of
    \[
    \Enc\Eval(\pk, d, \phi):=\mathbf{P}\cdot(d\otimes \phi\otimes \mathbf{g}_q)
    \]
    where the matrix $\mathbf{P} \in\Z_q^{(m\cdot \ell)\times (n\cdot k\cdot \log q)}$ can be publicly derived from $\pk$.
\end{itemize}
We say that an NI-OTE is \emph{succinct} if the size of the hash and the size of the encodings are sublinear in the size of the  hasher's input. Depending on the context, we will make the dependence explicit. If only the hash is sublinear, then we say that the scheme is \emph{half-succinct}.

\paragraph{Correctness.} Next, we define (approximate) correctness for an NI-OTE, parametrized by an error function $\alpha$. If $\alpha = 0$, then we say that the NI-OTE is perfectly correct, or simply correct.
{
\begin{definition}[$\alpha$-Correctness]\label{def:alphacorrect}
An NI-OTE scheme $(\Setup, \allowbreak \Hash,  \allowbreak \Enc,  \allowbreak \Hash\Eval,  \allowbreak \Enc\Eval)$ is $\alpha$-adaptively correct if 
there exists a negligible function $\negl(\sec)$ such that, for every adversary $\Adv$ and sufficiently large $\sec\in\N$, we have that:
\[
\Pr\left[\lVert\mathbf{v}+\mathbf{w}-\mathbf{x}\otimes \mathbf{y}\rVert_\infty>\alpha\cdot\lVert \mathbf{x}\rVert_\infty\ \right] \leq \negl(\sec)
\]
where the probability is taken over the random choice of $\pk\sample\Setup(\bbbone^\sec)$, $(\mathbf{x}, \mathbf{y})\sample\Adv(1^\sec, \pk)$, $(d, \psi)\sample\Hash(\pk, \mathbf{x})$, $(E, \phi)\sample\Enc(\pk, \mathbf{y})$, and we define $\mathbf{v}:= \Hash\Eval(\pk, E, \psi)$ $\mathbf{w}:= \Enc\Eval(\pk, d, \phi)$.
\end{definition}
The definition of approximate correctness will be sufficient for our applications. Nevertheless, for future reference, we mention here that there is a general (standard) method to drive the correctness error down to $\alpha = 0$, which may be more convenient to work with. Specifically, let $p$ to be a divisor of $q$ with $q= p\cdot \alpha\cdot\lVert \mathbf{x}\rVert_\infty\cdot\lambda^{\omega(\log \lambda)}$. Instead of encoding $\mathbf{y}$, one encodes $q/p \cdot \mathbf{y}$ and sends along a random PRG seed $\mathsf{seed}$. Then, we return \[\mathbf{v}:=\Bigl\lceil\Hash\Eval(\pk, E, \psi)+\mathsf{PRG}(\mathsf{seed})\Bigr\rfloor_p, \qquad\text{and}\qquad \mathbf{w}:=\Bigl\lceil\Enc\Eval(\pk, d, \phi)-\mathsf{PRG}(\mathsf{seed})\Bigr\rfloor_p.\] 
By the $\alpha$-correctness guarantee of the protocol, it holds that:
\[
\Hash\Eval(\pk, E, \psi) + \Enc\Eval(\pk, d, \phi) = q/p \cdot \mathbf{x} \otimes \mathbf{y} +\mathbf{e}, \qquad \qquad \lVert\mathbf{e}\rVert_\infty\le \alpha\lVert \mathbf{x}\rVert_\infty.
\]
Moreover, by the PRG security, $\mathbf{v}$ is indistinguishable to a random element in $\Z_q^{m\ell}$, independent of $\mathbf{x}$, $\mathbf{y}$ and the public parameters. Following \cite{C:DHRW16}, with overwhelming probability over the randomness of $\mathsf{seed}$, we obtain
\begin{align*}
\mathbf{v}+\mathbf{w}&=\Bigl\lceil\Hash\Eval(\pk, E, \psi)+\mathsf{PRG}(\mathsf{seed})+\Enc\Eval(\pk, d, \phi)-\mathsf{PRG}(\mathsf{seed})\Bigr\rfloor_p\\
&=\Bigl\lceil q/p \cdot \mathbf{x} \otimes \mathbf{y}+\mathbf{e}\Bigr\rfloor_p\\
&= \lceil q/p \cdot \mathbf{x} \otimes \mathbf{y}\rfloor_p\\
&=  \mathbf{x} \otimes \mathbf{y} \quad(\bmod p).
\end{align*}
Notice that with these modifications, the protocol outputs the correct output with overwhelming probability, and this holds even if the inputs are chosen by the adversary after seeing the public parameters.}

\paragraph{Encoder Privacy.} We define encoder privacy, which guarantees that the encoding is simulatable, without knowing the underlying vector. This is formalized as follows.

\begin{definition}[Encoder Privacy]
Consider the following experiment $\mathsf{EncExp}_{\mathcal{A}, \Sim}(\bbbone^\sec)$ parametrized by an adversary $\mathcal{A} = (\mathcal{A}_0, \allowbreak \mathcal{A}_1)$ and a simulator $\Sim$:
\begin{itemize}
    \item Sample a public key $\pk\sample\Setup(\bbbone^\sec)$.
    \item Activate the adversary $(\mathbf{y}, \aux)\sample\Adv_0(\bbbone^\sec, \pk)$.
    \item Sample a random bit $b\sample\{0, 1\}$. 
    \item If $b = 0$ compute $(E_0, \phi_0)\sample\Enc(\pk, \mathbf{y})$, else compute 
    $E_1\sample\Sim(\bbbone^\sec, \pk)$.
    \item Compute $b' \gets \Adv_1(E_b, \aux)$.
    \item Return $1$ if and only if $b =b'$.
\end{itemize}
A non-interactive OTE scheme $(\Setup, \allowbreak \Hash, \allowbreak \Enc, \allowbreak \Hash\Eval, \allowbreak \Enc\Eval)$ is encoder-private if there exists a PPT simulator $\Sim$, such that for every PPT adversary $\mathcal{A}$, there exists a negligible function $\negl(\sec)$ such that, for every $\sec\in\N$, we have that:
\[
\left|\frac{1}{2} - \Pr\left[\mathsf{EncExp}_{\mathcal{A}, \Sim}(\bbbone^\sec) = 1\right] \right|\le\negl(\sec).
\]
\end{definition}

{
\paragraph{Hasher Privacy.} We also define hasher privacy, which guarantees that the hash is simulatable, without knowing the underlying vector. This is formalized as follows.}
{
\begin{definition}[Encoder Privacy]
Consider the following experiment $\mathsf{HashExp}_{\mathcal{A}, \Sim}(\bbbone^\sec)$ parametrized by an adversary $\mathcal{A} = (\mathcal{A}_0, \allowbreak \mathcal{A}_1)$ and a simulator $\Sim$:
\begin{itemize}
    \item Sample a public key $\pk\sample\Setup(\bbbone^\sec)$.
    \item Activate the adversary $(\mathbf{x}, \aux)\sample\Adv_0(\bbbone^\sec, \pk)$.
    \item Sample a random bit $b\sample\{0, 1\}$. 
    \item If $b = 0$ compute $(d_0, \psi_0)\sample\Hash(\pk, \mathbf{x})$, else compute 
    $d_1\sample\Sim(\bbbone^\sec, \pk)$.
    \item Compute $b' \gets \Adv_1(d_b, \aux)$.
    \item Return $1$ if and only if $b =b'$.
\end{itemize}
A non-interactive OTE scheme $(\Setup, \allowbreak \Hash, \allowbreak \Enc, \allowbreak \Hash\Eval, \allowbreak \Enc\Eval)$ is hasher-private if there exists a PPT simulator $\Sim$, such that for every PPT adversary $\mathcal{A}$, there exists a negligible function $\negl(\sec)$ such that, for every $\sec\in\N$, we have that:
\[
\left|\frac{1}{2} - \Pr\left[\mathsf{HashExp}_{\mathcal{A}, \Sim}(\bbbone^\sec) = 1\right] \right|\le\negl(\sec).
\]
\end{definition}}

\subsection{Half-Succinct NI-OTE from LWE}\label{sec:halfNIOTE}

We present our first NI-OTE protocol that only satisfies a weak form of succinctness, namely that only the size of the digest is sublinear in the size of the hasher's input. Let $n:= k\cdot \log q$, where $k=\Theta(\sec)$, and suppose that for convenience that $q$ is a power of $2$. Our protocol is described in Construction~\ref{prot:HCVOLE}.

\begin{protocol}[Half-Succinct NI-OTE]\label{prot:HCVOLE}

\begin{description}
\item[$\Setup(1^\secpar)$:] Sample $\mathbf{A} \sample \Z_2^{n \times m}$ and $\mathbf{B} \sample \Z_2^{n\times (m \cdot \ell\cdot n)}$, then return $\pk:=(\mathbf{A}, \mathbf{B})$.

\item[$\Hash(\pk, \mathbf{x})$:] Compute $\mathbf{d}:= \mathbf{A}\cdot \mathbf{x}$ and return $d:=\mathbf{d}$ and $\psi:=\mathbf{x}$.

\item[$\Enc(\pk, \mathbf{y})$:] Sample $\mathbf{s}\sample \Z_q^k$, $\hat{\mathbf{E}}\sample\chi(\sec)$, 
and $\overline{\mathbf{E}}\sample\chi(\sec)$. Compute 
\[\mathbf{C}:= \mathbf{A}^\intercal \cdot \mathbf{B}\cdot (\id_{m\cdot \ell}\otimes \mathbf{s}\otimes \mathbf{g}_q)+ \mathbf{A}^\intercal\cdot \hat{\mathbf{E}}+ \overline{\mathbf{E}}+ \id_{m} \otimes \mathbf{y}^\intercal.\]
Return $E:=\mathbf{C}$ and $\phi:=\mathbf{s}$.

\item[$\Hash\Eval(\pk, E, \psi)$:] Return $\mathbf{v}:=\mathbf{C}^\intercal\cdot \mathbf{x}$.

\item[$\Enc\Eval(\pk, d, \phi)$:] Return $\mathbf{w}:= -(\id_{m\cdot \ell}\otimes \mathbf{s}^\intercal\otimes \mathbf{g}_q^\intercal) \cdot \mathbf{B}^\intercal\cdot \mathbf{d}$.
\end{description}
\end{protocol}
The scheme is trivially programmable, since $\phi$ consists of a uniformly sampled vector $\mathbf{s}$ that in particular is independent from $\mathbf{y}$.
Then, we observe that the scheme has indeed a bilinear encoding evaluation algorithm. By Lemma \ref{thm:linearisation}, it holds that:
\begin{equation}\label{eq:bilinear}
-(\id_{m\cdot \ell}\otimes \mathbf{s}^\intercal\otimes \mathbf{g}_q^\intercal) \cdot \mathbf{B}^\intercal\cdot \mathbf{d} 
=
\mathsf{Lin}(-\mathbf{B}^\intercal) (\mathbf{d}\otimes \mathbf{s}\otimes \mathbf{g}_q)
\end{equation}
with $\lVert \mathsf{Lin}(-\mathbf{B}^\intercal)\rVert_\infty=\lVert -\mathbf{B}^\intercal\rVert_\infty\le 1$. We can also bound the norm of the digest by 
\begin{equation}\label{eq:bound_d}
    \lVert \mathbf{d}\rVert_\infty=\lVert \mathbf{A}\cdot \mathbf{x}\rVert_\infty\le m\cdot \lVert \mathbf{A}\rVert_\infty\cdot \lVert \mathbf{x}\rVert_\infty\le m\cdot \lVert \mathbf{x}\rVert_\infty
\end{equation}
with a triangle inequality. Finally, it is easy to see that the scheme is half succinct, since the hash $\mathbf{d}$ is $n$-dimensional vector over $\Z_q$, whose size is in particular independent of the length of $\mathbf{x}$ and $\mathbf{y}$. On the other hand, the encoding consists of an $m\times (m\cdot \ell)$ matrix over $\Z_q$. 

Next, we argue that the scheme satisfies approximate correctness. Indeed, let us rewrite
\begin{align*}
    \mathbf{C}&=\mathbf{A}^\intercal \cdot \mathbf{Z}+ \mathbf{\widetilde{E}}+ \id_{m} \otimes \mathbf{y}^\intercal,
\end{align*}
where $\mathbf{Z}:=\mathbf{B}\cdot (\id_{m\cdot \ell}\otimes \mathbf{s}\otimes \mathbf{g}_q)$ and $\mathbf{\widetilde{E}}:=\mathbf{A}^\intercal\cdot \hat{\mathbf{E}}+ \overline{\mathbf{E}}$. Substituting,  we obtain
\begin{align*}
\mathbf{v}+\mathbf{w}&=\mathbf{C}^\intercal \cdot \mathbf{x}-\mathbf{Z}^\intercal \cdot \mathbf{d}\\
&=(\mathbf{Z}^\intercal \cdot \mathbf{A})\cdot \mathbf{x}+\mathbf{\widetilde{E}}^\intercal\cdot \mathbf{x}+(\id_{m} \otimes \mathbf{y})\cdot \mathbf{x}-\mathbf{Z}^\intercal \cdot \mathbf{d}\\
&=\mathbf{Z}^\intercal \cdot \mathbf{d}+\mathbf{\widetilde{E}}^\intercal\cdot \mathbf{x}+\mathbf{x} \otimes \mathbf{y}-\mathbf{Z}^\intercal \cdot \mathbf{d}\\
&=\mathbf{x} \otimes \mathbf{y}+\mathbf{\widetilde{E}}^\intercal\cdot \mathbf{x}.
\end{align*}
Since $\mathbf{A}$ is a matrix in $\Z_2$ with $n$ rows, the entries of $\mathbf{\widetilde{E}}=\mathbf{A}^\intercal\cdot \hat{\mathbf{E}}+\overline{\mathbf{E}}$ is obtained by adding at most $n+1$ entries of the vectors $\overline{\mathbf{E}}$ and $\hat{\mathbf{E}}$, which are $B(\sec)$-bounded. In other words, \[\left\lVert \mathbf{A}^\intercal \cdot \hat{\mathbf{E}}+\overline{\mathbf{E}}\right\rVert_\infty\le (n+1)\cdot B(\sec).\] 
Therefore we can bound the correctness error by 
\begin{equation}\label{eq:alpha_correct}
\left\lVert\mathbf{\widetilde{E}}^\intercal\cdot \mathbf{x}\right\rVert_\infty\le m\cdot (n+1) \cdot B(\sec)\cdot \lVert \mathbf{x}\rVert_\infty.
\end{equation}
with a triangle inequality. Finally, we show that the scheme satisfies encoder privacy.

\begin{theorem}[Encoder Privacy]
    Assuming the hardness of LWE, Construction~\ref{prot:HCVOLE} satisfies encoder privacy.
\end{theorem}
\begin{proof}
    Consider the following sequence of hybrids.
    \begin{itemize}
        \item Hybrid $\mathcal{H}_0$: This is the original distribution.
        \item Hybrid $\mathcal{H}_1$: We define 
        \[
        \mathbf{C} := \mathbf{A}^\intercal \cdot \mathbf{U}+\overline{\mathbf{E}}+(\id_{m}\otimes \mathbf{y}^\intercal)
        \]
        where $\mathbf{U}\sample\Z_q^{n\times (m\cdot \ell)}$.

        We claim that this hybrid is indistinguishable from the previous once, by the LWE assumption. To see why, observe that the encoding provided to the adversary in the previous hybrid equals:
    \begin{align*}\mathbf{C}&=\mathbf{A}^\intercal \cdot \mathbf{U}+\overline{\mathbf{E}}+(\id_{m}\otimes \mathbf{y}^\intercal),
\end{align*}
where $\mathbf{U}:=\mathbf{B}\cdot (\id_{m\cdot \ell}\otimes \mathbf{s}\otimes\mathbf{g}_q)+ \hat{\mathbf{E}}$ for $\hat{\mathbf{E}}\sample{\chi}(\sec)$ and $\overline{\mathbf{E}}\sample\chi(\sec)$. This last term consists of an LWE sample with secret $\mathbf{s}$ and public matrix obtained by splitting $\mathbf{B}\cdot(\id_{m\cdot \ell\cdot k}\otimes \mathbf{g}_q)$ in $m\cdot \ell$ blocks in $\Z_2^{n\times k}$ and stacking them one underneath the other. 

    In more details, consider a reduction that receives a matrix $\mathbf{M}\in\Z_q^{( n\cdot m\cdot \ell)\times k}$ and a vector $\mathbf{u}\in\Z_q^{n\cdot m\cdot \ell}$ that is either an LWE sample with respect to $\mathbf{M}$ or a uniformly random vector. The reduction splits the rows of $\mathbf{M}$ into $\ell\cdot m$ blocks in $\Z_q^{n\times k}$, denoted by
    \[
    \mathbf{M} = \begin{pmatrix}   
    \mathbf{M}_0\\
    \vdots\\
    \mathbf{M}_{m\cdot \ell-1}
    \end{pmatrix}.
    \]
    For every $j\in[m\cdot \ell]$, set $\mathbf{B}_j:= \bigl(\mathbf{G}^{-1}(\mathbf{M}_j^\intercal)\bigr)^\intercal$. Notice that $\mathbf{B}_j$ is a random matrix over $\Z_2^{n\times n}$, since $\mathbf{M}_j$ is uniformly sampled and $q$ is a power of 2. Finally, construct $\mathbf{B}$ as 
    \[
    \mathbf{B} = (\mathbf{B}_0, \dots, \mathbf{B}_{m\cdot\ell -1})
    \]
    and generate the columns in $\mathbf{U}$ by splitting $\mathbf{u}$ into $m\cdot \ell$ vectors in $\Z_q^n$. Then set $\mathbf{C}$ to $\mathbf{A}^\intercal \cdot \mathbf{U}+\overline{\mathbf{E}}+(\id_{m}\otimes \mathbf{y}^\intercal)$.

        \item Hybrid $\mathcal{H}_2$: We sample $\mathbf{C}\sample\Z_q^{(m)\times (m\cdot \ell)}$.

        Indistinguishability from the previous hybrid follows by another reduction to the LWE problem. Indeed, notice that $\mathbf{A}^\intercal\cdot \mathbf{U}+\overline{\mathbf{E}}$ is a batch of $m\cdot \ell$ LWE samples with $\mathbf{A}^\intercal$ as public matrix. The LWE secrets consist of the columns of $\mathbf{U}$.
    \end{itemize}
    The proof is concluded by defining the simulator $\Sim$ to output a uniformly random $\mathbf{C}\sample\Z_q^{(m)\times (m\cdot \ell)}$ as the encoding.
\end{proof}

\subsection{Bootstrapping to Fully-Succinct NI-OTE}\label{sec:NIOTE}

We now show a bootstrapping procedure to turn the NI-OTE constructed in Section~\ref{sec:halfNIOTE} into a fully succinct NI-OTE for $\Z_q^{m} \otimes \Z_q^\ell$. We assume that $m := t^r \cdot n$, for some $t,r\in \Z$, and we assume the existence of a half-succinct $\alpha$-correct $\BOTE=(\Setup, \allowbreak \Hash, \allowbreak \Enc, \allowbreak \Hash\Eval, \allowbreak \Enc\Eval)$ for  $\Z_q^{t\cdot n}\otimes \Z_q^n$ with bilinear encoder evaluation, where $d \in \Z_q^n$ and $\phi \in \Z_q^k$, with $n:= k \cdot \log q$. We describe our scheme in Construction~\ref{prot:SOTE}.

\begin{protocol}[Fully Succinct NI-OTE]\label{prot:SOTE}

\begin{description}
\item[$\Setup(1^\secpar)$:] Return $\pk\sample\BOTE.\Setup(\bbbone^\sec)$.

\item[$\Hash(\pk, \mathbf{x})$:] Set $\mathbf{x}_0 := \mathbf{x}$. Then for every $i\in[r]$ proceed as follows:
\begin{itemize}
    \item Parse $\mathbf{x}_i$ as the vertical concatenation of $\mathbf{x}_{i,0},\dots, \mathbf{x}_{i,t^{r-i-1}-1}$ where $\mathbf{x}_{i,j}\in\Z_q^{t\cdot n}$.
    \item For every $j\in[t^{r-i-1}]$, compute 
    \[
    (\mathbf{d}_{i,j}, \psi_{i,j})\sample \BOTE.\Hash(\pk, \mathbf{x}_{i,j}).
    \]
    \item Define $\mathbf{x}_{i+1}$ to be the vertical concatenation of the $\mathbf{d}_{i,j}$.
\end{itemize}
Return $d:=\mathbf{x}_r$ and $\psi:=\{\psi_{i,j}\}_{i\in[r], j\in[t^{r-i-1}]}$.

\item[$\Enc(\pk, \mathbf{y})$:] Set $\mathbf{y}_0 := \mathbf{y}$, then for $i \in [r]$, compute
\[
(E_i, \phi_{i+1})\sample\Enc(\pk, \mathbf{y}_i) \quad\text{and}\quad \mathbf{y}_{i+1}:= \phi_{i+1}\otimes \mathbf{g}_q
\]
Output $E:=\{E_i\}_{i\in[r]}$ and $\phi:=\phi_r$. 

\item[$\Hash\Eval(\pk, E, \psi)$:] Let $\mathbf{P}$ be the public matrix derived from $\pk$. For all $i\in[r]$ and $j\in [t^{r-i-1}]$, compute
\[
\mathbf{P}_i:= (\id_{t^{r-i}}\otimes \mathbf{P})\quad\text{and}\quad
\mathbf{v}_{i,j}:= \BOTE.\Hash\Eval(\pk, E_i, \psi_{i,j}).
\]
Let $\mathbf{v}_i$ be the vertical concatenation of $\{\mathbf{v}_{i,j}\}_{j\in[t^{r-i-1}]}$. Return 
$\mathbf{v}:=\sum_{i\in[r]} \left(\prod_{k=1}^{i} \mathbf{P}_k\right)\cdot \mathbf{v}_i$.

\item[$\Enc\Eval(\pk, d, \phi)$:] Let $\mathbf{P}_i$ be defined as above. Return  
\[\mathbf{w}:=  \left(\prod_{i=1}^{r-1}\mathbf{P}_{i}\right)\cdot \BOTE.\Enc\Eval(\pk, d, \phi_r).\]
\end{description}
\end{protocol}

The scheme is clearly programmable, if the underlying $\BOTE$ is.
Furthermore, note that the scheme maintains a bilinear encoder evaluation, since, by Equation (\ref{eq:bilinear}) we have that:
\begin{equation}\label{eq:bilinear2}
    \Enc\Eval(\pk, d, \phi)=\left(\prod_{i=1}^{r-1}\mathbf{P}_i\right)\cdot \mathbf{P} \cdot (d\otimes \phi\otimes \mathbf{g}_q).
\end{equation}
To bound the norm of the public matrix, recall that  $\mathbf{P}\in\Z_q^{n^2\times (t\cdot n^2)}$ and $\lVert \mathbf{P}\rVert_\infty\le 1$. 
So, for every $i\in[r]$, each row of $\mathbf{P}_i=(\id_{t^{r-i}}\otimes \mathbf{P})$ has at most $n^2$ non-zero entries. Thus, multiplying by $\mathbf{P}_i$ increases the infinity norm at most by a factor $n^2$. Thus we have that:
\begin{equation}\label{eq:matrix_norm2}
    \left\lVert \prod_{i=1}^{r-1}\mathbf{P}_i \cdot \mathbf{P} \right\rVert_\infty \leq n^{2r}.
\end{equation}
On the other hand, we can bound the norm of the digest by recalling that $\mathbf{x}_{i,j}\in\Z_{q}^{t\cdot n}$ and by Equation (\ref{eq:bound_d}) the norm of the digest is increased by a factor $t\cdot n$ each time it is hashed. Thus:
\begin{equation}\label{eq:bound_d2}
    \lVert d\rVert_\infty = \lVert \mathbf{x}_r\rVert_\infty\leq (tn)^r\cdot \lVert \mathbf{x}\rVert_\infty.
\end{equation}
Furthermore, the dimension of the digest is $n$, and therefore independent of the dimensions of the input vector, and its bit-length is at most $n\cdot (r \cdot \log tn +\log  \lVert \mathbf{x}\rVert_\infty)$. The size of the encoding is $r=O(\log m)$ times the size of an encoding of $\BOTE$ (whose size is independent of $m$). Thus, the scheme is fully succinct.

\paragraph{Approximate Correctness.} Towards proving approximate correctness, we begin by observing that, by the $\alpha$-correctness of  $\BOTE$, for every $i\in[r]$ and $j\in[t^{r-i-1}]$, we have:
\begin{equation}\label{basic}
\Enc\Eval(\pk, \mathbf{d}_{i,j}, \phi_{i+1})+\Hash\Eval(\pk, E_i, \psi_{i,j})=\mathbf{x}_{i,j}\otimes \mathbf{y}_i+\mathbf{e}_{i,j},
\end{equation}
where $\lVert \mathbf{e}_{i,j}\rVert_\infty\le \alpha\cdot \lVert \mathbf{x}_{i,j}\rVert_\infty \le \alpha\cdot \lVert \mathbf{x}_i\rVert_\infty \leq tn(n+1) B(\sec)\cdot \lVert \mathbf{x}_i\rVert_\infty$, by Equation (\ref{eq:alpha_correct}). To establish correctness, it suffices to prove the following.

\begin{lemma}[Approximate Correctness]\label{lmm:correct}
For every $i\in[r]$:
\[\sum_{j=i}^{r-1} \prod_{k=1}^{j} \mathbf{P}_k\cdot\mathbf{v}_j+\mathbf{w}=\left(\prod_{k=1}^{i} \mathbf{P}_k\right)\cdot(\mathbf{x}_{i}\otimes \mathbf{y}_i)+\widetilde{\mathbf{e}}_i\]
where $\lVert \widetilde{\mathbf{e}}_i\rVert_\infty\le \alpha \cdot \left(\sum_{k=i}^{r-1} (t \cdot n^3)^k\right)\cdot \lVert \mathbf{x}\rVert_\infty$.
\end{lemma}
Before proceeding with the proof, we can indeed see that, setting $i=0$, we obtain:
\[\mathbf{v}+\mathbf{w}=(\mathbf{x}\otimes \mathbf{y})+\mathbf{e}\]
where $\lVert\mathbf{e}\rVert_\infty\le \alpha\cdot (t\cdot n^3)^r\cdot \lVert \mathbf{x}\rVert_\infty$.
Thus, the scheme satisfies $\alpha(t\cdot n^3)^r$-correctness.

\begin{proof}
We proceed by induction starting from $i=r-1$ and going all the way down to $i=0$.
It is easy to see that $\mathbf{v}_{r-1}= \BOTE.\Hash\Eval(\pk, E_{r-1}, \psi_{r-1,0})$, whereas $d=\mathbf{d}_{r-1,0}$. Therefore, by Equation (\ref{basic}), we obtain:
\begin{align*}\mathbf{v}_{r-1}\cdot \prod_{k=1}^{r-1} \mathbf{P}_k+\mathbf{w}&=\left(\prod_{k=1}^{r-1} \mathbf{P}_k\right)\cdot \Bigl(\BOTE.\Hash\Eval(\pk, E_{r-1}, \psi_{r-1,0})\\
&+\BOTE.\Enc\Eval(\pk, d, \phi_r)\Bigr)\\
&=\left(\prod_{k=1}^{r-1} \mathbf{P}_k\right)\cdot (\mathbf{x}_{r-1,0}\otimes \mathbf{y}_{r-1})+\left(\prod_{k=1}^{r-1} \mathbf{P}_k\right)\cdot\mathbf{e}_{r-1,0}\\
&=\left(\prod_{k=1}^{r-1} \mathbf{P}_k\right)\cdot(\mathbf{x}_{r-1}\otimes \mathbf{y}_{r-1})+\widetilde{\mathbf{e}}_{r-1}\end{align*}
where $\widetilde{\mathbf{e}}_{r-1}:=\left(\prod_{k=1}^{r-1} \mathbf{P}_k\right)\cdot\mathbf{e}_{r-1,0}$. Notice that 
\begin{align*}\lVert\widetilde{\mathbf{e}}_{r-1}\rVert_\infty&\le n^{2(r-1)}\cdot \lVert \mathbf{e}_{r-1,0}\rVert_\infty\\
&\le \alpha n^{2(r-1)}\cdot \lVert \mathbf{x}_{r-1}\rVert_\infty\\
&\le \alpha (t n^3)^{r-1}\cdot \lVert \mathbf{x}\rVert_\infty\end{align*}
by Equation (\ref{eq:matrix_norm2}) and Equation (\ref{eq:bound_d2}).
This establishes the base case. For the inductive step, we first appeal to the induction hypothesis to show that:
\begin{align*}\sum_{j=i-1}^{r-1}\left(\prod_{k=1}^{j} \mathbf{P}_k\right)\cdot\mathbf{v}_{j}+\mathbf{w}&=\left(\prod_{k=1}^{i-1} \mathbf{P}_k\right)\cdot \mathbf{v}_{i-1}+\left(\prod_{k=1}^{i} \mathbf{P}_k\right)\cdot(\mathbf{x}_{i}\otimes \mathbf{y}_i)+\widetilde{\mathbf{e}}_i\\
&=\left(\prod_{k=1}^{i-1} \mathbf{P}_k\right)\cdot\Bigl(\mathbf{v}_{i-1}+\mathbf{P}_i\cdot(\mathbf{x}_{i}\otimes \mathbf{y}_i)\Bigr)+\widetilde{\mathbf{e}}_i.\end{align*}
 Moreover:
\begin{align*}
\mathbf{P}_i\cdot (\mathbf{x}_i\otimes \mathbf{y}_i)&=\begin{pmatrix} \mathbf{P}\cdot (\mathbf{d}_{i-1,0}\otimes \phi_i\otimes \mathbf{g}_q)\\
\mathbf{P}\cdot (\mathbf{d}_{i-1,1}\otimes \phi_i\otimes \mathbf{g}_q)\\
\vdots\\
\mathbf{P}\cdot (\mathbf{d}_{i-1,t^{r-i}-1}\otimes \phi_i\otimes \mathbf{g}_q)
    \end{pmatrix}\\
    &=\begin{pmatrix} \BOTE.\Enc\Eval(\pk, \mathbf{d}_{i-1,0}, \phi_{i})\\
\BOTE.\Enc\Eval(\pk, \mathbf{d}_{i-1,1}, \phi_{i})\\
\vdots\\
\BOTE.\Enc\Eval(\pk, \mathbf{d}_{i-1,t^{r-i}-1}, \phi_{i})
    \end{pmatrix}\end{align*}
We also observe  that:
\[\mathbf{v}_{i-1}=\begin{pmatrix} \mathbf{v}_{i-1,0}\\
\mathbf{v}_{i-1,1}\\
\vdots\\
\mathbf{v}_{i-1,t^{r-i}-1}
    \end{pmatrix}=\begin{pmatrix} \BOTE.\Enc\Hash(\pk, E_{i-1}, \psi_{i-1,0})\\
\BOTE.\Enc\Hash(\pk, E_{i-1}, \psi_{i-1,1})\\
\vdots\\
\BOTE.\Enc\Hash(\pk, E_{i-1}, \psi_{i-1,t^{r-i}-1})
    \end{pmatrix}.\]
Furthermore, by Equation \eqref{basic}, we have:
\begin{align*}
    \mathbf{v}_{i-1}+\mathbf{P}_i\cdot(\mathbf{x}_{i}\otimes \mathbf{y}_i)&=\begin{pmatrix} (\mathbf{x}_{i-1,0}\otimes \mathbf{y}_{i-1})+\mathbf{e}_{i-1,0}\\
(\mathbf{x}_{i-1, 1}\otimes \mathbf{y}_{i-1})+\mathbf{e}_{i-1,1}\\
\vdots\\
(\mathbf{x}_{i-1, t^{r-i}}\otimes \mathbf{y}_{i-1})+\mathbf{e}_{i-1,t^{r-i}-1}
    \end{pmatrix}\\
    &=(\mathbf{x}_{i-1}\otimes \mathbf{y}_{i-1}) +\mathbf{e}_{i-1}, \end{align*}
 where $\mathbf{e}_{i-1}^\intercal:=(\mathbf{e}_{i-1,0}^\intercal, 
\mathbf{e}_{i-1,1}^\intercal, 
\dots, 
\mathbf{e}_{i-1,t^{r-i}-1}^\intercal)$.
Notice that $\lVert \mathbf{e}_{i-1}\rVert_\infty\le \alpha\cdot \lVert \mathbf{x}_{i-1}\rVert_\infty\le \alpha\cdot (tn)^{i-1}\cdot \lVert \mathbf{x}\rVert_\infty$ by Equation (\ref{eq:matrix_norm2}) and Equation (\ref{eq:bound_d2}). We conclude that
\[\sum_{j=i-1}^{r-1} \left(\prod_{k=1}^{j} \mathbf{P}_k\right)\cdot\mathbf{v}_j+\mathbf{w}=\left(\prod_{k=1}^{i-1} \mathbf{P}_k\right)\cdot(\mathbf{x}_{i-1}\otimes \mathbf{y}_{i-1})+\widetilde{\mathbf{e}}_{i-1}\]
where $\widetilde{\mathbf{e}}_{i-1}:=\widetilde{\mathbf{e}}_{i}+\left(\prod_{k=1}^{i-1} \mathbf{P}_k\right)\cdot \mathbf{e}_{i-1}$.
Observe that \[\lVert\widetilde{\mathbf{e}}_{i-1}\rVert_\infty\le \lVert\widetilde{\mathbf{e}}_{i}\rVert_\infty+\alpha\cdot (t\cdot n^3)^{i-1}\cdot\lVert\mathbf{x}\rVert_\infty\le \alpha\cdot \left(\sum_{k=i-1}^{r-1} (t \cdot n^3)^k\right)\cdot \lVert \mathbf{x}\rVert_\infty\]
as desired.
\end{proof}

\paragraph{Encoder Privacy.} The following theorem establishes encoder privacy.

\begin{theorem}
    Assuming that $\BOTE$ is encoder private, Construction~\ref{prot:SOTE} is encoder private.
\end{theorem}
\begin{proof}
    The proof follows by a standard hybrid argument, where we gradually substitute the encodings $\{E_j\}_{j\in[r]}$ with the outputs of the simulator $\Sim(1^\sec, \pk)$, provided by the definition of encoder privacy of $\BOTE$.
\end{proof}

\paragraph{How to Build a Fully Succinct OTE for $\Z_q^m\otimes \Z_q^\ell$.}
Construction \ref{prot:SOTE} presents a non-interactive OTE for $\Z_q^m\otimes \Z_q^n$ where the hash size is a vector in $\Z_q^n$ and the encoding consists of $r$ matrices of size $(tn)\times (tn^2)$. A trivial way to build a fully succinct OTE for $\Z_q^m\otimes \Z_q^\ell$ is to instantiate Construction \ref{prot:SOTE} with $n:=\ell$. If $\ell=O(\sec\cdot \log q)$, this is secure, however, the encoder size would scale as $r\cdot \ell^3$ while the hash size would be linear in $\ell$. We can do better: just split $\mathbf{y}$ into blocks $\mathbf{y}_0, \dots, \mathbf{y}_N$ of size $n$. Then, compute an OTE encoding of each of them using Construction \ref{prot:SOTE} and send them to the hasher. Given the hasher's digest (notice, a single digest is sufficient for all $\mathbf{y}_0, \dots, \mathbf{y}_N$), we can then derive shares for $\mathbf{x}\otimes \mathbf{y}_0, \dots, \mathbf{x}\otimes \mathbf{y}_N$. By reordering these, we obtain a secret-sharing of $\mathbf{x}\otimes \mathbf{y}$. 

{\paragraph{Parameters.} Let $q=q(\lambda)$ be a power of $2$, let $\chi=\chi(\lambda)$ be a
$B(\lambda)$-bounded error distribution, and let $k=k(\lambda)$ be such that
the $(\chi,k,m,q)$-LWE problem is hard. Set $n := k \log q $ and fix integers $t=t(\lambda)\ge 2$ and $r=r(\lambda)\ge 1$ with $m = t^r \cdot n$. Instantiate Construction~\ref{prot:SOTE} using as underlying
half-succinct NI-OTE the scheme of Construction~\ref{prot:HCVOLE}
for $\mathbb Z_q^{t n} \otimes \mathbb Z_q^{n}$, we obtain the following parameters:
\begin{itemize}
    \item (Public Key) The public is exactly the CRS of the underlying half-succinct scheme, which consists of $\mathbf{A} \in \mathbb Z_2^{n\times t n}$ and $\mathbf{B} \in \mathbb Z_2^{n\times t n^3}$.
    \item (Digest) The digest is a vector in $\mathbb Z_q^{n}$ of $\ell_\infty$-norm at most $m^{r}\cdot \lVert\mathbf{x}\rVert_\infty$.
    \item (Encoding) Write $N := \lceil \ell/n \rceil$ and partition $\mathbf y$ into $N$ blocks of length at most $n$. The encoding consists of $N$ independent encodings
    produced by Construction~\ref{prot:SOTE}, each such encoding
    is a tuple of $r$ matrices of size $(t n)\times (t n^2)$. Therefore the
    total encoding size is
    $\ell  r  t^2 n^2$
     elements of $\mathbb Z_q$.
\end{itemize}
Moreover, Construction~\ref{prot:HCVOLE} on inputs
    $\mathbb Z_q^{t n} \otimes \mathbb Z_q^n$ is $t n (n+1) B(\lambda)$-correct. Hence, by Lemma~\ref{lmm:correct}, 
    Construction~\ref{prot:SOTE} is
    $\alpha$-correct with
    \[
    \alpha:=\alpha_0 \cdot (t n^3)^r
    = t n (n+1) B(\lambda) (t n^3)^r
    \]
    Thus it suffices to set $q/p > 2\alpha$ to obtain a perfectly correct scheme. 
For $t=\Theta(1)$, then $r=\Theta(\log m)$, the digest size and CRS size are
$\poly(\lambda,\log m)$, and the encoding size is
    $\ell \cdot \poly(\lambda,\log m)$. Thus, we obtain the following theorem.}

{
\begin{theorem}[Fully-Succinct NI-OTE]
Assume the hardness of LWE, there exists an
encoder-private programmable NI-OTE for
$\mathbb Z_q^m \otimes \mathbb Z_q^\ell$ with digest and public key size $\poly(\lambda,\log m)$ and encoding size $\ell \cdot \poly(\lambda,\log m)$.
\end{theorem}
}
{
\subsection{How to Achieve Hasher-Privacy}
The constructions presented so far rely on a deterministic hasher, hence, they are not hasher private. There is however an easy way to achieve this property by relying on a non-succinct NI-OTE protocol that provides privacy at both ends. The idea is to compose the latter with the fully succinct OTE to set up a (possibly noisy) additive secret-sharing of $d\otimes \phi\otimes \mathbf{g}_q$. Specifically, instead of sending $d$ in the clear, the hasher inputs $d$ in the hasher-private non-succinct NI-OTE. Simultaneously, the encoder inputs $\phi\otimes \mathbf{g}_q$ and sends the fully succinct OTE encoding $E$ along with it. Using $\psi$ and $E$, the hasher can derive the output share of the fully succinct OTE as usual. Moreover, by relying on the fact that the encoder output in the fully succinct OTE is just
\[\mathbf{w}:=\left(\prod_{i=1}^{r-1}\mathbf{P}_i\right)\cdot \mathbf{P} \cdot (d\otimes \phi\otimes \mathbf{g}_q),\]
where each $\mathbf{P}_i$ is  a public low-norm matrix, the parties can bootstrap the (noisy) secret sharing of $d\otimes \phi\otimes \mathbf{g}_q$ into a (noisy) secret-sharing of $\mathbf{w}$. Putting all these shares together provides a (noisy) secret-sharing of $\mathbf{x}\otimes \mathbf{y}$.}

{
Notice that the hasher-privacy of the non-succinct OTE guarantees the privacy of $\mathbf{x}$, whereas the encoder privacy of both the non-succinct and the fully-succinct OTE ensure that $\mathbf{y}$ remains secret. Finally, the protocol remains fully succinct. Indeed, the non-succinct OTE contributes only with a $\poly(\sec, \lvert d\rvert, \lvert \phi\rvert)=\poly(\sec)$ additive overhead. Also the CRS size increases only by a $\poly(\sec)$ amount.}

{
\paragraph{A Non-Succinct OTE Protocol with Hasher and Encoder Privacy.} We sketch how to build an LWE-based non-succinct OTE protocol with privacy at both sides. This can be easily done by following the blueprint of \cite{EC:OrlSchYak21,C:ADOS22,EC:AbrRoySch24}.}

{
  \begin{protocol}[Fully Private, Non-Succinct NI-OTE]\label{prot:privateVOLE}
\begin{description}
\item[$\Setup(1^\secpar)$:] Sample $\mathbf{A} \sample \Z_q^{k \times m}$ and $\mathbf{B}\sample\Z_q^{k\times 2n}$ and return $\pk:=(\mathbf{A}, \mathbf{B})$.
\item[$\Hash(\pk, \mathbf{x})$:] Compute $\mathbf{d}:= \mathbf{A}\cdot \mathbf{x}+\mathbf{B}\cdot \mathbf{r}$ where $\mathbf{r}\sample\Z_2^{2n}$ and return $d:=\mathbf{d}$ and $\psi:=(\mathbf{x}, \mathbf{r})$.
\item[$\Enc(\pk, \mathbf{y})$:] Sample $\mathbf{S}\sample \Z_q^{k\times mn}$, $\mathbf{E}\sample\chi(\sec)$, 
and $\hat{\mathbf{E}}\sample\chi(\sec)$. Compute 
\begin{align*}\mathbf{C}&:= \mathbf{A}^\intercal \mathbf{S}+ \mathbf{E}+ \id_{m} \otimes \mathbf{y}^\intercal,\\
\hat{\mathbf{C}}&:= \mathbf{B}^\intercal \mathbf{S}+ \hat{\mathbf{E}}.\end{align*}
Return $E:=(\mathbf{C}, \hat{\mathbf{C}})$ and $\phi:=\mathbf{S}$.
\item[$\Hash\Eval(\pk, E, \psi)$:] Return $\mathbf{v}:=\mathbf{C}^\intercal\cdot \mathbf{x}+\hat{\mathbf{C}}^\intercal\cdot \mathbf{r}$.
\item[$\Enc\Eval(\pk, d, \phi)$:] Return $\mathbf{w}:= -\mathbf{S}^\intercal\cdot \mathbf{d}$.
\end{description}
\end{protocol}
Observe that in Construction \ref{prot:privateVOLE}, we have that
\begin{align*}
    \mathbf{w}+\mathbf{v}&=\mathbf{C}^\intercal\cdot \mathbf{x}+\hat{\mathbf{C}}^\intercal\cdot \mathbf{r}-\mathbf{S}^\intercal\cdot \mathbf{d}\\
    &=\mathbf{S}^\intercal\mathbf{A}\mathbf{x}+\mathbf{E}^\intercal\cdot \mathbf{x}+\mathbf{x}\otimes \mathbf{y}+\mathbf{S}^\intercal\mathbf{B}\mathbf{r}+\hat{\mathbf{E}}^\intercal\cdot \mathbf{r}-\mathbf{S}^\intercal (\mathbf{A}\mathbf{x}+\mathbf{B}\mathbf{r})\\
    &=\mathbf{x}\otimes \mathbf{y}+\mathbf{E}^\intercal\cdot \mathbf{x}+\hat{\mathbf{E}}^\intercal\cdot \mathbf{r},
\end{align*}
and $\lVert \mathbf{E}^\intercal\cdot \mathbf{x}+\hat{\mathbf{E}}^\intercal\cdot \mathbf{r}\rVert_\infty\le mB(\sec)\cdot \lVert \mathbf{x}\rVert_\infty+2nB(\sec)$.}

{
Hasher privacy is an immediate application of the leftover hash lemma, whereas encoder privacy derives from the fact that $\mathbf{A}^\intercal \mathbf{S}+ \mathbf{E}$ and $\mathbf{B}^\intercal \mathbf{S}+ \hat{\mathbf{E}}$ look random under LWE. Notice that while the hasher message is succinct, the encoder message has size $m^2n\log q+2mn^2\log q$ where $m$ is the length of the input $\mathbf{x}$. Finally, notice that we can easily modify this construction to achieve $\alpha$-approximate correctness for $\alpha=0$ using the same rounding trick we mentioned in Section \ref{defOTE}.}
\subsection{From Succinct NI-OTE to Succinct NI-MOLE}\label{sec:MOLE}

We will also consider an modification of NI-OTE where instead of computing a tensor product, we compute a matrix-vector product. More specifically, we let the hasher take as input a matrix $\mathbf{M}$ and the encoder a vector $\mathbf{y}$ and we require that
\[
\Hash\Eval(\pk, E, \psi)+\Enc\Eval(\pk, d, \phi) \approx \mathbf{M} \mathbf{y}.
\]
Furthermore, we require the protocol to be succinct, in the sense that the communication complexity should be independent of the number of rows of the matrix $\mathbf{M}$. We refer to this protocol as Non-Interactive Matrix Oblivious Linear Evaluation (NI-MOLE).

It is easy to see that one can generically construct an $(\ell\cdot \log q\cdot \alpha)$-correct NI-MOLE from an $\alpha$-correct NI-OTE: Simply hash $\mathbf{x} := \mathbf{G}^{-1}(\mathsf{vec}(\mathbf{M}))$ and encode $\mathbf{y}\otimes \mathbf{g}_q$. Then the hasher and the encoder return
\begin{gather*}
\mathsf{Lin}(\id_{m\cdot \ell\cdot \log q})\cdot\Hash\Eval(\pk, E, \psi), \\
\text{and}
\\
\mathsf{Lin}(\id_{m\cdot \ell\cdot \log q})\cdot\Enc\Eval(\pk, d, \phi)
\end{gather*}
respectively. By the correctness of the NI-OTE we have:
\begin{align*}
\mathbf{v}+\mathbf{w}&=\mathsf{Lin}(\id_{m\cdot \ell\cdot \log q})\cdot (\mathbf{x}\otimes \mathbf{y}\otimes \mathbf{g}_q)+\mathsf{Lin}(\id_{m\cdot \ell\cdot \log q})\cdot \mathbf{e}\\
&=(\id_{m}\otimes \mathbf{y}^\intercal\otimes \mathbf{g}_q^\intercal)\cdot \mathbf{x}+\mathsf{Lin}(\id_{m\cdot \ell\cdot \log q})\cdot \mathbf{e}\\
&=(\id_{m}\otimes \mathbf{y}^\intercal)\cdot \mathsf{vec}(\mathbf{M}^\intercal)+\mathsf{Lin}(\id_{m\cdot \ell\cdot \log q})\cdot \mathbf{e}\\
&=\mathbf{M}\cdot \mathbf{y}+\mathsf{Lin}(\id_{m\cdot \ell\cdot \log q})\cdot \mathbf{e}\\
&\approx\mathbf{M}\cdot \mathbf{y}
\end{align*}
ignoring low-order error terms, since $\lVert \mathsf{Lin}(\id_{m\cdot \ell\cdot \log q}) \rVert_\infty=1$ by Lemma \ref{thm:linearisation}.  Notice that \[\lVert \mathsf{Lin}(\id_{m\cdot \ell\cdot \log q})\cdot \mathbf{e}\rVert_\infty\le \ell\cdot\log q\cdot \alpha\cdot \lVert \mathbf{x}\rVert_\infty=\ell\cdot \log q\cdot \alpha,\] 
{as each row of $\mathsf{Lin}(\id_{m\cdot \ell\cdot \log q})$ has Hamming weight $\ell \cdot \log q$.} Thus, henceforth we will assume the existence of succinct NI-OTE and succinct NI-MOLE interchangeably, with the understanding the succinct NI-OTE implies the existence of both.

\section{Adaptive Lattice Encodings}
\seclab{encodings}
\newcommand{\lenc}{\mathsf{LEnc}}

In the following we present our new construction of adaptive lattice encodings. Let $k:=k(\sec)$ and $q:=q(\sec)$ be positive integers, and let $\mathbf{G}$ be the $k$-dimensional gadget matrix $\mathbf{G}:=\id_k\otimes \mathbf{g}_q^\intercal$. For element $x\in\Z_q$, vectors $\mathbf{s}, \mathbf{r}\in\Z_q^k$, matrix $\mathbf{A}\in\Z_q^{k\times (k\cdot \log q)}$ and noise term $\mathbf{e}\in\Z^{k\cdot \log q}$, we define the corresponding \emph{adaptive lattice encoding} as:
\[
\lenc_\mathbf{A} (x;\mathbf{s}, \mathbf{r}, \mathbf{e}) := \mathbf{s}^\intercal\mathbf{A}+ x \cdot \mathbf{r}^\intercal \mathbf{G}+\mathbf{e}^\intercal.
\]
Note that, for an appropriately sampled $\mathbf{A}$, $\mathbf{s}$, and $\mathbf{e}$. It is straightforward to see that the encoding is computationally close to uniform under the LWE assumption. Next, we demonstrate the homomorphic properties of such encodings.

\subsection{Homomorphic Operations}\label{sec:hom_enc}

We show that our lattice encodings support addition, scalar multiplication, and even multiplication, provided that the encodings are encrypted with correlated secrets. We present a formal description of the algorithms below.

\begin{protocol}[Homomorphic Operations on Lattice Encodings]\label{prot:Homomorphic}
\begin{description}
\item[Addition:] For every $x_0, x_1\in\Z_q$, vectors $\mathbf{s}, \mathbf{r}\in \Z_q^k$, matrices $\mathbf{A}_0, \mathbf{A}_1\in\Z_q^{k\times (k\cdot \log q)}$ and noise terms $\mathbf{e}_0, \mathbf{e}_1\in\Z_q^{k\cdot \log q}$, we have:
\[
\lenc_{\mathbf{A}_0} (x_0 ;\mathbf{s}, \mathbf{r}, \mathbf{e}_0) +
\lenc_{\mathbf{A}_1} (x_1;\mathbf{s}, \mathbf{r}, \mathbf{e}_1)
=
\lenc_{\mathbf{A}_0+ \mathbf{A}_1} (x_0+x_1; \mathbf{s}, \mathbf{r}, \mathbf{e}_0 +\mathbf{e}_1).
\]

\item[Scalar Multiplication:] For every $x, \delta\in\Z_q$, vectors $\mathbf{s}, \mathbf{r}\in \Z_q^k$, matrix $\mathbf{A}\in\Z_q^{k\times (k\cdot \log q)}$ and noise terms $\mathbf{e}\in\Z_q^{k\cdot \log q}$, we have:
\[
\lenc_\mathbf{A} (x; \mathbf{s}, \mathbf{r}, \mathbf{e}) \cdot \mathbf{G}^{-1}(\delta\cdot \mathbf{G})
=
\lenc_{\mathbf{A}\cdot \mathbf{G}^{-1}(\delta\cdot \mathbf{G})} (x\cdot \delta; \mathbf{s}, \mathbf{r}, \mathbf{G}^{-1}(\delta\cdot \mathbf{G})^\intercal\cdot\mathbf{e}).
\]

\item[Multiplication:] For every $x_0, x_1\in\Z_q$, vectors $\mathbf{s}_0, \mathbf{s}_1, \mathbf{s}_2\in \Z_q^k$, matrices $\mathbf{A}_0, \mathbf{A}_1\in\Z_q^{k\times (k\cdot \log q)}$ and noise terms $\mathbf{e}_0, \mathbf{e}_1\in\Z_q^{k\cdot \log q}$, we have:
\begin{align*}
&-\lenc_{\mathbf{A}_0}(x_0; \mathbf{s}_0, \mathbf{s}_1, \mathbf{e}_0) \cdot \mathbf{G}^{-1}(\mathbf{A}_1)+x_0\cdot \lenc_{\mathbf{A}_1}(x_1; \mathbf{s}_1, \mathbf{s}_2, \mathbf{e}_1)
\\
&\quad = \lenc_{-\mathbf{A}_0\cdot \mathbf{G}^{-1}(\mathbf{A}_1)}(x_0\cdot x_1; (\mathbf{s}_0, \mathbf{s}_2, -\mathbf{G}^{-1}(\mathbf{A}_1)^\intercal\cdot \mathbf{e}_0+x_0\cdot\mathbf{e}_1)).
\end{align*}
\end{description}
\end{protocol}
We elaborate on the correctness of the claimed homomorphic operations. For addition, we observe that:
\begin{align*}
    \lenc_{\mathbf{A}_0} (x_0 ;\mathbf{s}, \mathbf{r}, \mathbf{e}_0) +
\lenc_{\mathbf{A}_1} (x_1;\mathbf{s}, \mathbf{r}, \mathbf{e}_1)
&=\mathbf{s}^\intercal\mathbf{A}_0+x_0\cdot\mathbf{r}^\intercal \mathbf{G}+\mathbf{e}_0^\intercal+\mathbf{s}^\intercal\mathbf{A}_1+x_1\cdot\mathbf{r}^\intercal \mathbf{G}+\mathbf{e}_1^\intercal\\
&=\mathbf{s}^\intercal\cdot(\mathbf{A}_0+\mathbf{A}_1)+(x_0+x_1)\cdot\mathbf{r}^\intercal \mathbf{G}+(\mathbf{e}_0^\intercal+\mathbf{e}_1^\intercal)\\
    &=\lenc_{\mathbf{A}_0+ \mathbf{A}_1} (x_0+x_1; \mathbf{s}, \mathbf{r}, \mathbf{e}_0 +\mathbf{e}_1).
\end{align*}
For scalar multiplication, we have:
\begin{align*}
    \lenc_\mathbf{A} (x; \mathbf{s}, \mathbf{r}, \mathbf{e}) \cdot \mathbf{G}^{-1}(\delta\cdot \mathbf{G})
    &=\mathbf{s}^\intercal\mathbf{A} \mathbf{G}^{-1}(\delta\cdot \mathbf{G})+x\cdot\mathbf{r}^\intercal \mathbf{G} \mathbf{G}^{-1}(\delta\cdot \mathbf{G})+\mathbf{e}^\intercal \mathbf{G}^{-1}(\delta\cdot \mathbf{G})\\
    &=\mathbf{s}^\intercal\mathbf{A} \mathbf{G}^{-1}(\delta\cdot \mathbf{G})+(x\cdot\delta)\cdot\mathbf{r}^\intercal \mathbf{G} +\mathbf{e}^\intercal \mathbf{G}^{-1}(\delta\cdot \mathbf{G})\\
    &=\lenc_{\mathbf{A}\cdot \mathbf{G}^{-1}(\delta\cdot \mathbf{G})} (x\cdot \delta; \mathbf{s}, \mathbf{r}, \mathbf{G}^{-1}(\delta\cdot \mathbf{G})^\intercal\cdot\mathbf{e}).
\end{align*}
Finally, for homomorphic multiplication, we have:
\begin{align*}
    &-\lenc_{\mathbf{A}_0}(x_0; \mathbf{s}_0, \mathbf{s}_1, \mathbf{e}_0) \cdot \mathbf{G}^{-1}(\mathbf{A}_1)+x_0\cdot \lenc_{\mathbf{A}_1}(x_1; \mathbf{s}_1, \mathbf{s}_2, \mathbf{e}_1)\\
    &\quad=-(\mathbf{s}_0^\intercal\mathbf{A}_0+x_0\cdot\mathbf{s}_1^\intercal \mathbf{G}+\mathbf{e}_0^\intercal)\cdot\mathbf{G}^{-1}(\mathbf{A}_1)+x_0\cdot (\mathbf{s}_1^\intercal\mathbf{A}_1+x_1\cdot\mathbf{s}_2^\intercal \mathbf{G}+ \mathbf{e}_1^\intercal)\\
    &\quad=\mathbf{s}_0^\intercal\cdot(-\mathbf{A}_0\mathbf{G}^{-1}(\mathbf{A}_1))-x_0\cdot \mathbf{s}_1^\intercal\mathbf{A}_1+x_0\cdot \mathbf{s}_1^\intercal\mathbf{A}_1+(x_0\cdot x_1)\cdot\mathbf{s}_2^\intercal \mathbf{G}-\mathbf{e}_0^\intercal\mathbf{G}^{-1}(\mathbf{A}_1)+x_0\cdot \mathbf{e}_1^\intercal\\
    &\quad=\lenc_{-\mathbf{A}_0\cdot \mathbf{G}^{-1}(\mathbf{A}_1)}(x_0\cdot x_1; (\mathbf{s}_0, \mathbf{s}_2, -\mathbf{G}^{-1}(\mathbf{A}_1)^\intercal\cdot \mathbf{e}_0+x_0\cdot\mathbf{e}_1)).
\end{align*}

\paragraph{Evaluating RMS Programs.} From the homomorphic operations described above, one can derive a general routine to evaluate any $T$-bounded (i.e., the maximum norm of an intermediate variable of the computation is bounded by $T$) RMS program of depth $d$. Recall that in an RMS program, one can sum any two variables, whereas multiplication can be done only so long as one of the two variables is an input. Before discussing the evaluation algorithm, let us generalize the notation described above to vectors, with:
\[
\lenc_\mathbf{A} (\mathbf{x};\mathbf{s}, \mathbf{r}, \mathbf{e}) := \mathbf{s}^\intercal\mathbf{A}+\mathbf{r}^\intercal (\mathbf{x}^\intercal\otimes\mathbf{G})+\mathbf{e}^\intercal
\]
by increasing the dimensions of the components appropriately. We refer to $\mathbf{s}$ as the \emph{encryption key} and $\mathbf{r}$ as the \emph{authentication key}. Homomorphic operations can be extended to vectors in a straightforward manner. We are now ready to state the algorithm to evaluate RMS programs. 

\begin{lemma}[Evaluation of RMS Programs]
\label{evalRMS}
Let $\mathbf{A}\in\Z_q^{k\times (m\cdot k\cdot \log q)}$ be a matrix, $\mathbf{x}\in\Z_q^{m-1}$ be an input, and let $f$ be a $T$-bounded depth-$d(\sec)$ RMS program. Define $\hat{\mathbf{x}}$ as the vertical concatenation of $\mathbf{x}$ and $1$.
For all $i\in[d]$, let:
\[
\mathbf{c}_i^\intercal := \lenc_{\mathbf{A}}(\hat{\mathbf{x}}; \mathbf{s}_i, \mathbf{s}_{i+1}, \mathbf{e}_i)
\]
with $\mathbf{s}_i\in \Z_q^k$ and such that $\max_i \|\mathbf{e}_i\|_\infty \leq \beta$. Then there exist two polynomial-time algorithms $\Eval\mathsf{RMSK}$ and $\Eval\mathsf{RMSC}$ such that:
\[
\mathbf{A}_f\gets \Eval\mathsf{RMSK}(\mathbf{A}, f) \quad\text{and}\quad
\Tilde{\mathbf{c}}\gets \Eval\mathsf{RMSC}(\mathbf{A}, f, \mathbf{x}, \{\mathbf{c}_i\}_{i\in [d]})
\]
with $\Tilde{\mathbf{c}}^\intercal \in \lenc_{\mathbf{A}_f}(f(\mathbf{x}), \mathbf{s}_0, \mathbf{s}_d, \Tilde{\mathbf{e}})$ such that $\|\Tilde{\mathbf{e}}\|_\infty \leq \beta \cdot O(T \cdot (k\cdot \log q)^d)$.
\end{lemma}

\begin{proof}
We refer to any encoding using $\mathbf{s_i}$ as encryption key and $\mathbf{s_j}$ as authentication key a level-$(i, j)$ encoding. Observe that splitting a level-$(i, i+1)$ encoding $\mathbf{c}_i$ into blocks of dimension $k\cdot \log q$, we obtain level-$(i, i+1)$ encodings: 
\[
\lenc_{\mathbf{A}_j}(x_j; \mathbf{s}_i, \mathbf{s}_{i+1}, \mathbf{e}_{i,j})
\]
where $x_j$ is the $j$-th entry of $\hat{\mathbf{x}}$, $\mathbf{A}_j$ is the $j$-th block of $\mathbf{A}$, and $\mathbf{e}_{i,j}$ is the $j$-th block of $k\cdot\log q$ entries in $\widetilde{\mathbf{e}}_{i}$).

Using the addition and multiplication by a constant, we can apply linear operations over encodings lying on the same level $(i, j)$. In this way, we obtain a level-$(i, j)$ encoding of the result. On the other hand, the operation increases the norm of the noise in the encodings: If the linear operation is described by a vector $\bm{\ell}$, the noise magnitude increases by a factor $\sum_{v} \log \ell_v$.
We can also homomorphically compute multiplications between any encoding on level $(i, j)$ and any encoding on level $(j, j+1)$. In this way, we obtain a level-$(i, j+1)$ encoding of the product. This time the noise magnitude increases by a factor $O(k\cdot \log q)$ and by an additive term $T\cdot \beta$ for each multiplication.
Finally, we observe that we can convert a level-$(i, j)$ encoding into a level $(i, j+1)$ encoding by simply multiplying by a level-$(j, j+1)$ encoding of 1. Notice that the latter is know given that the last entry of $\hat{\mathbf{x}}$ is a $1$.
Overall, the growth of the noise norm is bounded by a factor $\beta \cdot O(T \cdot (k\cdot \log q)^d)$.

We also highlight that for all these operations we described, we are able to derive the matrix underlying the output encodings from the matrices underlying the input encodings. Thus both algorithms are well-defined.
\end{proof}

\subsection{Compressing Lattice Encodings}\label{sec:comp_enc}

We describe a procedure to compress and decompress lattice encodings. Formally, this consists of a triple of algorithms $(\Setup, \Compress, \Expand)$ that allows one to sample a compressed version of a lattice encoding, that can be later on expanded into the format described above.

Let $n:=n(\sec), k:=k(\sec)$, $m:=m(\sec)$, and $q:=q(\sec)$ be positive integers. We are going to assume the existence of an $\alpha$-correct succinct NI-OTE protocol $\OTE=(\Setup, \allowbreak\Hash, \allowbreak\Enc, \allowbreak\Hash\Eval, \allowbreak\Enc\Eval)$ with bilinear encoder evaluation for $\Z_q^{m}\otimes\Z_q^{k\cdot\log q}$ where digests are vectors in $\Z_q^n$ and the encoder private information consists of a vector in $\Z_q^k$. Let $\mathbf{G}:=\id_k\otimes \mathbf{g}_q^\intercal$. Our scheme is formally described in Construction \ref{prot:R1BGG}. Observe that if we instantiate $\OTE$ with Construction \ref{prot:SOTE}, the size of the compressed encoding scales logarithmically in the size of its input.

\begin{protocol}[Compression of Lattice Encodings]\label{prot:R1BGG}

\begin{description}
\item[$\Setup(1^\secpar)$:] Return $\ck:= \pk\sample \OTE.\Setup(\bbbone^\sec)$.

\item[$\Compress(\ck, \mathbf{A}, \mathbf{x}, \mathbf{s}_0, \mathbf{s}_1, \mathbf{r})$:]  Compute $(\mathbf{d}, \psi) := \OTE.\Hash(\pk, \mathbf{x})$, then sample $\mathbf{e}\sample\chi(\sec)$. Compute
\[
\mathbf{h}^\intercal:= \mathbf{s}_0^\intercal \mathbf{A}+\mathbf{r}^\intercal (\mathbf{d}^\intercal \otimes \mathbf{G})+\mathbf{e}^\intercal
\]
and set $E\sample \OTE.\Enc(\pk, \mathbf{s}_1\otimes \mathbf{g}_q, \mathbf{r})$. Return 
$(\mathbf{h}, E)$.

 \item[$\Expand(\ck, \mathbf{A}, \mathbf{h}, E, \mathbf{x})$:] Recompute the hash 
 $(\mathbf{d}, \psi) := \OTE.\Hash(\pk, \mathbf{x})$ and set $\mathbf{v}:=\OTE.\Hash\Eval(\pk, E, \psi)$. Let $\mathbf{P}$ be the matrix of the $\OTE$ protocol that can be publicly derived from $\pk$. Return
 \[
 \mathbf{c}^\intercal:=\mathbf{h}^\intercal \mathbf{P}^\intercal+\mathbf{v}^\intercal.
 \]
\end{description}
\end{protocol}
We show that the expansion algorithm indeed leads to well-formed lattice encoding. By the correctness of the NI-OTE protocol, we have that:
\[
\mathbf{w}:=\OTE.\Enc\Eval(\pk, \mathbf{d}, \mathbf{r})=\mathbf{P}\cdot(\mathbf{d} \otimes \mathbf{r}\otimes \mathbf{g}_q).
\]
Then let us rewrite:
    \begin{align*}\mathbf{c}^\intercal&=\mathbf{h}^\intercal \mathbf{P}^\intercal+\mathbf{v}^\intercal\\
    &=\mathbf{s}_0^\intercal\cdot \mathbf{A}+\mathbf{r}^\intercal\cdot(\mathbf{d}^\intercal \otimes \mathbf{G})\cdot \mathbf{P}^\intercal+\mathbf{e}^\intercal\cdot \mathbf{P}^\intercal+\mathbf{v}^\intercal\\
    &=\mathbf{s}_0^\intercal\cdot \mathbf{A}\cdot \mathbf{P}^\intercal+\mathbf{r}^\intercal\cdot (\mathbf{d}^\intercal \otimes \id_k\otimes \mathbf{g}_q^\intercal)\cdot \mathbf{P}^\intercal+\mathbf{e}^\intercal\cdot \mathbf{P}^\intercal+\mathbf{v}^\intercal\\
    &=\mathbf{s}_0^\intercal\cdot \mathbf{A}\cdot \mathbf{P}^\intercal+(\mathbf{d}^\intercal \otimes \mathbf{r}^\intercal\otimes \mathbf{g}_q^\intercal)\cdot \mathbf{P}^\intercal+\mathbf{e}^\intercal\cdot \mathbf{P}^\intercal+\mathbf{v}^\intercal\\
    &=\mathbf{s}_0^\intercal\cdot \mathbf{A}\cdot \mathbf{P}^\intercal+{\mathbf{w}}^\intercal+\mathbf{e}^\intercal\cdot \mathbf{P}^\intercal+\mathbf{v}^\intercal\\
    &=\mathbf{s}_0^\intercal\cdot \mathbf{A}\cdot \mathbf{P}^\intercal+\mathbf{x}^\intercal\otimes (\mathbf{s}_1^\intercal\otimes \mathbf{g}_q^\intercal)+\mathbf{e}^\intercal\cdot \mathbf{P}^\intercal+\mathbf{e}'^\intercal\\
        &=\mathbf{s}_0^\intercal\cdot \mathbf{A}\cdot \mathbf{P}^\intercal+\mathbf{s}_1^\intercal\cdot (\mathbf{x}^\intercal\otimes\id_k\otimes \mathbf{g}_q^\intercal)+\mathbf{e}^\intercal\cdot \mathbf{P}^\intercal+\mathbf{e}'^\intercal\\
    &=\mathbf{s}_0^\intercal\cdot \mathbf{A}\cdot \mathbf{P}^\intercal+\mathbf{s}_1^\intercal\cdot (\mathbf{x}^\intercal\otimes \mathbf{G})+\mathbf{e}^\intercal\cdot \mathbf{P}^\intercal+\mathbf{e}'^\intercal\\
    &= \lenc_{\mathbf{A}\mathbf{P}^\intercal}(\mathbf{x}; \mathbf{s}_0, \mathbf{s}_1, \mathbf{e}\cdot \mathbf{P}+\mathbf{e}')
    \end{align*}
by the $\alpha$-correctness of the NI-OTE. We can bound the norm of the noise term by:
\[\lVert \mathbf{e}^\intercal\cdot \mathbf{P}^\intercal+\mathbf{e}'^\intercal\rVert_\infty\le k\cdot n\cdot \log q\cdot \lVert \mathbf{e}\rVert_\infty\cdot \lVert \mathbf{P}\rVert_\infty+\lVert \mathbf{e}'\rVert_\infty \le k\cdot n^{2r + 1}\cdot \log q\cdot B(\lambda) + \alpha \cdot\lVert\mathbf{x}\rVert_\infty.
\] 
by Equation (\ref{eq:matrix_norm2}) and by the $\alpha$-correctness of the NI-OTE.
We prove that the compressed encodings satisfy a notion of simulation that we define below.
\begin{theorem}[Simulatability]\thmlab{expansion}
Consider the following experiment $\mathsf{CompExp}_{\mathcal{A}, \mathsf{C}\Sim}(\bbbone^\sec)$ parametrized by an adversary $\mathcal{A} = (\mathcal{A}_0, \mathcal{A}_1)$ and a simulator $\mathsf{C}\Sim$:
\begin{itemize}
\item Sample $\ck\sample\Setup(\bbbone^\sec)$, $\mathbf{A}\sample\Z_q^{k\times (n\cdot k\cdot \log q)}$, and $\mathbf{s}_0, \mathbf{r}\sample\Z_q^k$.
\item Activate the adversary $(\mathbf{x}, \mathbf{s}_1, \aux)\gets\Adv_0(\bbbone^\sec, \ck, \mathbf{A})$.
\item Sample $b\sample\{0, 1\}$. If $b = 0$ compute:
\[
(\mathbf{h}_0, E_0)\sample\Compress(\ck, \mathbf{A}, \mathbf{x}, \mathbf{s}_0, \mathbf{s}_1, \mathbf{r})
\]
whereas if $b=1$ compute:
\[
(\mathbf{h}_1, E_1)\sample\mathsf{C}\Sim(\bbbone^\sec, \ck, \mathbf{A}).
\]
\item Obtain $b' \gets \Adv_1(\mathbf{h}_b, E_b, \aux)$ and return $1$ if and only if $b=b'$.
\end{itemize}
Then assuming the hardness of LWE and that $\OTE$ is encoder private, there exists a PPT simulator $\mathsf{C}\Sim$, such that for every PPT adversary $\mathcal{A}$, there exists a negligible function $\negl(\sec)$ such that, for every $\sec\in\N$, we have that:
\[
\left|\frac{1}{2} - \Pr\left[\mathsf{CompExp}_{\mathcal{A}, \mathsf{C}\Sim}(\bbbone^\sec) = 1\right] \right|\le\negl(\sec)
\]
where the probability is taken over the random coins of the experiment.
\end{theorem}
\begin{proof}
Consider the following sequence of hybrid experiments.
\begin{itemize}
    \item Hybrid $\mathcal{H}_0$: This is the original experiment.
    \item Hybrid $\mathcal{H}_1$: We provide the adversary $\Adv_1$ with a pair $(\mathbf{h}, E)$ where $\mathbf{h}\sample\Z_q^{n\cdot k\cdot \log q}$ and $E\sample \OTE.\Enc(\pk, \mathbf{s}_1\otimes \mathbf{g}_q, \mathbf{r})$ for $\mathbf{r}\sample\Z_q^k$.

    Indistinguishability from the previous hybrid follows by a direct reduction to LWE. Indeed, in Hybrid $\mathcal{H}_0$, we have that $\mathbf{h} =\mathbf{A}^\intercal\cdot \mathbf{s}_0+\mathbf{e}+(\mathbf{d}\otimes\mathbf{G}^\intercal)\cdot \mathbf{r}$, where in particular $\mathbf{A}^\intercal\cdot \mathbf{s}_0+\mathbf{e}$ is an LWE sample.
    
    \item Hybrid $\mathcal{H}_2$: We also simulate $E\sample \OTE.\Sim(\bbbone^\sec, \pk)$.

    Indistinguishability follows from a straightforward reduction to the encoder privacy of $\OTE$.
\end{itemize}
We conclude by observing that in Hybrid $\mathcal{H}_2$ the pair $(\mathbf{h}, E)$ given to $\Adv_1$ is independent of the values $\mathbf{x}$ and $\mathbf{s}_1$ chosen by $\Adv_0$. This concludes the proof.
\end{proof}

\section{Reverse Trapdoor Hashing for all Functions}\label{sec:rtdh}

We construct a selectively-secure reverse TDHs for all functions. Note that, for an expressive enough class of functions, in terms of functionality reverse TDHs are identical to the standard notion of TDH \cite{C:DGIMMO19}, since one can always encode the universal circuit as the input, and vice-versa.
However, since the encoding key grows with the size of the input, embedding a universal circuit introduces a dependency in the size of the function. Thus, we pay a price in succinctness, when going from TDH to reverse TDH. On the other hand, the opposite direction (reverse TDH $\implies$ TDH) has no such problem, since the size of the hash is anyway constant. Thus, reverse TDH appears to be a more powerful abstraction.

\subsection{Definition}

\begin{definition}[Reverse Trapdoor Hashing]
A reverse trapdoor hashing scheme for the function class $\mathcal{F}=(\mathcal{F}_\sec)_{\sec\in\N}$ with input size $m(\sec)$ and output of size $\ell(\sec)$ consists of a tuple of PPT algorithms $(\Setup, \Hash, \Gen, \Enc, \Dec)$ with the following syntax:
\begin{description}
\item[$\Setup(\bbbone^\sec)$:] The setup algorithm is randomised and takes as input the security parameter $\bbbone^\sec$. The output is a hash key $\hk$.
\item[$\Hash(\hk, f)$:] The hashing algorithm is randomised takes as input a hash key $\hk$ and a function $f\in\mathcal{F}_\sec$. The output is a digest $d$ and hasher's private information $\rho$.
\item[$\Gen(\hk, \mathbf{x})$:] The generation algorithm is randomised and takes as input an hash key $\hk$, an element $\mathbf{x}\in\Z_2^m$. The output is an encoding key $\ek$ and a trapdoor $\td$. 
\item[$\Enc(\hk, \ek, f, \rho)$:] The encoding algorithm is deterministic and takes as input a hash key $\hk$, an encoding key $\ek$ and a function $f\in\mathcal{F}_\sec$ and hasher's private information $\rho$. The output is an encoding $\mathbf{e}\in\Z_2^\ell$.
\item[$\Dec(\hk, \td, d)$:] The decoding procedure is deterministic and takes as input a hash key $\hk$, a trapdoor $\td$ and a digest $d$. The output is an encoding $\mathbf{e}'\in\Z_2^\ell$.
\end{description}
\end{definition}

{
\begin{definition}[Correctness]
We say that a reverse TDH scheme is correct if there exists a negligible function $\negl(\sec)$ such that, for every valid adversary\footnote{An adversary is valid if it outputs a vector $\mathbf{x}\in\Z_2^m$ and a function $f\in\mathcal{F}_\sec$.} $\Adv$ and sufficiently large $\sec\in\N$, it holds that:
\[\Pr\left[\Enc(\hk, \ek, f, \rho)\oplus \Dec(\hk, \td, d)\ne f(\mathbf{x})\right]\le \negl(\sec)\]
where the probability is taken over the random choice of $\hk\sample\Setup(\bbbone^\sec)$, $(\mathbf{x}, f)\sample\Adv(1^\sec, \hk)$, $(d,\rho)\sample\Hash(\hk, f)$, and $(\ek, \td)\sample\Gen(\hk, \mathbf{x})$.
\end{definition}}

\begin{definition}[Selective Function Privacy]\leavevmode\par
Consider the following experiment $\mathsf{FuncExp}_{\mathcal{A}}(\bbbone^\sec)$ parametrized by an adversary $\mathcal{A} = (\mathcal{A}_0, \mathcal{A}_1)$:
\begin{itemize}
    \item Activate the adversary $(f_0, f_1, \aux)\sample\Adv_0(\bbbone^\sec)$.
    \item Sample a hash key $\hk\sample\Setup(\bbbone^\sec)$.
    \item Sample a random bit $b\sample\{0, 1\}$. 
    \item Compute $(d, \rho)\sample\Hash(\hk, f_b)$.
    \item Compute $b' \gets \Adv_1(\hk, d, \aux)$.
    \item Return $1$ if and only if $b =b'$.
\end{itemize}
We say that a reverse trapdoor hashing scheme $(\Setup, \Hash, \Gen, \Enc, \Dec)$ is function private if for every PPT adversary $\mathcal{A}$, there exists a negligible function $\negl(\sec)$ such that, for every $\sec\in\N$, we have that:
\[
\left|\frac{1}{2} - \Pr\left[\mathsf{FuncExp}_{\mathcal{A}}(\bbbone^\sec) = 1\right] \right|\le\negl(\sec).
\]
If the above property holds for every adversary (even computationally unbounded ones) we say that the scheme is statistically function private.
\end{definition}

\begin{definition}[Selective Input Privacy]
Consider the following experiment $\mathsf{InpExp}_{\mathcal{A}, \Sim}(\bbbone^\sec)$ parametrized by an adversary $\mathcal{A} = (\mathcal{A}_0, \mathcal{A}_1)$ and a simulator $\Sim$:
\begin{itemize}
    \item Activate the adversary $(\mathbf{x}, \aux)\sample\Adv_0(\bbbone^\sec)$.
    \item Sample a hash key $\hk\sample\Setup(\bbbone^\sec)$.
    \item Sample a random bit $b\sample\{0, 1\}$. 
    \item If $b = 0$ compute $(\ek_0, \td_0)\sample\Gen(\hk, \mathbf{x})$, else compute 
    $\ek_1\sample\Sim(\bbbone^\sec, \hk)$.
    \item Compute $b' \gets \Adv_1(\hk, \ek_b, \aux)$.
    \item Return $1$ if and only if $b =b'$.
\end{itemize}
We say that a reverse trapdoor hashing scheme $(\Setup, \Hash, \Gen, \Enc, \Dec)$ is function-private if there exists a PPT simulator $\Sim$, such that for every PPT adversary $\mathcal{A}$, there exists a negligible function $\negl(\sec)$ such that, for every $\sec\in\N$, we have that:
\[
\left|\frac{1}{2} - \Pr\left[\mathsf{InpExp}_{\mathcal{A}, \Sim}(\bbbone^\sec) = 1\right] \right|\le\negl(\sec).
\]
\end{definition}

\subsection{Laconic Function Evaluation with Pre-Encoding}

We now define a new variant of LFE. We ask that the encoding procedure is split into two parts: First, the encoder sends a function-independent pre-encoding of the input, then, once the hash of the function is revealed, it sends an input-independent post-encoding. We also require that the function digest and the post-encoding have a particular structure: the former consists of a matrix, the latter consists of a \emph{LWE-like} sample. Finally, we ask that the output of the decoding procedure is produced by rounding the sum between the post-encoding and a (possibly non-linear) function of the pre-encoding. The reader may have already noticed that these properties are satisfied by many LFE schemes studied in the literature \cite{FOCS:QuaWeeWic18,FOCS:HsiLinLuo23,C:DHMWW24,C:Wee24}.

\begin{definition}[Laconic Function Evaluation with Pre-Encoding]\label{def:LFE}
Let $\mathcal{F}:=(\mathcal{F}_\sec)_{\sec\in\N}$ be a family of functions, an LFE scheme consists of a tuple of PPT algorithms $(\Setup, \Hash, \Enc, \Dec)$ with the following syntax:
\begin{description}
    \item[$\Setup(\bbbone^\sec)$:] The probabilistic setup algorithm takes as input the security parameter and outputs a public key $\pk$.
    \item[$\Hash(\pk, f)$:] The hashing algorithm takes as input a public key $\pk$ and the description of a function $f\in\mathcal{F}_\sec$ with $f: \{0,1\}^{m} \to \{0,1\}$.
    The output is a digest $\mathbf{A}_f\in\Z_q^{k\times (k\cdot \log q)}$ and hasher's private information $\psi$.
    \item[$\Enc(\pk,\mathbf{A}_f, \mathbf{x})$:] The encoding algorithm takes as input a public key $\pk$ a hash $\mathbf{A}_f$, and an input string $\mathbf{x}\in\{0,1\}^m$. The algorithm is divided into two subroutines.
    \begin{description}
    \item[$\Pre\Enc(\pk, \mathbf{x})$:] The pre-encoding algorithm is independent of the hash: it takes as input a public key $\pk$, an input $\mathbf{x}$ (and sometimes a vector $\mathbf{r}\in\Z_q^{k-1}$). It returns a pre-encoding of the input $E$, along with a private information $\mathbf{s}\in\Z_q^{k}$. 
    \item[$\Post\Enc(\pk, \mathbf{A}_f, \mathbf{s})$:] The post-encoding algorithm does not depend on the input and returns a post-encoding defined as:
    \[
    c := \mathbf{s}^\intercal \cdot \mathbf{A}_f\cdot \mathbf{t} + e \in \Z_p
    \]
    where $e \sample \chi(\bbbone^\sec)$ and $\mathbf{t}$ is sampled over $\Z_2^{k\cdot \log q}$ (not necessarily at random).

    \end{description}
    The algorithm outputs an encoding $(E, c, \mathbf{t})$.
    
    \item[$\Dec(\pk, (E, c, \mathbf{t}), \psi)$:] The decoding algorithm takes as input a public key $\pk$, an encoding $(E, c, \mathbf{t})$ and hasher's private information $\phi$. The output is a bit 
    \[
    \lceil\Eval(E, \psi) \cdot \mathbf{t}+c\rfloor_2
    \]
    where $\Eval$ is a polynomial-time deterministic algorithm that returns a vector in $\Z_q^{k\cdot \log q}$.
\end{description}
\end{definition}

We define a particular version of correctness that applies in most scheme as in \defref{LFE}

\begin{definition}[Special Correctness]
Let $\widetilde{B}:=\widetilde{B}(\sec)$ be a function of the security parameter. An LFE scheme with pre-encoding satisfies $\widetilde{B}(\sec)$-special correctness if, for every $\sec\in\N$, function $f\in\mathcal{F}_\sec$ and input $\mathbf{x}\in\Z_2^m$, we have
\[
\Eval(E, \psi) = \mathbf{s}^\intercal\cdot \mathbf{A}_f+ \mathbf{r}^\intercal\cdot f(\mathbf{x})\cdot \mathbf{G}+\mathbf{\widetilde{e}}^\intercal
\]
where $\mathbf{G}=\id_k\otimes \mathbf{g}_q^\intercal$, $\mathbf{\widetilde{e}}$ is such that $\lVert\mathbf{\widetilde{e}}\rVert_\infty\le \widetilde{B}(\sec)$ and the last entry of $\mathbf{r}$ is $-1$.
\end{definition}
We recall the notion of hasher privacy.
\begin{definition}[Selective Hasher Privacy]
Consider the following experiment $\mathsf{HashExp}_{\mathcal{A}}(\bbbone^\sec)$ parametrized by an adversary $\mathcal{A} = (\mathcal{A}_0, \mathcal{A}_1)$:
\begin{itemize}
    \item Activate the adversary $(f_0, f_1, \aux)\sample\Adv_0(\bbbone^\sec)$.
    \item Sample a public key $\pk\sample\Setup(\bbbone^\sec)$.
    \item Sample a random bit $b\sample\{0, 1\}$. 
    \item Compute $(\mathbf{A}_f, \rho)\sample\Hash(\pk, f_b)$.
    \item Compute $b' \gets \Adv_1(\pk, \mathbf{A}_f, \aux)$.
    \item Return $1$ if and only if $b =b'$.
\end{itemize}
We say that an LFE scheme with pre-encoding $(\Setup, \Hash, \Pre\Enc, \Post\Enc, \Dec)$ is hasher-private if for every PPT adversary $\mathcal{A}$, there exists a negligible function $\negl(\sec)$ such that, for every $\sec\in\N$, we have that:
\[
\left|\frac{1}{2} - \Pr\left[\mathsf{HashExp}_{\mathcal{A}}(\bbbone^\sec) = 1\right] \right|\le\negl(\sec).
\]
If the above property holds for every adversary (even computationally unbounded ones) we say that the scheme is statistically hasher private.
\end{definition}
Finally we define a slightly weaker version of the standard encoder privacy, which however suffices for our purposes. Namely, we only require security against a distinguisher that sees the pre-encoding information (and we do not pose any requirement on the post-encoding).

\begin{definition}[Selective Pre-Encoding Privacy]
Consider the following experiment $\mathsf{PreEncExp}_{\mathcal{A}, \Sim}(\bbbone^\sec)$ parametrized by an adversary $\mathcal{A} = (\mathcal{A}_0, \mathcal{A}_1)$ and a simulator $\Sim$:
\begin{itemize}
    \item Activate the adversary $(\mathbf{x}, \aux)\sample\Adv_0(\bbbone^\sec)$.
    \item Sample a public key $\pk\sample\Setup(\bbbone^\sec)$.
    \item Sample a random bit $b\sample\{0, 1\}$. 
    \item If $b=0$, compute $(E_0, \phi_0)\sample\Enc(\pk, \mathbf{x})$. Otherwise, compute $E_1\sample\Sim(\bbbone^\sec, \pk)$.
    \item Compute $b' \gets \Adv_1(\pk, E_b, \aux)$.
    \item Return $1$ if and only if $b =b'$.
\end{itemize}
We say that an LFE scheme with pre-encoding $(\Setup, \Hash, \Pre\Enc, \Post\Enc, \Dec)$ is pre-encoder private if for every PPT adversary $\mathcal{A}$, there exists a negligible function $\negl(\sec)$ such that, for every $\sec\in\N$, we have that:
\[
\left|\frac{1}{2} - \Pr\left[\mathsf{PreEncExp}_{\mathcal{A}}(\bbbone^\sec) = 1\right] \right|\le\negl(\sec).
\]
\end{definition}
Most known LFE schemes satisfy (or can be adapted to satisfy) the above syntactical requirements. We summarize the state of the art in the following:
\begin{itemize}
    \item Assuming the hardness of LWE, there exists an LFE scheme with pre-encoding for all bounded-depth circuits \cite[Appendix E]{FOCS:QuaWeeWic18}.
    \item Assuming the hardness of small-secret circular LWE, there exists an LFE scheme with pre-encoding for all (no bound on the depth) circuits \cite{FOCS:HsiLinLuo23}.
    \item Assuming the hardness of Ring-LWE (small-secret circular Ring-LWE, resp.) there exists an LFE scheme with pre-encoding for all bounded-depth (unbounded-depth, resp.) RAM programs \cite{C:DHMWW24}.\footnote{{Although the construction presented in \cite{C:DHMWW24} is not shown to satisfy function hiding, the generic compiler from \cite{FOCS:QuaWeeWic18} can be applied here to upgrade the scheme to satisfy this property. The compiler preserves the asymptotic efficiency of the construction, and we refer the reader to \cite{FOCS:QuaWeeWic18} for details.}}
\end{itemize}
For our purposes, the details of these constructions will be irrelevant, provided that they satisfy the above syntax. Henceforth, we just assume that such an LFE exists, with the understanding that the exact efficiency guarantees of our construction will depend on the building block used to instantiate it.

\subsection{Construction}

In the following we define our construction for a reverse TDH for functions
\[
f: \{0,1\}^m \to \{0,1\}^\ell.
\]
We will use the following ingredients instantiating them over the ring $\Z_q$:
\begin{itemize}
    \item The $\alpha$-correct succinct MOLE protocol
     $(\Setup, \Hash, \Enc, \Hash\Eval, \Enc\Eval)$, see Protocol \ref{prot:HCVOLE}.
    \item An LFE scheme with pre-encoding.
     $(\Setup, \Hash, \Pre\Enc, \Post\Enc, \Dec)$. Suppose that the scheme satisfies $\widetilde{B}$-special correctness.
   \item A pseudorandom generator $\mathsf{PRG}:\{0, 1\}^\sec\rightarrow \Z_q^{\ell}$ that can be evaluated uniformly (e.g., instantiated via a pseudorandom function).
\end{itemize}
For convenience, we assume that $q$ is even. Moreover, we assume that $\alpha, \widetilde{B}\le q\cdot\negl(\sec)$. The protocol description is presented below.

\begin{protocol}[Reverse Trapdoor Hash]\label{prot:RTDH}
\setlength{\abovedisplayskip}{4pt}\setlength{\belowdisplayskip}{4pt}
\setlength{\belowdisplayshortskip}{4pt}

\begin{description}
\item[$\Setup(1^\secpar)$:] Sample keys
\[\pk \sample \LFE.\Setup(\bbbone^\sec), \qquad\qquad \mpk \sample \MOLE.\Setup(\bbbone^\sec).\]
Then, output $\hk:=(\pk, \mpk)$.

\item[$\Gen(\hk, \mathbf{x})$:] Compute
\[(E, \mathbf{s}) \sample \LFE.\Pre\Enc(\pk, \mathbf{x}), \qquad (E', \phi) \sample \MOLE.\Enc(\mpk, \mathbf{s})\]
{Then, sample $\mathsf{seed}\sample\{0, 1\}^\sec$ and output $\ek:=(E, E', \mathsf{seed})$ and $\td:=(\phi, \mathsf{seed})$.}

\item[$\Hash(\hk, f)$:] Let $\mathbf{u}$ be the last element of the standard basis over $\Z_q^k$. Let $(f_0, \dots, f_{\ell-1})$ be the functions that compute the $i$-th bit of the output of $f$. For every $i\in[\ell]$ compute
\[(\mathbf{A}_{f_i}, \psi_i)\sample\LFE.\Hash(\pk, f_i).\]
Finally, compute
\begin{align*}\mathbf{A}&\gets \left(\mathbf{A}_{f_0}\cdot \mathbf{G}^{-1}\left(-\frac{q}{2}\cdot \mathbf{u}\right)\mathbin{\Big\|}\dots \mathbin{\Big\|}\mathbf{A}_{f_{\ell-1}}\cdot \mathbf{G}^{-1}\left(-\frac{q}{2}\cdot \mathbf{u}\right)\right)\\
(d, \psi) &\sample \MOLE.\Hash(\mpk, \mathbf{A}^\intercal)\end{align*}
and output $d$ and $\rho:=(\psi, \psi_0, \dots, \psi_{\ell-1})$.

\item[$\Enc(\hk, \ek, f, \rho)$:] Let $\mathbf{u}$ be the last element of the standard basis over $\Z_q^k$. Compute
\[\mathbf{v}\gets \MOLE.\Hash\Eval(\mpk, E', \psi).\]
Parse $\mathbf{v}$ as the vertical concatenation of $(v_0, \dots, v_{\ell-1})$ 
and let $(r_0, \dots, r_{\ell-1})\gets \mathsf{PRG}(\mathsf{seed})$. For every $i\in[\ell]$, compute
\[e_i \gets \left\lceil r_i + \LFE.\Eval(E, \psi_i)\cdot  \mathbf{G}^{-1}(-q/2\cdot \mathbf{u}) - v_i \right\rfloor_2\]
Finally, output $\mathbf{e}:= \begin{pmatrix}
  e_0 \\ \vdots\\ e_{\ell-1}  
\end{pmatrix}$.
\item[$\Dec(\hk, d, \td)$:] Compute
\[\mathbf{w}\gets\MOLE.\Enc\Eval(\mpk, d, \phi).\]
 Parse $\mathbf{w}$ as the vertical concatenation of $(w_0, \dots, w_{\ell-1})$ 
and derive $(r_0, \dots, r_{\ell-1})\gets \mathsf{PRG}(\mathsf{seed})$. Then, for all $i\in[\ell]$, compute
\[e_i' \gets \left\lceil - r_i -w_i\right\rfloor_2\]
Finally, output $\mathbf{e}':=\begin{pmatrix}e_0'\\ \vdots\\ e_{\ell-1}'\end{pmatrix}$.
\end{description}
\end{protocol}

\paragraph{Correctness.}
For correctness, observe that for all $i\in[\ell]$ we have that:
\begin{align*}
    &\underbrace{r_i + \LFE.\Eval(E, \psi_i)\cdot  \mathbf{G}^{-1}(-q/2\cdot \mathbf{u}) - v_i}_{s_0} +
    \underbrace{- r_i -v_i'}_{s_1}\\
    &=\LFE.\Eval(E, \psi_i)\cdot  \mathbf{G}^{-1}(-q/2\cdot \mathbf{u}) -(v_i +v'_i)\\
    &=(\mathbf{s}^\intercal\cdot \mathbf{A}_{f_i}+ \mathbf{r}^\intercal\cdot f_i(\mathbf{x})\cdot \mathbf{G}+\mathbf{\widetilde{e}}^\intercal)\cdot  \mathbf{G}^{-1}(-q/2\cdot \mathbf{u}) - \mathbf{s}^\intercal\cdot \mathbf{A}_{f_i}\cdot \mathbf{G}^{-1}(-q/2\cdot \mathbf{u})-e_i\\  
    &= -q/2\cdot f_i(\mathbf{x})\cdot \mathbf{r}^\intercal\cdot\mathbf{u}+ \widetilde{\mathbf{e}}^\intercal\cdot  \mathbf{G}^{-1}(-q/2\cdot \mathbf{u})-e_i\\
    &= q/2\cdot f_i(\mathbf{x}) +\hat{e}_i
\end{align*}
where $\hat{e}_i:=\widetilde{\mathbf{e}}^\intercal\cdot  \mathbf{G}^{-1}(-q/2\cdot \mathbf{u})-e_i$ and the second equality follows from the $\alpha$-correctness of the MOLE and the special correctness of the LFE. Notice that $\lvert \widetilde{\mathbf{e}}^\intercal\cdot  \mathbf{G}^{-1}(-q/2\cdot \mathbf{u})-e_i\rvert\le \alpha(\sec)+\widetilde{B}(\sec)$. Thus we have that $s_0$ and $s_1$ is a noisy secret sharing of $f_i(\mathbf{x})$. Furthermore, it holds that
\[
f_i(\mathbf{x}) = \lceil s_0 + s_1\rfloor_2
= \lceil s_0\rfloor_2 \oplus \lceil s_1\rfloor_2
\]
except if $s_0 \in [q/4 -  \lvert\hat{e}_i\rvert, q/4 + \lvert\hat{e}_i\rvert] \cup [3q/4 - \lvert\hat{e}_i\rvert, 3q/4 + \lvert\hat{e}_i\rvert]$. Note that the size of the interval is at most $ 4\cdot\lvert\hat{e}_i\rvert$, which is a negligible fraction of $q$. Thus, by the pseudorandomness of $\mathsf{PRG}$, this event happens with negligible probability {(independently of the inputs and the public parameters),} concluding our proof of correctness.

\paragraph{Security.} We prove security in the following.
\begin{theorem}
    Suppose that $\MOLE$ is encoder private. If $\LFE$ is pre-encoding secure, Construction \ref{prot:RTDH} is input private. 
    Finally, if $\LFE$ is (statistically) function private, the reverse tradpoor hashing scheme is (statistically) hasher private.
\end{theorem}
\begin{proof}
It is easy to see that the construction is (statistically) function private if $\LFE$ is (statistically) function private.

As for input privacy, we proceed by a series of indistinguishable hybrids.

\begin{itemize}
\item Hybrid $\mathcal{H}_0$: This hybrid corresponds to the original game: We provide the adversary $\Adv_1$ with an encoding key $(E, E', \mathsf{seed})$ generated using $\Gen(\hk, \mathbf{x})$. 

\item Hybrid ${\mathcal{H}_1}$: In this hybrid, we generate $E'$ using $\MOLE.\Sim(\bbbone^\sec, \mpk)$. The rest remains as in the previous hybrid.

Hybrid ${\mathcal{H}_0}$ and Hybrid ${\mathcal{H}_1}$ are computationally indistinguishable thanks to the encoder privacy of the MOLE.

\item Hybrid ${\mathcal{H}_2}$: In this hybrid, we generate $E$ using $\LFE.\Sim(\bbbone^\sec, \pk)$. 
The rest remains as in the previous hybrid.

Hybrid ${\mathcal{H}_1}$ and ${\mathcal{H}_2}$ are indistinguishable under the input privacy of $\LFE$.
\end{itemize}
Notice that in Hybrid ${\mathcal{H}_2}$, the triple $(E, E', \mathsf{seed})$ provided to the adversary contains no information about $\mathbf{x}$. From this, we can easily derive a simulator. This ends the proof.
\end{proof}

{\paragraph{Parameters.} Construction~\ref{prot:RTDH}
inherits its efficiency from two ingredients: The pre-encoding part of the
underlying LFE scheme and the succinct MOLE used to compress the collection of
post-encoding matrices. Let
$s_{\mathsf{pre}}(\lambda,m,\ell)$
denote the size of the pre-encoding in the underlying LFE instantiation and let
$s_{\mathsf{hash}}(\lambda,\ell)$
denote the size of the corresponding reverse-TDH digest produced after MOLE
compression. Then Construction~\ref{prot:RTDH} has digest size
$s_{\mathsf{hash}}(\lambda,\ell)$,
encoding-key size
\[
s_{\mathsf{pre}}(\lambda,m,\ell)
+\poly(\lambda,\log \ell,\log q,\log k),
\]
and online encoding size exactly $\ell$. In particular, any dependence on the circuit depth is inherited from the
underlying LFE parameters, including the induced dependence of $k$ and $q$ on
that depth.}

{
\begin{theorem}[Reverse Trapdoor Hashing]\label{thm:reverseTDH}
Assume that MOLE is encoder private and that LFE is pre-encoding secure and
function private. Then Construction~\ref{prot:RTDH} is an input-private
and function-private reverse TDH for functions
$f:\{0,1\}^m\to\{0,1\}^\ell$ with digest size $s_{\mathsf{pre}}(\lambda,m,\ell)$,
encoding-key size
$s_{\mathsf{pre}}(\lambda,m,\ell)
+\poly(\lambda,\log \ell,\log q,\log k)$,
and online encoding size $\ell$.
\end{theorem}}

{\subsection{Applications: Laconic OT and Succinct HSS}}

\noindent{We briefly record two direct consequences of the reverse trapdoor hash of Construction \ref{prot:RTDH}: A public-key two-party homomorphic secret-sharing scheme and a batched laconic oblivious transfer protocol.}

{\begin{theorem}[Succinct HSS from Reverse TDH]
Assume that the hypotheses of Theorem~\ref{thm:reverseTDH} hold.  Then there exists a
public-key two-party homomorphic secret-sharing scheme for all functions
$f : \{0,1\}^{m} \times \{0,1\}^{n} \to \{0,1\}^{\ell}$
with communication complexity $m \cdot \poly(\lambda,\log n)$.
\end{theorem}}

{\begin{proof}[Proof Sketch]
    Fix $x$ and let $f_x(\,\cdot\,) := f(x,\cdot)$. Alice runs the reverse-TDH hashing algorithm on $f_x$, obtaining a digest $d$ and local randomness $\rho$.  She sends only the digest $d$ to Bob and keeps $\rho$ as her local share material. Bob runs the reverse-TDH key-generation algorithm on his input $y$, obtaining an encoding key $\mathsf{ek}$ and trapdoor $\mathsf{td}$. He sends $\mathsf{ek}$ to Alice and keeps $\mathsf{td}$ as his local share material.
Alice computes
\[
z_A := \Enc(\mathsf{hk},\mathsf{ek},f_x,\rho),
\]
while Bob computes
\[
z_B := \Dec(\mathsf{hk},\mathsf{td},d).
\]
By correctness of the reverse TDH we have $z_A \oplus z_B = f_x(y) = f(x,y)$.
Thus $(z_A,z_B)$ is a valid additive secret-sharing of the output.
The communication is exactly the reverse-TDH digest for the function side plus
the reverse-TDH encoding key for the input side.  Since the digest is succinct in
the description of the hashed function, this yields communication polylogarithmic in $|x|$ and linear in the short side $|y|$.
\end{proof}}
{The following theorem uses  Construction~\ref{prot:RTDH} instantiated with LFE for RAM programs, whose security follows from the hardness of Ring-LWE.}

{\begin{theorem}[Laconic OT from Reverse TDH]
Assume that the hypotheses of Theorem~\ref{thm:reverseTDH} hold and that pseudorandom functions exist. Then there exists a batched laconic OT protocol with download rate $1/2$.
\end{theorem}}

{\begin{proof}[Proof Sketch]
Let $D=\{D_i\}_{i\in[N]} \in \{0,1\}^N$ be the receiver's choice bits, and let the
sender hold pairs of messages
\[
\{(m_{i,0},m_{i,1})\}_{i\in I}
\]
for a batch of queried indices $I \subseteq [N]$. Define the RAM program $f_D$ which has the database $D$ hardwired and, on input a PRF key $k$, returns for every queried pair $(i,b)$:
\[
f_D(k;i,b)=
\begin{cases}
0 & \text{if } b = D_i,\\
\mathsf{PRF}(k,i) & \text{if } b \neq D_i.
\end{cases}
\]
The receiver hashes $f_D$ using the reverse TDH, obtaining digest $d$ and local
randomness $\rho$, and sends the digest to the sender. The sender samples a fresh PRF key $k$, runs the reverse-TDH generation algorithm
to obtain an encoding key $\mathsf{ek}$ for $k$, and computes, for each requested
$(i,b)$,
\[
c_{i,b} := \Dec(\mathsf{hk},\mathsf{td},d)_{i,b} \oplus m_{i,b},
\]
where $\mathsf{td}$ is the sender's trapdoor output by key generation.
The sender transmits $\mathsf{ek}$ together with all ciphertext bits
$\{c_{i,b}\}_{i\in I,b\in\{0,1\}}$. For the selected branch $b=D_i$, correctness of reverse TDH gives
\[
\Enc(\mathsf{hk},\mathsf{ek},f_D,\rho)_{i,b}
=
\Dec(\mathsf{hk},\mathsf{td},d)_{i,b},
\]
and since in that case $f_D(k;i,b)=0$, the receiver can recover
$m_{i,D_i}$.
For the unselected branch $b \neq D_i$, the value revealed by $f_D$ is
$\mathsf{PRF}(k,i)$, so the corresponding pad is pseudorandom and computationally hides
$m_{i,1-D_i}$.
The sender sends one mask per two transferred messages, so the asymptotic transfer
rate for a growing $|I|$ is $1/2$.
\end{proof}}
\section{Rate-1 Adaptive LFE for all Bounded-Depth Functions}

\subsection{Input-Succinct Reverse Trapdoor Hashing for RMS}\label{sec:weakTDH}

As a stepping stone, we present a variant of a trapdoor hashing (TDH) scheme \cite{C:DGIMMO19} with reversed syntax, achieving almost optimal encoding key size but supporting a restricted family of functions. 

Let $m:=m(\sec)$, $\ell:=\ell(\sec)$, $p:=p(\sec)$, $d:=d(\sec)$, $T:=T(\sec)$, $k:=k(\sec)$ and $t:=t(\sec)$ be positive integers. Construction \ref{prot:R1LFE} allows the evaluation of functions that map any pair $(\mathbf{x}, \mathbf{a})$, where $\mathbf{x}\in\{0, 1\}^m$ and $\mathbf{a}\in\Z_p^{t\cdot k}$, into $f(\mathbf{x})\otimes \mathbf{a}$ where $f:\{0, 1\}^m\rightarrow \{0, 1\}^\ell$ is described by a depth-$d$, $T$-bounded RMS program. 

Our construction only achieves \emph{adaptive} input privacy, but only a weak notion of correctness: Although the encoding key leaks no information about $\mathbf{x}$ and $\mathbf{a}$, in order to successfully run the encoding procedure, we are required to know $\mathbf{x}$ (but not $\mathbf{a}$). Although these properties may seem artificial, they will later be useful in combination with techniques of \cite{TCC:BTVW17}, to build adaptive LFE with rate-1 encodings (Section~\ref{sec:mainLFE}).

Concerning efficiency, the digest size in our construction is logarithmic in $m$, $t$ and linear in $\ell$, $d$. Moreover, the size of the encoding key is $t\cdot \poly(\sec, \log m, \log \ell, d)$. Notice that we can achieve logarithmic dependency in $m$ only because the encoding procedure needs $\mathbf{x}$ in order to run successfully. 

\paragraph{The Construction.} We now invite the reader to take a look at Construction \ref{prot:R1LFE}. The scheme relies on the compressed, adaptive lattice encodings of \secref{encodings}. We instantiate them over $\Z_q$, where $q=\Delta\cdot p$ where $\Delta:=\Delta(\sec)$ is a positive integer. We use $(\Setup, \Compress, \Expand)$ to denote the algorithms of Construction \ref{prot:R1BGG} and we use $n:=n(\sec)$ to denote the digest size of the $\alpha$-correct OTE with which it is instantiated. Let $\mathbf{P}$ be the matrix used for the hasher evaluation in the OTE protocol. We also rely on an $\hat{\alpha}$-correct succinct NI-MOLE protocol $\MOLE=(\Setup, \Hash, \Enc, \Hash\Eval, \Enc\Eval)$ for matrices of size $\Z_q^{(t\cdot \ell\cdot k\cdot \log q)\times k}$.

\begin{protocol}[Input-Succinct Reverse TDH for RMS]\label{prot:R1LFE}
\setlength{\abovedisplayskip}{6pt}\setlength{\belowdisplayskip}{6pt}
\setlength{\belowdisplayshortskip}{6pt}

\begin{description}
\item[$\Setup(1^\secpar)$:] Run the setup for the MOLE and the compressed lattice encodings
\[\ck\sample\Setup(\bbbone^\sec),\qquad \qquad \mpk\sample\MOLE.\Setup(\bbbone^\sec).\]
Then, sample $\mathbf{A}\sample\Z_q^{k\times (n\cdot k\cdot \log q)}$, $\mathbf{B}\sample\Z_q^{k\times (t\cdot k\cdot \log q)}$ and return $\hk:=(\ck, \mpk, \mathbf{A}, \mathbf{B})$.

\item[$\Hash(\hk, f)$:] Compute:
\[ \mathbf{A}_f\gets \Eval\mathsf{RMSK}(\mathbf{A}\mathbf{P}^\intercal, f), \qquad \qquad \mathbf{F}\gets -\mathbf{A}_f\cdot (\id_\ell \otimes \mathbf{G}^{-1}(\mathbf{B})).\]
Then output $(\mathbf{d}, \psi)\sample\MOLE.\Hash(\mpk,\mathbf{F}^\intercal)$.

\item[$\Gen\left(\hk, \mathbf{x}, \mathbf{a}; \{\mathbf{r}_i\}_{i\in[d]}\right)$:] Let $\hat{\mathbf{x}}$ be the vertical concatenation of $\mathbf{x}$ and $1$.
For every $i\in[d+1]$, sample $\mathbf{s}_i\sample\Z_q^k$ and set:
\[(\mathbf{h}_i, E_{i})\sample\Compress(\ck, \mathbf{A}, \hat{\mathbf{x}}, \mathbf{s}_i, \mathbf{s}_{i+1}, \mathbf{r}_i).\]
Next, sample $\mathbf{e}\sample\chi(\sec)$ and set
\[
\mathbf{b}^\intercal:= \mathbf{s}_d^\intercal \mathbf{B}+\mathbf{a}^\intercal\cdot (\id_t\otimes \mathbf{G}) +\mathbf{e}^\intercal,
\]
where $\mathbf{G}=\id_k\otimes \mathbf{g}_q^\intercal$.

Proceed by generating $(C, \phi)\sample\MOLE.\Enc(\mpk, \mathbf{s}_0)$ and sampling $\mathsf{seed}\sample\{0, 1\}^\sec$.

 Output $\ek:=(C, \mathbf{b}, \{\mathbf{h}_i, E_i\}_{i\in[d]}, \mathsf{seed})$ and $\td:=(\phi, \mathsf{seed})$

\item[$\Enc(\hk, \ek, \psi, \mathbf{x}, f)$:] Let $\hat{\mathbf{x}}$ be the vertical concatenation of $\mathbf{x}$ and $1$. For every $i\in[d]$, compute
\[
\mathbf{c}_i^\intercal\gets \Expand(\ck, \mathbf{A}, \mathbf{h}_i, E_i, \hat{\mathbf{x}})
\]
and set
\[
\mathbf{c}^\intercal\gets \Eval\mathsf{RMSC}(\mathbf{A}\mathbf{P}^\intercal, f, \mathbf{x}, \{\mathbf{c}_i\}_{i\in[d]}).
\]
Next, derive
\[
\mathbf{z}\gets -(\id_\ell\otimes \mathbf{G}^{-1}(\mathbf{B}))^\intercal \mathbf{c}
+f(\mathbf{x})\otimes \mathbf{b}.
\]
Finally, compute
\begin{gather*}
\mathbf{v}\gets \MOLE.\Hash\Eval(\mpk, C, \psi), \qquad\qquad
\tilde{\mathbf{u}}\gets (\id_{t\cdot \ell\cdot k}\otimes \mathbf{G}^{-1}(\Delta))^\intercal (\mathbf{z}-\mathbf{v}),\\
\tilde{\mathbf{y}}\gets \left\lceil \tilde{\mathbf{u}}+\mathsf{PRG}(\mathsf{seed})\right\rfloor_p.
\end{gather*}
 Output $\tilde{\mathbf{y}}$. 

\item[$\Dec(\hk, \mathbf{d}, \td)$:] Compute
\begin{gather*}
\mathbf{w}\gets \MOLE.\Enc\Eval(\mpk, \mathbf{d}, \phi), \qquad \qquad
\mathbf{u}\gets (\id_{t\cdot \ell\cdot k}\otimes \mathbf{G}^{-1}(\Delta))^\intercal \mathbf{w},\\
\mathbf{y}\gets \left\lceil -\mathbf{u}-\mathsf{PRG}(\mathsf{seed})\right\rfloor_p.
\end{gather*}
Then, output $\mathbf{y}$.
\end{description}
\end{protocol}

\paragraph{Adaptive Correctness.} In order for the construction to be fully correct, we need to choose a sufficiently large modulus $q$. Specifically, it must hold that 
\[\log q \cdot p\cdot \frac{n\cdot \lVert \mathbf{P}\rVert_\infty \cdot  T\cdot B \cdot (k\cdot \log q)^{d+2}+\alpha\cdot T\cdot (k\cdot \log q)^d+B+\hat{\alpha}}{q}=2^{-\omega(\log \sec)}.\] Notice that in the OTE of \secref{NIOTE}, it holds that $\lVert \mathbf{P}\rVert_\infty=n^{2r}$ where $r=O(\log m)$.\footnote{By a careful choice of parameters, $r$ can also be made a constant: E.g., by setting $t=\sec$ in Construction \ref{prot:SOTE}.}
We begin by observing that, by the correctness of Construction \ref{prot:R1BGG}, for every $i\in[d]$:
\[
\mathbf{c}_i^\intercal=\mathbf{s}_i^\intercal\mathbf{A}\mathbf{P}^\intercal+\mathbf{s}_{i+1}^\intercal(\hat{\mathbf{x}}^\intercal\otimes \mathbf{G})+\mathbf{e}_i^\intercal
\]
where $\lVert \mathbf{e}_i\rVert_\infty\le k\cdot n\cdot \lVert \mathbf{P}\rVert_\infty\cdot \log q \cdot B(\sec)+\alpha$.
Therefore, by Lemma \ref{evalRMS}, we have that:
\begin{align*}
\mathbf{c}^\intercal
&=\Eval\mathsf{RMSC}(\mathbf{A}\mathbf{P}^\intercal, f, \mathbf{x}, \{\mathbf{c}_i\}_{i\in[d]}) \\
&=\mathbf{s}_0^\intercal\mathbf{A}_f +\mathbf{s}_d^\intercal(f(\mathbf{x})^\intercal\otimes \mathbf{G})+\widetilde{\mathbf e}^\intercal.
\end{align*}
     where $\lVert \mathbf{\widetilde{e}}\rVert_\infty\le n\cdot \lVert \mathbf{P}\rVert_\infty\cdot B\cdot T\cdot (k\cdot \log q)^{d+1}+\alpha\cdot T\cdot (k\cdot \log q)^{d}$.
    We obtain that: 
\begin{align*}
\mathbf{z}
&=-(\id_\ell\otimes \mathbf{G}^{-1}(\mathbf{B}))^\intercal \mathbf{c}+f(\mathbf{x})\otimes \mathbf{b}\\
&=-(\id_\ell\otimes \mathbf{G}^{-1}(\mathbf{B}))^\intercal
\bigl(\mathbf{A}_f^\intercal \mathbf{s}_0+(f(\mathbf{x})\otimes \mathbf{G}^\intercal)\mathbf{s}_d+\widetilde{\mathbf e}\bigr)
+f(\mathbf{x})\otimes \mathbf{b}\\
&=\mathbf{F}^\intercal \mathbf{s}_0
-(\id_\ell\otimes \mathbf{G}^{-1}(\mathbf{B}))^\intercal (f(\mathbf{x})\otimes \mathbf{G}^\intercal)\mathbf{s}_d
-(\id_\ell\otimes \mathbf{G}^{-1}(\mathbf{B}))^\intercal \widetilde{\mathbf e}
+f(\mathbf{x})\otimes \mathbf{b}\\
&=\mathbf{F}^\intercal \mathbf{s}_0
-f(\mathbf{x})\otimes (\mathbf{B}^\intercal \mathbf{s}_d)
-(\id_\ell\otimes \mathbf{G}^{-1}(\mathbf{B}))^\intercal \widetilde{\mathbf e}
+f(\mathbf{x})\otimes
\bigl(\mathbf{B}^\intercal \mathbf{s}_d+(\id_t\otimes \mathbf{G})^\intercal\mathbf{a}+\mathbf e\bigr)\\
&=\mathbf{F}^\intercal \mathbf{s}_0
+f(\mathbf{x})\otimes\bigl((\id_t\otimes \mathbf{G})^\intercal\mathbf a\bigr)
-(\id_\ell\otimes \mathbf{G}^{-1}(\mathbf{B}))^\intercal \widetilde{\mathbf e}
+f(\mathbf{x})\otimes \mathbf e.
\end{align*}
    Finally, by the $\hat{\alpha}$-correctness of the MOLE, we have
\[
\mathbf{v}+\mathbf{w}=\mathbf{F}^\intercal\mathbf{s}_0+\widehat{\mathbf e}
\]
for some error vector $\widehat{\mathbf e}$ such that $\lVert \widehat{\mathbf e}\rVert_\infty\le \hat{\alpha}$. Therefore
\begin{align*}
\mathbf{z}-\mathbf{v}-\mathbf{w}
&=\mathbf{F}^\intercal \mathbf{s}_0
+f(\mathbf{x})\otimes\bigl((\id_t\otimes \mathbf{G})^\intercal\mathbf a\bigr)
-(\id_\ell\otimes \mathbf{G}^{-1}(\mathbf{B}))^\intercal \widetilde{\mathbf e}
+f(\mathbf{x})\otimes \mathbf e
-\mathbf{F}^\intercal\mathbf{s}_0-\widehat{\mathbf e}\\
&=f(\mathbf{x})\otimes\bigl((\id_t\otimes \mathbf{G})^\intercal\mathbf a\bigr)
-(\id_\ell\otimes \mathbf{G}^{-1}(\mathbf{B}))^\intercal \widetilde{\mathbf e}
+f(\mathbf{x})\otimes \mathbf e
-\widehat{\mathbf e}.
\end{align*}
Define
\[
\bm\varepsilon:=-(\id_\ell\otimes \mathbf{G}^{-1}(\mathbf{B}))^\intercal \widetilde{\mathbf e}
+f(\mathbf{x})\otimes \mathbf e
-\widehat{\mathbf e}.
\] It holds that:
    \[\lVert\bm{\varepsilon}\rVert_\infty\le n\cdot \lVert \mathbf{P}\rVert_\infty\cdot B\cdot T\cdot (k\cdot \log q)^{d+2}+\alpha\cdot T\cdot (k\cdot \log q)^{d}+B+\hat{\alpha}.\]
    We observe that:
    \begin{align*}
(\id_{t\cdot \ell\cdot k}\otimes \mathbf{G}^{-1}(\Delta))^\intercal(\mathbf{z}-\mathbf{v}-\mathbf{w})
&=\Delta\cdot (f(\mathbf{x})\otimes \mathbf a)
+(\id_{t\cdot \ell\cdot k}\otimes \mathbf{G}^{-1}(\Delta))^\intercal\bm\varepsilon.
\end{align*}
Notice also that
\[
\left\lVert(\id_{t\cdot \ell\cdot k}\otimes \mathbf{G}^{-1}(\Delta))^\intercal\bm\varepsilon\right\rVert_\infty
\le \log q\cdot \lVert\bm\varepsilon\rVert_\infty.
\]

    Finally, we observe that $\tilde{\mathbf{y}}+\mathbf{y}\ne f(\mathbf{x})\otimes \mathbf{a}\pmod p$ only if one of the entries of $\mathbf{u}-\mathsf{PRG}(\mathsf{seed})$ is less than $\log q\cdot \lVert \bm{\varepsilon}\rVert_\infty$ away from an odd multiple of $\Delta/2$. Since $\mathsf{seed}$ is independent of $\mathbf{u}$ and by the security of the PRG, the probability of this event is at most $p\cdot \log q\cdot \lVert\bm{\varepsilon}\rVert_\infty/q+\negl(\sec)$ {(this holds irrespectively of the inputs chosen by the adversary).} By hypothesis, this quantity is negligible.

\paragraph{Adaptive Privacy.} We now show that our construction satisfies adaptive privacy: under the LWE assumption with superpolynomial modulus-noise ratio, we can simulate $(C, \mathbf{b}, \{\mathbf{h}_i, E_i\}_{i\in[d]}, \mathsf{seed})$ without knowing any information about $\mathbf{x}$ and $\mathbf{a}$ even if these are chosen by the adversary after seeing the output of the setup. 
\begin{theorem}[Adaptive Privacy]
\thmlab{RMSLFE}
    Consider the following experiment $\mathsf{AdaptivePrivacy}_{\Adv, \Sim}(\bbbone^\sec)$ parametrized by an adversary $\Adv = (\Adv_0, \Adv_1)$ and a simulator $\Sim$:
\begin{itemize}
\item Sample $\hk\sample\Setup(\bbbone^\sec)$.
\item Activate the adversary $(\mathbf{x}, \mathbf{a}, \aux)\gets\Adv_0(\bbbone^\sec, \hk)$.
\item Sample $b\sample\{0, 1\}$. If $b = 0$, sample $\bm{r}_{i}\sample\Z_q^k$ for every $i\in[d]$. Then, compute 
\begin{align*}
    (\ek_0, \td_0)&\sample\Gen(\hk, \mathbf{x}, \mathbf{a}; \{\mathbf{r}_i\}_{i\in[d]})
\end{align*}
If $b=1$ compute:
\[
\ek_1\sample\Sim(\bbbone^\sec, \hk).
\]
\item Obtain $b' \gets \Adv_1(\ek_b, \aux)$ and return $1$ if and only if $b=b'$.
\end{itemize}
Then assuming the hardness of LWE with superpolynomial modulus-noise ratio and that $\MOLE$ and $\OTE$ are encoder private, there exists a PPT simulator $\Sim$, such that for every PPT adversary $\mathcal{A}$, there exists a negligible function $\negl(\sec)$ such that, for every $\sec\in\N$, we have that:
\[
\left|\frac{1}{2} - \Pr\left[\mathsf{AdaptivePrivacy}_{\mathcal{A}, \Sim}(\bbbone^\sec) = 1\right] \right|\le\negl(\sec)
\]
where the probability is taken over the random coins of the experiment.

\end{theorem}

\begin{proof}

We prove the theorem through a series of indistinguishable hybrids.
\begin{itemize}
\item Hybrid ${\mathcal{H}_0}$: This is the original experiment: We provide the adversary $\Adv_1$ with $(C, \mathbf{b}, \{\mathbf{h}_i, E_i\}_{i\in[d]}, \mathsf{seed})$ where $((C, \mathbf{b}, \{\mathbf{h}_i, E_i\}_{i\in[d]}, \mathsf{seed}), \td)\sample\Gen(\hk, \mathbf{x}, \mathbf{a}; \{\mathbf{r}_i\}_{i\in[d]})$.

\item Hybrid ${\mathcal{H}_1}$: In this hybrid, we generate $C$ using $\MOLE.\Sim(\bbbone^\sec, \mpk)$. The rest remains as in the previous hybrid.

Indistinguishability from Hybrid ${\mathcal{H}_0}$ follows from the encoder privacy of $\MOLE$.

\item Hybrid ${\mathcal{H}_{2}^i}$: In this hybrid, for every $j<i$, we generate $(\mathbf{h}_j, E_j)\sample\mathsf{C}\Sim(\bbbone^\sec, \ck, \mathbf{A})$, where $\mathsf{C}\Sim$ is the simulator of \thmref{expansion}. The rest remains as in the previous hybrid.

We observe that Hybrid ${\mathcal{H}_1}$ is identical to Hybrid ${\mathcal{H}_2^0}$. Moreover, for every $i\in[d]$, Hybrid ${\mathcal{H}_2^i}$ is computationally indistinguishable from Hybrid ${\mathcal{H}_2^{i+1}}$ due to \thmref{expansion} (notice that in Hybrid ${\mathcal{H}_2^i}$, the pair $(\mathbf{h}_{i-1}, E_{i-1})$ is independent of $\mathbf{s}_i$).

\item Hybrid ${\mathcal{H}_3}$: In this hybrid, we sample $\mathbf{b}\sample\Z_q^{t\cdot k\cdot \log q}$. The rest remains as in the previous hybrid.

Hybrid ${\mathcal{H}_2^{d}}$ is computationally indistinguishable from ${\mathcal{H}_3}$ under LWE. Notice that in Hybrid ${\mathcal{H}_2^{d}}$, the terms $C, \{\mathbf{h}_i, E_i\}_{i\in[d]}$ no longer contain information about $\mathbf{s}_d$. This allows us to reduce indistinguishability to LWE. Indeed, all information about $(\id_t\otimes \mathbf{G})^\intercal\mathbf{a}$ in $\mathbf{b}$ is masked by
\[ \mathbf{B}^\intercal \mathbf{s}_d+\mathbf e. \]
This can be viewed as an LWE sample with respect to the matrix ${\mathbf{B}}^\intercal$ and secret $\mathbf{s}_d$.
\end{itemize}
Notice that in Hybrid ${\mathcal{H}_3}$ all material provided to $\Adv_1$ is independent of $\mathbf{x}$ and $\mathbf{a}$.
\end{proof}

\subsection{Rate-1 Adaptive LFE for all Bounded-Depth Functions}\label{sec:mainLFE}

We now present our adaptive LFE scheme for bounded-depth functions. We described it in Construction \ref{prot:R1LFEforP}. Let $\ell:=\ell(\sec)$, $m:=m(\sec)$,  and $D:=D(\sec)$ be positive integers. The construction allows the evaluation of any function $f:\{0, 1\}^m\rightarrow \{0, 1\}^\ell$ described by a polynomial-sized circuit of depth at most $D(\sec)$ and width bounded by $W(\sec)=t^r\cdot n$, where $t$ is the parameter in the fully succinct OTE. The digest size is $\poly(\sec, \log m, \log \ell, \log D)$, whereas the size of the encodings is $m+\ell+D\cdot \poly(\sec, \log m, \log \ell, \log D)$.

The construction builds upon the reverse trapdoor hashing scheme $\RTDH:=(\Setup, \allowbreak\Hash, \allowbreak\Gen, \Enc, \Dec)$ of Construction \ref{prot:R1LFE}. {Notice that we can interpret the outputs of $\Enc$ and $\Dec$ as a secret-sharing over $\Z_p$ where $p=m^r\cdot 2^{\omega(\log \sec)}$, setting the modulus appropriately.} Moreover, we assume that $\RTDH$ is instantiated using the OTE protocols in Construction \ref{prot:SOTE} and \ref{prot:HCVOLE}. The  basic idea is the following: we pick a constant $d=O(1)$ and we rewrite $f$ as the composition of $L=D/\log d$ functions of depth at most $\log d$. Each of these can be regarded as a depth-$d$ RMS program, which can therefore be evaluated using $\RTDH$.

\begin{protocol}[Rate-1 Adaptive LFE for Bounded-Depth Circuits (to be continued)]%
\label{prot:R1LFEforP}

\begin{description}
\item[$\Setup(1^\secpar)$:] Output $\pk:=\hk\sample\RTDH.\Setup(\bbbone^\sec)$

\item[$\Hash(\pk, f)$:] For any vector $\mathbf{z}$, let $f_{\mathbf{z}}$ be the function that maps a key $K$ to $f(\mathbf{z}\oplus \mathsf{PRG}(K))$. Let $f'$ be the function that maps a pair $(\ct, \mathbf{z})$ to $\BV.\Eval(\ct, f_{\mathbf{z}})$ where the evaluation is performed as in \cite{ITCS:BraVai14}. 
Rewrite $f'$ as \[f'=f_{L-1}'\circ f_{L-2}'\circ\dots\circ f'_0\] where each $f'_i$ is a boolean function described by a $T$-bounded depth-$d$ RMS program, all with the same input and output size $W=t^r\cdot n$.
For every $i\in[L-1]$, let $\hat{f}_i$ be the function that maps $\mathbf{x}$ to $\OTE.\Hash(\tpk, f'_i(\mathbf{x}))$ and set $\hat{f}_{L-1}\gets f'_{L-1}$, where $\tpk$ denotes the OTE public key that is used in Construction \ref{prot:R1BGG}. Notice that $\OTE.\Hash$ is linear so its depth is $0$.\footnote{See Construction \ref{prot:HCVOLE} and Construction \ref{prot:SOTE}.} {Therefore, $\hat{f}_i$ is described by a $\max(T, m^{r})$-bounded depth-$d$ RMS program.\footnote{Since, we are hashing a binary string, the $\ell_\infty$-norm of the digest is $m^r$.}} Let $\hat{f}$ the function that maps any vector $\mathbf{x}$ to $(\hat{f}_0(\mathbf{x}), \dots, \hat{f}_{L-1}(\mathbf{x}))$. Output $(\mathbf{d}, \psi)\gets\RTDH.\Hash(\hk, \hat{f})$.

\item[$\Enc(\pk, \mathbf{d}, \mathbf{x})$:] Sample $\mathbf{M}\sample\Z_q^{(k-1)\times (k\cdot \log q+\sec)}$, $\mathbf{r}\sample\Z_q^{k-1}$ and $\mathbf{e}\sample\chi(\sec)$ and generate a $\BV$ key 
\[
\sk_\BV:=\begin{pmatrix}
\mathbf r\\
-1
\end{pmatrix},
\qquad
\pk_\BV:=\begin{pmatrix}
\mathbf M\\
(\mathbf M^\intercal \mathbf r+\mathbf e)^\intercal
\end{pmatrix}.
\]
Sample $K\sample\{0, 1\}^\sec$ and encrypt the input 
\[\mathbf{z}\gets \mathbf{x}\oplus\mathsf{PRG}(K),\qquad\qquad \ct\sample\BV.\Enc(\pk_\BV, K).\]
 For every $i\in[L], j\in[d]$ sample randomness $\mathbf{r}_{i, j}\sample\Z_q^k$ and define:
\[
\mathbf{a}_i:=\begin{pmatrix}
\mathbf{G}^{-1}(\mathbf{r}_{i,0})\\
\vdots\\
\mathbf{G}^{-1}(\mathbf{r}_{i,d-1})
\end{pmatrix},
\qquad
\mathbf{a}_L:=\mathbf{G}^{-1}(\sk_\BV).
\]
Then, set $\mathbf{x}_0:=(\mathsf{Bits}(\ct), \mathbf{z})$ and compute
\begin{align*}(\ek_0, \td_0)&\sample\RTDH.\Gen(\hk, \mathbf{x}_0, \mathbf{a}_{1}; \{\mathbf{r}_{0, j}\}_{j\in[d]})\\
  \forall i>0:\qquad(\ek'_i, \td_i)&\sample\RTDH.\Gen(\hk, \mathbf{0}, \mathbf{a}_{i+1}; \{\mathbf{r}_{i, j}\}_{j\in[d]})
\end{align*}
\end{description}
\end{protocol}

\addtocounter{algocf}{-1}
\begin{protocol}[(cont.)]
  \begin{description}
  \item
    For every $i\in[L]$, let $\ek'_i=(C_i, \mathbf{b}_i, \{\mathbf{h}'_{i, j}, E_{i, j}\}_{j\in[d]}, \mathsf{seed}_i)$. Compute 
\[\mathbf{y}_i\gets -\RTDH.\Dec(\hk, \mathbf{d}, \td_i)\bmod p\]
and set $\overline{\ek}_i:= (C_i, \mathbf{b}_i, \{E_{i, j}\}_{j\in[d]}, \mathsf{seed}_i)$.
 Let $n$ be the digest size of $\OTE$. For every $i\in[L-1]$, take the block of $\mathbf{y}_{i}$ corresponding to the evaluation of $\hat{f}_i$ and split it into $d\cdot n$ subblocks $\hat{\mathbf{y}}_{i, 0}, \dots, \hat{\mathbf{y}}_{i, n\cdot d-1}$ of dimension $k\cdot \log q$ and set
\begin{align*}
\mathbf{y}_{i, j}^\intercal&\gets \bigl(\hat{\mathbf{y}}_{i, j}^\intercal \mathbin{\|} \hat{\mathbf{y}}_{i, d+j}^\intercal \mathbin{\|}\dots \mathbin{\|} \hat{\mathbf{y}}_{i, d\cdot (n-1)+j}^\intercal\bigr)\\
  \mathbf{w}_{i, j}&\gets \mathbf{h}'_{i+1, j}- (\id_{n\cdot k}\otimes \mathbf{g}_q^\intercal\otimes \id_{\log q})\,(\mathbf{y}_{i,j}\otimes \mathbf{g}_q)\pmod q.
\end{align*}                     
Let $\mathbf{u}$ be the last element of the standard basis of $\Z_q^k$ and let $\mathbf{y}^\ast_{L-1}$ be the block of ${\mathbf{y}}_{L-1}$ corresponding to the output of $\hat{f}_{L-1}$. Compute \[ \mathbf{w}_{L-1}\gets -\mathbf{G}^{-1}(-\Delta\cdot (\id_\ell\otimes \mathbf u))^\intercal \mathsf{Lin}(\id_{\ell\cdot k^2\cdot \log q}) (\id_{\ell\cdot k^2\cdot \log q}\otimes \mathbf g_q^\intercal\otimes \id_k\otimes \mathbf g_q^\intercal) \,\mathbf{y}^\ast_{L-1}\pmod q. \]
Sample $\mathsf{seed}\sample\{0, 1\}^\sec$ and set $\mathbf{w}\gets \lceil\mathbf{w}_{L-1}+\mathsf{PRG}(\mathsf{seed})\rfloor_2$, parse $\mathsf{ek}_0=(C_0,\mathbf b_0,\{\mathbf h_{0,j},E_{0,j}\}_{j\in[d]},\mathsf{seed}_0)$ and output \[E:=(\ct, \mathbf{z}, \mathsf{seed}, \ek_0, \{\overline{\ek}_i\}_{i\in[L]}, \{\mathbf{w}_{i, j}\}_{i\in[L-1], j\in[d]}, \mathbf{w}).\]

\item[$\Dec(\pk, E, f, \psi)$:]
Initially, set $\mathbf{x}_0\gets (\mathsf{Bits}(\ct), \mathbf{z})$.
Then, for $i=0, \dots, L-2$: 
\begin{itemize}
\item compute $\tilde{\mathbf{y}}_i\gets\RTDH.\Enc(\hk, \ek_i, \psi, \mathbf{x}_i, \hat{f})$
\item take the block of $\tilde{\mathbf{y}}_{i}$ corresponding to the evaluation of $\hat{f}_i$ and split it into $d\cdot n$ subblocks $\overline{\mathbf{y}}_{i, 0}, \dots, \overline{\mathbf{y}}_{i, n\cdot d-1}$ of dimension $k\cdot \log q$ and set
\begin{align*}
\tilde{\mathbf{y}}_{i, j}^\intercal&\gets \bigl(\overline{\mathbf{y}}_{i, j}^\intercal \mathbin{\|} \overline{\mathbf{y}}_{i, d+j}^\intercal \mathbin{\|}\dots \mathbin{\|} \overline{\mathbf{y}}_{i, d\cdot (n-1)+j}^\intercal\bigr)\\
\mathbf{h}_{i+1,j}&\gets \mathbf{w}_{i,j} +(\id_{n\cdot k}\otimes \mathbf{g}_q^\intercal\otimes \id_{\log q})\,(\tilde{\mathbf{y}}_{i,j}\otimes \mathbf{g}_q)\pmod q.\end{align*}
\item set $\ek_{i+1}\gets(\overline{\ek}_{i+1}, \{\mathbf{h}_{i+1, j}\}_{j\in[d]})$
\item compute the input to the next RMS program $\mathbf{x}_{i+1}\gets f'_i(\mathbf{x}_i)$
\end{itemize}
Finally, let $\mathbf{u}$ be the last element of the standard basis of $\Z_q^k$. Derive
\begin{align*}\tilde{\mathbf{y}}_{L-1}&\gets\RTDH.\Enc(\hk, \ek_{L-1}, \psi, \mathbf{x}_{L-1}, \hat{f})\\
\tilde{\mathbf{w}}_{L-1}&\gets \mathbf{G}^{-1}(-\Delta\cdot (\id_\ell\otimes \mathbf u))^\intercal \mathsf{Lin}(\id_{\ell\cdot k^2\cdot \log q}) (\id_{\ell\cdot k^2\cdot \log q}\otimes \mathbf g_q^\intercal\otimes \id_k\otimes \mathbf g_q^\intercal) \,\overline{\mathbf y}_{L-1}^\ast\pmod q.
\end{align*}
where $\overline{\mathbf{y}}_{L-1}^\ast$ denotes the block of $\tilde{\mathbf{y}}_{L-1}$ corresponding to the evaluation of $\hat{f}_{L-1}$.

 Output $\lceil\tilde{\mathbf{w}}_{L-1}-\mathsf{PRG}(\mathsf{seed})\rfloor_2\oplus \mathbf{w}$
\end{description}
\end{protocol}

    \paragraph{Adaptive Correctness.} Towards proving correctness, we prove the following lemma, which states that the term $\mathbf{h}_{i, j}$ computed by the server during the decoding procedure is an adaptive lattice encoding of $\OTE.\Hash(\tpk, \mathbf{x}_i)$. In other words, $\ek_i$ is an encoding key for $\mathbf{x}_i=(f_{i-1}'\circ\dots\circ f'_0)(\mathsf{Bits}(\ct), \mathbf{z})$.
\begin{lemma}
\label{claimboh}
For every $i\in[L]$ and $j\in[d]$, we have that
\[ \mathbf{h}_{i,j}=\mathbf{A}^\intercal \mathbf{s}_{i,j}+(\mathbf{d}'_i\otimes \mathbf{G}^\intercal)\mathbf{r}_{i,j}+\mathbf e_i. \]
where $(\mathbf{d}_i', \psi'_i)\gets\OTE.\Hash(\tpk, \mathbf{x}_i)$, $\mathbf{s}_{i, j}$ is the encryption key used in the encoding $\mathbf{h}'_{i, j}$ $(\mathbf{h}_{0, j}$, if $i=0)$ and $\lVert \mathbf{e}_i\rVert_\infty\le B(\sec)$. {This holds with overwhelming probability, irrespectively of the inputs chosen by the adversary.}
\end{lemma}
\begin{proof}
We proceed by induction over $i$. The claim is trivially true for $i=0$. Now, we show that if it holds for $i$, it holds also for $i+1$.

We observe that
\begin{align*}
\mathbf{h}_{i+1,j}
&=\mathbf{w}_{i,j}+(\id_{n\cdot k}\otimes \mathbf g_q^\intercal\otimes \id_{\log q})(\tilde{\mathbf{y}}_{i,j}\otimes \mathbf g_q)\pmod q\\
&=\mathbf{h}'_{i+1,j}
-(\id_{n\cdot k}\otimes \mathbf g_q^\intercal\otimes \id_{\log q})(\mathbf y_{i,j}\otimes \mathbf g_q)
+(\id_{n\cdot k}\otimes \mathbf g_q^\intercal\otimes \id_{\log q})(\tilde{\mathbf{y}}_{i,j}\otimes \mathbf g_q)\pmod q\\
&=\mathbf{h}'_{i+1,j}
+(\id_{n\cdot k}\otimes \mathbf g_q^\intercal\otimes \id_{\log q})\bigl((\tilde{\mathbf{y}}_{i,j}-\mathbf y_{i,j})\otimes \mathbf g_q\bigr)\pmod q.
\end{align*}
Now, by the correctness of $\RTDH$ and the inductive hypothesis, we recall that $ \tilde{\mathbf{y}}_i-\mathbf{y}_i=\hat{f}(\mathbf{x}_i)\otimes \mathbf{a}_{i+1}\pmod p$. {This holds with overwhelming probability, irrespectively of the inputs chosen by the adversary.} We also recall that this subtractive secret-sharing was initially over $\Z_q$, i.e. the parties held vectors $\tilde{\mathbf{u}}_i$ and $\mathbf{u}_i$ such that $\tilde{\mathbf{u}}_i-\mathbf{u}_i=\frac{q}{p}\cdot\bigl(\hat{f}(\mathbf{x}_i)\otimes \mathbf{a}_{i+1}\bigr)+\bm{\varepsilon}_i$ where $\lVert \bm{\varepsilon}_i\rVert_\infty$ is low. This secret-sharing was rerandomised using $\mathsf{PRG}(\mathsf{seed}_i)$ and then rounded. We argue that $\mathbf{y}_i$ is pseudorandom. {In particular, since $p=2^{\omega(\log \sec)}$, with overwhelming probability over the choice of $\mathsf{seed}_i$, all entries of $\mathbf{y}_i$ are in the interval $[-p/2, p/2-m^r)$. Since $\hat{f}(\mathbf{x}_i)\otimes \mathbf{a}_{i+1}$ has all $\ell_\infty$-norm at most $m^{r}$, we conclude that (with overwhelming probability over the choice of $\mathsf{seed}_i$)} $\tilde{\mathbf{y}}_i$ and $\mathbf{y}_i$ are a subtractive secret-sharing of $\hat{f}(\mathbf{x}_i)\otimes \mathbf{a}_{i+1}$ even over the integers. Continuing, we observe that the $i$-th $n$-entry block of $\hat{f}(\mathbf{x}_i)$ coincides with $\mathbf{d}'_{i+1}$ where $\mathbf{d}'_{i+1}$ is the output of $\OTE.\Hash(\tpk, f'_i(\mathbf{x}_i))=\OTE.\Hash(\tpk, \mathbf{x}_{i+1})$. Moreover, we observe that, by the way we constructed $\mathbf{a}_{i+1}$, $\tilde{\mathbf{y}}_{i, j}$ and $\mathbf{y}_{i, j}$, we have that 
\[ \tilde{\mathbf{y}}_{i,j}-\mathbf{y}_{i,j} =\mathbf{d}'_{i+1}\otimes \mathbf{G}^{-1}(\mathbf{r}_{i+1,j}). \]
Putting everything together, we obtain
\begin{align}
\mathbf{h}_{i+1,j}
&=\mathbf{h}'_{i+1,j}
+(\id_{n\cdot k}\otimes \mathbf g_q^\intercal\otimes \id_{\log q})\bigl((\tilde{\mathbf{y}}_{i,j}-\mathbf y_{i,j})\otimes \mathbf g_q\bigr)\nonumber\\
&=\mathbf{h}'_{i+1,j}
+(\id_{n\cdot k}\otimes \mathbf g_q^\intercal\otimes \id_{\log q})
\bigl((\mathbf d'_{i+1}\otimes \mathbf G^{-1}(\mathbf r_{i+1,j}))\otimes \mathbf g_q\bigr)\nonumber\\
&=\mathbf{h}'_{i+1,j}+(\mathbf d'_{i+1}\otimes \mathbf G^\intercal)\mathbf r_{i+1,j}.\label{eqclaimboh}
\end{align}
Since $\mathbf{h}'_{i+1,j}=\mathbf{A}^\intercal \mathbf{s}_{i+1,j}+\mathbf e_{i+1}$ where $\lVert \mathbf{e}_{i+1}\rVert_\infty\le B(\sec)$, we obtain that
\[ \mathbf{h}_{i+1,j}=\mathbf{A}^\intercal \mathbf{s}_{i+1,j}+(\mathbf{d}'_{i+1}\otimes \mathbf{G}^\intercal)\mathbf{r}_{i+1,j}+\mathbf e_{i+1}. \]
This ends the proof of the lemma.
\end{proof}

Continuing with our analysis, by Lemma \ref{claimboh}, we have that $\tilde{\mathbf{y}}_{L-1}$ and $\mathbf{y}_{L-1}$ form a subtractive secret sharing of $\hat{f}(\mathbf{x}_{L-1})\otimes \mathbf{a}_{L}$ over $\Z_p$. {This holds with overwhelming probability, irrespectively of the inputs chosen by the adversary.} Furthermore, similarly to how we argued in the proof of Lemma \ref{claimboh}, it is possible to prove that the probability that any entry of $\mathbf{y}_{L-1}$ lies outside of $[-p/2, p/2-1)$ is negligible {(irrespectively of the inputs chosen by the adversary)}. Given that all entries of $\hat{f}(\mathbf{x}_{L-1})\otimes \mathbf{a}_{L}$ belong to $\{0, 1\}$, we conclude that, with overwhelming probability, $\tilde{\mathbf{y}}_{L-1}$ and $\mathbf{y}_{L-1}$ form a subtractive secret sharing of $\hat{f}(\mathbf{x}_{L-1})\otimes \mathbf{a}_{L}$ even over $\Z$. By the way we constructed $\hat{f}, \overline{\mathbf{y}}^\ast_{L-1}$ and $\mathbf{y}^\ast_{L-1}$, we infer that \[
\overline{\mathbf{y}}_{L-1}^\ast-\mathbf{y}^\ast_{L-1}
=\hat{f}_{L-1}(\mathbf{x}_{L-1})\otimes \mathbf{a}_{L}.
\]
We also notice that $\mathbf{a}_L=\mathbf{G}^{-1}(\sk_\BV)$, whereas
\[\hat{f}_{L-1}(\mathbf{x}_{L-1})=f'_{L-1}(\mathbf{x}_{L-1})=(f'_{L-1}\circ\dots \circ f'_0)(\mathbf{x}_0)=\mathsf{Bits}(f'(\ct, \mathbf{z}))=\mathsf{Bits}(\mathsf{vec}(\BV.\Eval(\ct, f_{\mathbf{z}}))).\]
Let $\hat{\ct}:=\BV.\Eval(\ct, f_{\mathbf{z}})$. We obtain that
\begin{align*}
\tilde{\mathbf{w}}_{L-1}+\mathbf{w}_{L-1}
&=\mathsf{G}^{-1}(-\Delta\cdot (\id_\ell\otimes \mathbf u))^\intercal
\mathsf{Lin}(\id_{\ell\cdot k^2\cdot \log q})
(\id_{\ell\cdot k^2\cdot \log q}\otimes \mathbf g_q^\intercal\otimes \id_k\otimes \mathbf g_q^\intercal)
\bigl(\overline{\mathbf y}^\ast_{L-1}-\mathbf{y}^\ast_{L-1}\bigr)\\
&=\mathsf{G}^{-1}(-\Delta\cdot (\id_\ell\otimes \mathbf u))^\intercal
\mathsf{Lin}(\id_{\ell\cdot k^2\cdot \log q})
(\id_{\ell\cdot k^2\cdot \log q}\otimes \mathbf g_q^\intercal\otimes \id_k\otimes \mathbf g_q^\intercal)
\bigl(\hat f_{L-1}(\mathbf x_{L-1})\otimes \mathbf a_L\bigr)\\
&=\mathsf{G}^{-1}(-\Delta\cdot (\id_\ell\otimes \mathbf u))^\intercal
(\id_{\ell\cdot k\cdot \log q}\otimes \sk_\BV^\intercal)\,\mathsf{vec}(\hat{\ct}).
\end{align*}
By the correctness of $\BV$ evaluation, we have that
\[
(\id_{\ell\cdot k\cdot \log q}\otimes \sk_\BV^\intercal)\,\mathsf{vec}(\hat{\ct})
=\Delta\cdot f(\mathbf{x})+\widehat{\mathbf e}
\]
for some error vector $\widehat{\mathbf e}$ such that $\lVert\widehat{\mathbf e}\rVert_\infty\le B\cdot D\cdot \poly(\sec)$. Notice that $f_{\mathbf{z}}(K)=f(\mathbf{z}\oplus \mathsf{PRG}(K))=f(\mathbf{x})$.
We conclude that
\begin{align}
\mathbf{w}'_{L-1}+\mathbf{w}_{L-1}
&=\Delta\cdot f(\mathbf{x})+\widehat{\mathbf e}.\label{eq45}
\end{align}
So, unless any of the entries of $\mathbf{w}_{L-1}+\mathsf{PRG}(\mathsf{seed})$ is less than $\lVert\widehat{\mathbf e}\rVert_\infty$ away from $q/4$ or $-q/4$, we have that \[f(\mathbf{x})=\lceil\mathbf{w}_{L-1}'-\mathsf{PRG}(\mathsf{seed})\rfloor_2\oplus \lceil\mathbf{w}_{L-1}+\mathsf{PRG}(\mathsf{seed})\rfloor_2.\] 
Since $\mathsf{seed}$ is sampled independently of $\mathbf{w}_{L-1}$, the probability of the bad event is at most $\poly(\sec)\cdot \lVert\widehat{\mathbf e}\rVert_\infty/q)+\negl(\sec)$. {This holds irrespectively of the inputs chosen by the adversary.} Given our choice of $q$ for $\RTDH$, this is a negligible amount. 

\paragraph{Security.} Next, we prove that our construction is encoder private.

\begin{theorem}
Assume the hardness of LWE with superpolynomial modulus-noise ratio. Then, Construction \ref{prot:R1LFEforP} is an adaptively encoder private LFE. 
\end{theorem}

\begin{proof}
We prove our claim by relying on a series of indistinguishable hybrids.

\begin{itemize}
\item Hybrid ${\mathcal{H}_0}$: This hybrid corresponds to the original game: we provide the adversary with a tuple $(\ct, \mathbf{z}, \mathsf{seed}, \{\mathbf{h}_{0, j}\}_{j\in[d]}, \{\overline{\ek}_i\}_{i\in[L]}, \{\mathbf{w}_{i, j}\}_{i\in[L-1], j\in[d]}, \mathbf{w})$ generated using $\LFE.\Enc(\pk, \mathbf d, \mathbf{x})$. 

\item Hybrid ${\mathcal{H}_1}$: In this hybrid, we change the distribution of $\mathbf{w}$: using $\pk, f$, $(\ct, \allowbreak\mathbf{z}, \allowbreak\mathsf{seed}, \allowbreak\{\mathbf{h}_{0, j}\}_{j\in[d]}, \allowbreak\{\overline{\ek}_i\}_{i\in[L]}, \allowbreak\{\mathbf{w}_{i, j}\}_{i\in[L-1], j\in[d]})$ and following the same operations as in the decoding procedure, we compute $\mathbf{w}'_{L-1}$. Then, we set $\mathbf{w}\gets f(\mathbf{x})\oplus \lceil\mathbf{w}'_{L-1}-\mathsf{PRG}(\mathsf{seed})\rfloor_2$.  

This hybrid is statistically indistinguishable from Hybrid ${\mathcal{H}_0}$ due to the correctness of the primitive.

\item Hybrid ${\mathcal{H}_2}$: In this hybrid, we change the distribution of $\mathbf{w}_{i, j}$ for every $i\in[L-1]$ and $j\in[d]$: using $\pk, f$, $(\ct, \mathbf{z}, \{\mathbf{h_{0, \beta}}\}_{\beta\in[d]}, \{\overline{\ek}_\gamma\}_{\gamma\in[i]}, \{\mathbf{w}_{\gamma, \beta}\}_{\gamma\in[i], \beta\in[d]})$ and following the same operations as in the decoding procedure, we compute $\mathbf{y}_{i, j}'$. Then, we set 
\begin{equation}\mathbf{w}_{i,j}\gets \mathbf{h}'_{i+1,j}+(\mathbf d'_{i+1}\otimes \mathbf G^\intercal)\mathbf r_{i+1,j} -(\id_{n\cdot k}\otimes \mathbf g_q^\intercal\otimes \id_{\log q})(\tilde{\mathbf{y}}_{i,j}\otimes \mathbf g_q)\label{eqhybrid2} \end{equation}
where $(\mathbf{d}_{i+1}', \psi_{i+1}')\gets \OTE.\Hash(\tpk, \mathbf{x}_{i+1})$ and $\mathbf{x}_{i+1}\gets (f'_i\circ \dots\circ f'_0)(\mathsf{Bits}(\ct), \mathbf{z})$. 

This hybrid is statistically indistinguishable from Hybrid ${\mathcal{H}_1}$. Indeed, $\mathbf{w}_{i, j}$ satisfies the relation in \eqref{eqhybrid2} even in ${\mathcal{H}_1}$. This is highlighted in Lemma \ref{claimboh}, specifically in \eqref{eqclaimboh}.

\item Hybrid ${\mathcal{H}_3}$: In this hybrid, for every $i\in[L]\setminus\{0\}$, we compute
\[(\ek_i, \td_i)\sample\RTDH.\Gen(\hk, \mathbf{x}_{i+1}, \mathbf{a}_{i+1}; \{\mathbf{r}_{i, j}\}_{j\in[d]})\]
where $\mathbf{x}_{i+1}\gets (f'_i\circ \dots\circ f'_0)(\mathsf{Bits}(\ct), \mathbf{z})$. Let $\ek_i=(C_i, \mathbf{b}_i, \{\mathbf{h}_{i, j}, E_{i, j}\}_{j\in[d]}, \mathsf{seed}_i)$. For every $i\in[L-1]$ and $j\in[d]$, we set
\[ \mathbf{w}_{i,j}\gets \mathbf{h}_{i+1,j} -(\id_{n\cdot k}\otimes \mathbf g_q^\intercal\otimes \id_{\log q})(\tilde{\mathbf{y}}_{i,j}\otimes \mathbf g_q). \]
We also argue that ${\mathcal{H}_2}$ is perfectly indistinguishable from Hybrid ${\mathcal{H}_3}$. Indeed, the only difference between the output of $\RTDH.\Gen(\hk, \allowbreak\mathbf{x}_{i+1}, \allowbreak\mathbf{a}_{i+1}; \allowbreak\{\mathbf{r}_{i, j}\}_{j\in[d]})$ and $\RTDH.\Gen(\hk, \allowbreak\mathbf{0}, \allowbreak\mathbf{a}_{i+1}; \allowbreak\{\mathbf{r}_{i, j}\}_{j\in[d]})$
is that, in the first case, the encoding $\mathbf{h}_{i, j}$ is ``shifted'' by $(\mathbf d'_{i+1}\otimes \mathbf G^\intercal)\mathbf r_{i+1,j}$ where $(\mathbf{d}_{i+1}', \psi_{i+1}')\gets \OTE.\Hash(\tpk, \mathbf{x}_{i+1})$ (see Construction \ref{prot:R1BGG}). In the second case, no shift is essentially applied. This is because, by the linearity of our OTE hashing (see Construction \ref{prot:HCVOLE} and Construction \ref{prot:SOTE}), $\OTE.\Hash(\tpk, \mathbf{0})$ outputs $\mathbf{0}$.

\item Hybrid ${\mathcal{H}_4^\iota}$: In this hybrid, for every $i<\iota$, we generate
\[\ek_i:=(C_i, \mathbf{b}_i, \{\mathbf{h}_{i, j}, E_{i, j}\}_{j\in[d]}, \mathsf{seed}_i)\sample\RTDH.\Sim(\bbbone^\sec, \hk)\]

We observe that Hybrid ${\mathcal{H}_3}$ is identical to ${\mathcal{H}_4^0}$. Moreover, by \thmref{RMSLFE}, for every $\iota\in[L-1]$, we have that ${\mathcal{H}_4^\iota}$ is computationally indistinguishable from ${\mathcal{H}_4^{\iota+1}}$. Notice that here we are implicitly relying on the fact that, in Hybrid ${\mathcal{H}_4^\iota}$, the tuple
$\{\ek_i\}_{i<\iota}$ no longer contains information about $\{\mathbf{r}_{\iota, j}\}_{j\in[d]}$.

\item Hybrid ${\mathcal{H}_5}$: In this hybrid, we sample $\pk_\BV\sample\Z_q^{k\times (k\cdot \log q+\sec)}$.

We argue that ${\mathcal{H}_4^L}$ is computationally indistinguishable from ${\mathcal{H}_5}$. Indeed, under LWE, the term $\mathbf{M}^\intercal\mathbf r+\mathbf e$ is indistinguishable from random. Here, we are implicitly relying on the fact that, in ${\mathcal{H}_4^L}$, the tuple $\{\ek_i\}_{i\in[L]}$ no longer contains information about $\mathbf{r}$.

\item Hybrid ${\mathcal{H}_6}$: In this hybrid, we generate $\ct\sample\Z_q^{k\times (\sec\cdot k\cdot \log q)}$.

Hybrid ${\mathcal{H}_6}$ is statistically indistinguishable from ${\mathcal{H}_5}$. Indeed, now $\pk_\BV$ is a uniformly random matrix, so we can apply the leftover hash lemma to argue that $\ct$ is indistinguishable from random.

\item Hybrid ${\mathcal{H}_7}$: In this hybrid, we sample $\mathbf{z}\sample\Z_2^{m}$.

Since $\ct$ no longer contains information about the PRG seed $K$, we can argue that ${\mathcal{H}_6}$ is computationally indistinguishable from ${\mathcal{H}_7}$ thanks to the security of the PRG.
\end{itemize}

 Notice that in $\mathbf{\mathcal{H}_7}$ all the material provided to the adversary can be computed by a simulator with no information about $\mathbf{x}$ except $f(\mathbf{x})$. This ends the proof of security.
\end{proof}

{\paragraph{Parameters.} Construction~\ref{prot:R1LFEforP}
uses a single input-succinct reverse-TDH hash, and therefore its public key and digest
size are
\[
\poly(\lambda,\log m,\log \ell,\log D).
\]
Moreover, the encoding consists of the masked input, the final output share,
and $L=\lceil D/\log d\rceil$ auxiliary encodings. Hence the total encoding
size is
\[
m+\ell+D\cdot \poly(\lambda,\log m,\log \ell,\log D).
\]
Thus, we obtain the following theorem.}

{
\begin{theorem}[Rate-1 Adaptively Secure LFE]
Assume the hardness of LWE, there exists an adaptively secure LFE for
depth-$D$ functions
$f:\{0,1\}^m\to\{0,1\}^\ell$
with public key and digest size $\poly(\lambda,\log m,\log \ell,\log D)$
and encoding size
$m+\ell+D\cdot \poly(\lambda,\log m,\log \ell,\log D)$.
\end{theorem}}
\section{A Counterexample for Adaptive LWE}\label{sec:counter}

We present a counterexample to a conjecture from \cite{FOCS:QuaWeeWic18}. The conjecture (adaptive LWE) says that the probability of that any polynomial-time attacker wins the following experiment is negligibly close to $1/2$.
\begin{itemize}
    \item The challenger samples random matrices $\{\mathbf{A}_i \sample \Z_q^{n\times m} \}_{i\in[k]}$ and sends them to the attacker.
    \item The attacker chooses an $x\in\{0,1\}^{k-1}$ and sends it to the challenger. Let $\hat{x} := (x, 0)$.
    \item The challenger samples an $\mathbf{s}\sample\Z_q^n$ and a bit $b\sample\{0,1\}$.
    \item If $b =0$ it computes:
    \[
    \{\mathbf{b}_i^\intercal := \mathbf{s}^\intercal(\mathbf{A}_i - \hat{x}_i\cdot \mathbf{G}) + \mathbf{e}_i^\intercal\}_{i\in[k]}
    \]
    where $\mathbf{e}_i \sample \chi(\lambda)$.
    \item If $b =1$ it samples $\{\mathbf{b}_i \sample \Z_q^m\}_{i\in[k]}$.
    \item The attacker wins if, given $\{\mathbf{b}_i\}_{i\in[k]}$, it correctly guesses $b$.
\end{itemize}
We show that, if $k$ is allowed to grow independently of the LWE parameters $(n,q,\chi)$, the assumption is false. This roughly corresponds to the \emph{optimistic parameter} settings proposed in \cite{FOCS:QuaWeeWic18}. We sketch the attack in the following. We assume for convenience that $q$ is a power of $2$, but the attack can be adapted to other moduli.

Let $\mathbf{A}:=(\mathbf{A}_0 \mathbin{\|}\dots\mathbin{\|}\mathbf{A}_{k-1})$. By the homomorphic properties of the lattice encodings, we know that there exists a matrix $\mathbf{H}$ with $\|\mathbf{H}\|_\infty =1$ such that, for any matrix $\mathbf{M}\in\Z_q^{n\times n}$, if $\mathbf{x}:=\mathsf{Bits}(\mathbf{M})$, we have that:
\[
(\mathbf{s}^\intercal(\mathbf{A} - \mathbf{x}\otimes  \mathbf{G}) + \mathbf{e}^\intercal)\cdot \mathbf{H}= \mathbf{s}^\intercal  (\mathbf{A}\cdot \mathbf{H}-\mathbf{M}) + \mathbf{e}^\intercal \cdot \mathbf{H}
\approx \mathbf{s}^\intercal (\mathbf{A}\cdot \mathbf{H}-\mathbf{M})
\]
since $\|\mathbf{e}^\intercal  \mathbf{H}\|_\infty \approx 0$.  
To recover the $i$-th most significant bit of (each component of) $\mathbf{s}$, it is sufficient to set $\mathbf{M}_i:=\mathbf{A}\cdot \mathbf{H}- 2^i\cdot \id_n$ and $\mathbf{x}_i := \mathsf{Bits}(\mathbf{M}_i)$. Then compute the (component-wise) most significant bit of:
\[
(\mathbf{s}^\intercal(\mathbf{A} - \mathbf{x}_i\otimes  \mathbf{G}) + \mathbf{e}^\intercal)\cdot \mathbf{H} \approx \mathbf{s}^\intercal (\mathbf{A}\cdot \mathbf{H}-\mathbf{M}_i) = 2^i\cdot\mathbf{s}^\intercal.
\]
We can then recover the entire secret key $\mathbf{s}$, by setting the input $\mathbf{x}$ to be the concatenation of $(\mathbf{x}_1, \dots, \mathbf{x}_{\log q})$. This shows that the conjecture is false if $k$ is allowed to be arbitrarily bigger than $n$. Notice in the provable parameter setting (where the encodings are secure under the \emph{subexponential} hardness of LWE), our attack fails as the bound on $k$ is too small to encode $\mathbf{M}_i:=\mathbf{A}\cdot \mathbf{H}- 2^i\cdot \id_n$.

\newpage
   \printbibliography
 
\end{document}